\documentclass[a4paper,11pt,reqno]{amsart}

\usepackage[margin=1in,includehead,includefoot]{geometry}
\usepackage[utf8]{inputenc}
\usepackage[T1]{fontenc}
\usepackage{lmodern,microtype}
\usepackage{upref}
\usepackage{mathtools,amssymb,amsthm}
\usepackage[foot]{amsaddr}
\usepackage{enumitem}
\usepackage{xifthen}
\usepackage[draft]{todonotes}
\usepackage{tikz}
\usepackage{graphicx}
\usepackage{wrapfig}
\usepackage{hyperref}
\usepackage[b]{esvect} 
\usepackage{textcomp}
\usepackage{mdframed}
\usepackage{xpatch}
\usepackage{xcolor}
\usepackage{bookmark}
\usepackage{multirow}
\usepackage{transparent}
\usepackage{mycommands}

\graphicspath{{./graphics/}}

\makeatletter
\let\old@setaddresses\@setaddresses
\def\@setaddresses{\bigskip{\parindent 0pt\let\scshape\relax\let\ttfamily\relax\old@setaddresses}}
\makeatother

\newcommand{\tikzstairs}[1]{}
\newcommand{\tikzsubset}[1]{}
\newcommand{\tikzpartition}[1]{}
\newcommand{\tikztriang}[1]{}
\newcommand{\tikztree}[1]{}
\newcommand{\tikzdyck}[1]{}
\newcommand{\tikzrect}[1]{}

\newcommand{\mesh}[2] 
{
  \foreach \x in {1,...,#1} \draw[mesh line] (\x,0)--(\x,#1+1) (0,\x)--(#1+1,\x);
  \foreach \y [count=\x] in {#2} \node[mesh node] (\x) at (\x,\y) {};
}

\tikzstyle{stair values}=[thick]
\tikzstyle{subset values}=[thick]
\tikzstyle{subset container}=[thick, rounded corners]
\tikzstyle{partition values}=[draw, shape=circle, thin]
\tikzstyle{partition container}=[thick, rounded corners, opacity=0.5]
\tikzstyle{tree edge}=[thick]
\tikzstyle{tree vertex}=[circle,minimum size=8pt,inner sep=0pt]
\tikzstyle{triang vertex}=[circle,fill=black,minimum size=2.5pt,inner sep=0pt]
\tikzstyle{mesh line}=[thin]
\tikzstyle{mesh node}=[circle,minimum size=4pt,inner sep=0pt,draw=black,fill=black]
\tikzstyle{marked mesh node}=[circle,minimum size=4pt,inner sep=3pt,draw=black,fill=none]
\tikzstyle{poset node}=[rounded corners=2pt,minimum size=6pt,inner sep=2pt,draw=black,fill=white]
\tikzstyle{grid edge}=[draw=gray, line width=1.5pt]
\tikzstyle{slope edge}=[draw=black, thick]

\definecolor{p1col1}{rgb}{0.266122, 0.486664, 0.802529}
\definecolor{p1col2}{rgb}{0.513417, 0.72992, 0.440682}
\definecolor{p1col3}{rgb}{0.863512, 0.670771, 0.236564}
\definecolor{p1col4}{rgb}{0.857359, 0.131106, 0.132128}

\definecolor{p2col1}{rgb}{0.264135, 0.201429, 0.745889}
\definecolor{p2col2}{rgb}{0.256326, 0.430921, 0.808553}
\definecolor{p2col3}{rgb}{0.324106, 0.60897, 0.708341}
\definecolor{p2col4}{rgb}{0.439128, 0.704968, 0.52925}
\definecolor{p2col5}{rgb}{0.597181, 0.742185, 0.36771}
\definecolor{p2col6}{rgb}{0.764712, 0.728302, 0.273608}
\definecolor{p2col7}{rgb}{0.88018, 0.631684, 0.227665}
\definecolor{p2col8}{rgb}{0.897354, 0.41824, 0.185007}
\definecolor{p2col9}{rgb}{0.857359, 0.131106, 0.132128}

\newtheorem{theorem}{Theorem}

\newtheorem{lemma}[theorem]{Lemma}
\theoremstyle{remark}

\newtheorem{remark}[theorem]{Remark}


\hypersetup{
  pdftitle={Combinatorial generation via permutation languages. I. Fundamentals},
  pdfauthor={Elizabeth Hartung, Hung Phuc Hoang, Torsten M\"utze, Aaron Williams}
}

\title{Combinatorial generation via permutation languages. \\ I. Fundamentals}

\author{Elizabeth Hartung}
\address[Elizabeth Hartung]{Massachusetts College of Liberal Arts, United States}
\email{e.hartung@mcla.edu}

\author{Hung P. Hoang}
\address[Hung P. Hoang]{Department of Computer Science, ETH Z\"urich, Switzerland}
\email{hung.hoang@inf.ethz.ch}

\author{Torsten M\"utze}
\address[Torsten M\"utze]{Department of Computer Science, University of Warwick, United Kingdom}
\email{torsten.mutze@warwick.ac.uk}

\author{Aaron Williams}
\address[Aaron Williams]{Computer Science Department, Williams College, United States}
\email{aaron.williams@williams.edu}

\thanks{An extended abstract of this paper appeared in the Proceedings of the 31st Annual ACM-SIAM Symposium on Discrete Algorithms (SODA~2020)~\cite{DBLP:conf/soda/HartungHMW20}.}
\thanks{Torsten M\"utze is also affiliated with the Faculty of Mathematics and Physics, Charles University Prague, Czech Republic. He was supported by Czech Science Foundation grant GA~19-08554S, and by German Science Foundation grant~413902284.}

\begin{document}

\begin{abstract}
In this work we present a general and versatile algorithmic framework for exhaustively generating a large variety of different combinatorial objects, based on encoding them as permutations.
This approach provides a unified view on many known results and allows us to prove many new ones.
In particular, we obtain the following four classical Gray codes as special cases: the Steinhaus-Johnson-Trotter algorithm to generate all permutations of an $n$-element set by adjacent transpositions; the binary reflected Gray code to generate all $n$-bit strings by flipping a single bit in each step; the Gray code for generating all $n$-vertex binary trees by rotations due to Lucas, Roelants van Baronaigien, and Ruskey; the Gray code for generating all partitions of an $n$-element ground set by element exchanges due to Kaye.

We present two distinct applications for our new framework:
The first main application is the generation of pattern-avoiding permutations, yielding new Gray codes for different families of permutations that are characterized by the avoidance of certain classical patterns, (bi)vincular patterns, barred patterns, boxed patterns, Bruhat-restricted patterns, mesh patterns, monotone and geometric grid classes, and many others.
We also obtain new Gray codes for all the combinatorial objects that are in bijection to these permutations, in particular for five different types of geometric rectangulations, also known as floorplans, which are divisions of a square into $n$ rectangles subject to certain restrictions.

The second main application of our framework are lattice congruences of the weak order on the symmetric group~$S_n$.
Recently, Pilaud and Santos realized all those lattice congruences as $(n-1)$-dimensional polytopes, called quotientopes, which generalize hypercubes, associahedra, permutahedra etc.
Our algorithm generates the equivalence classes of each of those lattice congruences, by producing a Hamilton path on the skeleton of the corresponding quotientope, yielding a constructive proof that each of these highly symmetric graphs is Hamiltonian.
We thus also obtain a provable notion of optimality for the Gray codes obtained from our framework: They translate into walks along the edges of a polytope.
\end{abstract}


\maketitle

\newpage

\section{Introduction}
\label{sec:intro}

In mathematics and computer science we frequently encounter different kinds of combinatorial objects, such as permutations, binary strings, binary trees, set partitions, spanning trees of a graph, and so forth.
There are three recurring fundamental algorithmic tasks that we want to perform with such objects: counting, random generation, and exhaustive generation.
For the first two tasks, there are powerful general methods available, such as generating functions~\cite{MR2483235} and Markov chains~\cite{MR1960003}, solving both problems for a large variety of different objects.
For the third task, namely exhaustive generation, however, we are lacking such a powerful and unifying theory, even though some first steps in this direction have been made (see Section~\ref{sec:related} below).
Nonetheless, the literature contains a vast number of algorithms that solve the exhaustive generation problem for specific classes of objects, and many of these algorithms are covered in depth in the most recent volume of Knuth's seminal series \emph{`The Art of Computer Programming'}~\cite{MR3444818}.

\subsection{Overview of our results}

The main contribution of this paper is a general and versatile algorithmic framework for exhaustively generating a large variety of different combinatorial objects, which provides a unified view on many known results and allows us to prove many new ones.
The basic idea is to encode a particular set of objects as a set of permutations~$L_n\seq S_n$, where $S_n$ denotes the set of all permutations of $[n]:=\{1,2,\ldots,n\}$, and to use a simple greedy algorithm to generate those permutations by cyclic rotations of substrings, an operation we call a \emph{jump}.
This works under very mild assumptions on the set~$L_n$, and allows us to generate more than double-exponentially (in~$n$) many distinct sets~$L_n$.
Moreover, the jump orderings obtained from our algorithm translate into listings of combinatorial objects where consecutive objects differ by small changes, i.e., we obtain \emph{Gray codes}~\cite{MR1491049}, and those changes are smallest possible in a provable sense.
The main tools of our framework are Algorithm~J and Theorem~\ref{thm:jump} in Section~\ref{sec:jump}.
In particular, we obtain the following four classical Gray codes as special cases:
(1) the Steinhaus-Johnson-Trotter algorithm to generate all permutations of~$[n]$ by adjacent transpositions, also known as plain change order~\cite{DBLP:journals/cacm/Trotter62,MR0159764};
(2) the binary reflected Gray code (BRGC) to generate all binary strings of length~$n$ by flipping a single bit in each step~\cite{gray_1953};
(3) the Gray code for generating all $n$-vertex binary trees by rotations due to Lucas, Roelants van Baronaigien, and Ruskey~\cite{MR1239499};
(4) the Gray code for generating all set partitions of~$[n]$ by exchanging an element in each step due to Kaye~\cite{MR0443423}.

We see two main applications for our new framework:
The first application is the generation of pattern-avoiding permutations, yielding new Gray codes for different families of permutations that are characterized by the avoidance of certain classical patterns, vincular and bivincular patterns~\cite{MR1758852,MR2652101}, barred patterns~\cite{MR2716312}, boxed patterns~\cite{MR2973348}, Bruhat-restricted patterns~\cite{MR2264071}, mesh patterns~\cite{MR2795782}, monotone and geometric grid classes~\cite{MR2240760,MR3091268}, and many others.
We also obtain new Gray codes for all the combinatorial objects that are in bijection to these permutations, in particular for five different types of geometric rectangulations~\cite{MR2233287,MR2864445,MR3084577,MR3878132}, also known as floorplans, which are divisions of a square into $n$ rectangles subject to different restrictions.
Our results on pattern-avoiding permutations are the focus of this paper.

The second application of our framework are lattice congruences of the weak order on the symmetric group~$S_n$.
This area has beautiful ramifications into groups, posets, polytopes, geometry, and combinatorics, and has been developed considerably in recent years, in particular thanks to Nathan Reading's works, summarized in~\cite{MR3221544, MR3645056, MR3645055}.
There are double-exponentially many distinct such lattice congruences, and they generalize many known lattices such as the Boolean lattice, the Tamari lattice~\cite{MR0146227}, and certain Cambrian lattices~\cite{MR2258260,MR3628225}.
Recently, Pilaud and Santos~\cite{MR3964495} realized all those lattice congruences as $(n-1)$-dimensional polytopes, called quotientopes, which generalize hypercubes, associahedra, permutahedra etc.
Our algorithm generates the equivalence classes of each of those lattice congruences, by producing a Hamilton path on the skeleton of the corresponding quotientope, yielding a constructive proof that each of these highly symmetric graphs is Hamiltonian.
Our results in this area are presented in part~II of this paper series~\cite{perm_series_ii}.

\subsection{Related work}
\label{sec:related}

Avis and Fukuda~\cite{MR1380066} introduced \emph{reverse-search} as a general technique for exhaustive generation.
Their idea is to consider the set of objects to be generated as the nodes of a graph, and to connect them by edges that model local modification operations (for instance, adjacent transpositions for permutations).
The resulting \emph{flip graph} is equipped with an objective function, and the directed tree formed by the movements of a local search algorithm that optimizes this function is traversed backwards from the optimum node, using an adjacency oracle.
The authors applied this technique successfully to derive efficient generation algorithms for a number of different objects; for instance, triangulations of a point set, spanning trees of a graph, etc.
Reverse-search is complementary to our permutation based approach, as both techniques use fundamentally different encodings of the objects.
The permutation encoding seems to allow for more fine-grained control (optimal Gray codes) and even faster generation algorithms.

Another method for combinatorial counting and exhaustive generation is the \emph{ECO framework} introduced by Barcucci, Del Lungo, Pergola, and Pinzani~\cite{MR1717162}.
The main tool is an infinite tree with integer node labels, and a set of production rules for creating the children of a node based on its label.
Bacchelli, Barcucci, Grazzini, and Pergola~\cite{MR2074330} also used ECO for exhaustive generation, deriving an efficient algorithm for generating the corresponding root-to-node label sequences in the ECO tree in lexicographic order, which was later turned into a Gray code~\cite{MR2329122}.
Dukes, Flanagan, Mansour, and Vajnovszki~\cite{MR2412243}, Baril~\cite{MR2530633}, and Do, Tran and Vajnovszki~\cite{do_tran_vajnovszki_2019} used ECO for deriving Gray codes for different classes of pattern-avoiding permutations, which works under certain regularity assumptions on the production rules.
Vajnovszki~\cite{MR2654252} also applied ECO for efficiently generating other classes of permutations, such as involutions and derangements.
The main difference between ECO and our framework is that the change operations on the label sequences of the ECO tree do not necessarily correspond to Gray-code like changes on the corresponding combinatorial objects.
Minimal jumps in a permutation, on the other hand, always correspond to minimal changes on the combinatorial objects in a provable sense, even though they may involve several entries of the permutation.

Li and Sawada~\cite{MR2483809} considered another tree-based approach for generating so-called \emph{reflectable languages}, yielding Gray codes for $k$-ary strings and trees, restricted growth strings, and open meandric systems (see also~\cite{DBLP:conf/iccsa/XiangCU10}).
Ruskey, Sawada, and Williams~\cite{MR2844089, MR2900417} proposed a generation framework based on binary strings with a fixed numbers of 1s, called \emph{bubble languages}, which can generate e.g.\ combinations, necklaces, Dyck words, and Lyndon words.
In the resulting cool-lex Gray codes, any two consecutive words differ by cyclic rotation of some prefix.

Pattern avoidance in permutations is a central topic in combinatorics, as illustrated by the books~\cite{MR3012380,MR2919720}, and by the conference `Permutation Patterns', held annually since~2003.
Given two permutations~$\pi$ and~$\tau$, we say that $\pi$ \emph{contains the pattern $\tau$}, if $\pi$ contains a subpermutation formed by (not necessarily consecutive) entries that appear in the same relative order as in~$\tau$; otherwise we say that \emph{$\pi$ avoids~$\tau$}.
It is well known that many fundamental classes of combinatorial objects are in bijection with pattern-avoiding permutations (see Tables~\ref{tab:tame1} and~\ref{tab:tame2} and~\cite{tenner_website}).
For instance, Knuth~\cite{MR3077154} first proved that all 123-avoiding and 132-avoiding permutations are counted by the Catalan numbers (see also~\cite{MR2721479}).
With regards to counting and exhaustive generation, a few tree-based algorithms for pattern-avoiding permutations have been proposed~\cite{MR2376115,MR2412243,MR2566173,MR2530633}.
Pattern-avoidance has also been studied extensively from an algorithmic point of view.
In fact, testing whether a permutation~$\pi$ contains another permutation~$\tau$ as a pattern is known to be NP-complete in general~\cite{MR1620935}.
Jel{\'{i}}nek and Kyn{\v{c}}l~\cite{MR3627753} proved that the problem remains hard even if~$\pi$ and~$\tau$ have no decreasing subsequence of length~4 and~3, respectively, which is best possible.
On the algorithmic side, Guillemot and Marx~\cite{MR3376367} showed that the problem can be solved in time $2^{O(k^2\log k)}n \log n$, where $n$ is the length of~$\pi$ and~$k$ is the length of~$\tau$, a considerable improvement over the obvious $O(n^k)$ algorithm (see also~\cite{kozma_2019}).
In particular, for a pattern of constant length~$k$ their algorithm runs in almost linear time.
It is also known that for a fixed set of forbidden patterns, computing the number of pattern-avoiding permutations is hard~\cite{MR3478442}.

\subsection{Outline of this and future papers}

This is the first in a series of papers where we develop our theory of combinatorial generation via permutation languages.
In this first paper we focus on presenting the fundamental algorithmic ideas (Section~\ref{sec:jump}) and show how to derive the four aforementioned classical Gray codes from our framework (Section~\ref{sec:examples}).
We also discuss the main applications of our framework to pattern-avoiding permutations (Sections~\ref{sec:pattern} and~\ref{sec:algebra}), and its limitations (Section~\ref{sec:limits}).
In part~II of the series~\cite{perm_series_ii}, we apply the framework to lattice congruences of the weak order on the symmetric group~$S_n$.
In part~III we discuss the generation of different types of rectangulations.
In part~IV we cover general methods to make the generation algorithms derived from our framework efficient, and we discuss the problems of ranking/unranking for the orderings obtained from our framework, which are highly relevant for the task of random generation mentioned in the beginning.
Two of the latter topics, namely efficient algorithms and rectangulations, are only very briefly discussed in this paper (see Section~\ref{sec:efficient} and Figure~\ref{fig:tbaxter4} below, respectively).

\section{Generating permutations by jumps}
\label{sec:jump}

In this section we present a simple greedy algorithm, Algorithm~J, for exhaustively generating a given set~$L_n\seq S_n$ of permutations, and we show that the algorithm works successfully under very mild assumptions on the set~$L_n$ (Theorem~\ref{thm:jump}).

\subsection{Preliminaries}

We use $S_n$ to denote the set of all permutations of $[n]:=\{1,\ldots,n\}$, and we write~$\pi\in S_n$ in one-line notation as $\pi=\pi(1) \pi(2)\cdots \pi(n)=a_1 a_2\cdots a_n$.
We use $\ide_n=12\cdots n$ to denote the identity permutation, and $\varepsilon\in S_0$ to denote the empty permutation.
For any $\pi\in S_{n-1}$ and any $1\leq i\leq n$, we write $c_i(\pi)\in S_n$ for the permutation obtained from~$\pi$ by inserting the new largest value~$n$ at position~$i$ of~$\pi$, i.e., if $\pi=a_1\cdots a_{n-1}$ then $c_i(\pi)=a_1\cdots a_{i-1} \, n\, a_i \cdots a_{n-1}$.
Moreover, for~$\pi\in S_n$, we write $p(\pi)\in S_{n-1}$ for the permutation obtained from~$\pi$ by removing the largest entry~$n$.
Here, $c_i$ and~$p$ stand for the child and parent of a node in the tree of permutations discussed shortly.

Given a permutation $\pi=a_1\cdots a_n$ with a substring $a_i\cdots a_j$ with $a_i>a_{i+1},\ldots,a_j$, a \emph{right jump of the value~$a_i$ by $j-i$~steps} is a cyclic left rotation of this substring by one position to $a_{i+1}\cdots a_j a_i$.
Similarly, given a substring $a_i\cdots a_j$ with $a_j>a_i,\ldots,a_{j-1}$, a \emph{left jump of the value~$a_j$ by $j-i$~steps} is a cyclic right rotation of this substring to $a_j a_i\cdots a_{j-1}$.
For example, a right jump of the value~5 in the permutation~$265134$ by 2 steps yields~$261354$.

\subsection{The basic algorithm}
\label{sec:algo}

Our approach starts with the following simple greedy algorithm to generate a set of permutations $L_n\seq S_n$.
We say that a jump is \emph{minimal} (w.r.t.~$L_n$), if every jump of the same value in the same direction by fewer steps creates a permutation that is not in~$L_n$.
Note that each entry of the permutation admits at most one minimal left jump and at most one minimal right jump.

\begin{algo}{Algorithm~J}{Greedy minimal jumps}
This algorithm attempts to greedily generate a set of permutations $L_n\seq S_n$ using minimal jumps starting from an initial permutation $\pi_0 \in L_n$.
\begin{enumerate}[label={\bfseries J\arabic*.}, leftmargin=8mm, noitemsep, topsep=3pt plus 3pt]
\item{} [Initialize] Visit the initial permutation~$\pi_0$.
\item{} [Jump] Generate an unvisited permutation from~$L_n$ by performing a minimal jump of the largest possible value in the most recently visited permutation.
If no such jump exists, or the jump direction is ambiguous, then terminate.
Otherwise visit this permutation and repeat~J2.
\end{enumerate}
\end{algo}

Put differently, in step~J2 we consider the entries $n,n-1,\ldots,2$ of the current permutation in decreasing order, and for each of them we check whether it allows a minimal left or right jump that creates a previously unvisited permutation, and we perform the first such jump we find, unless the same entry also allows a jump in the opposite direction, in which case we terminate.
If no minimal jump creates an unvisited permutation, we also terminate the algorithm.
For example, consider $L_4 = \{1243, 1423, 4123, 4213, 2134\}$.
Starting with $\pi_0 = 1243$, the algorithm generates~$\pi_1=1423$ (obtained from~$\pi_0$ by a left jump of the value~4 by 1~step), then~$\pi_2=4123$, then~$\pi_3=4213$ (in~$\pi_2$, 4 cannot jump, as~$\pi_0$ and~$\pi_1$ have been visited before; 3 cannot jump either to create any permutation from~$L_4$, so 2 jumps left by 1~step), then $\pi_4=2134$, successfully generating~$L_4$.
If instead we initialize with $\pi_0=4213$, then the algorithm generates $\pi_1=2134$, and then stops, as no further jump is possible.
If we choose $\pi_0=1423$, then we may jump~4 to the left or right (by 1~step), but as the direction is ambiguous, the algorithm stops immediately.
As mentioned before, the algorithm may stop before having visited the entire set~$L_n$ either because no minimal jump leading to a new permutation from~$L_n$ is possible, or because the direction of jump is ambiguous in some step.
By the definition of step~J2, the algorithm will never visit any permutation twice.

\subsection{Zigzag languages}
\label{sec:zigzag}

The following main result of our paper provides a sufficient condition on the set~$L_n$ to guarantee that Algorithm~J is successful (cf.~Section~\ref{sec:limits}).
This condition is captured by the following closure property of the set~$L_n$.
A set of permutations~$L_n\seq S_n$ is called a \emph{zigzag language}, if either $n=0$ and $L_0=\{\varepsilon\}$, or if $n\geq 1$ and $L_{n-1}:=\{p(\pi)\mid \pi\in L_n\}$ is a zigzag language satisfying the following condition:
\begin{enumerate}[leftmargin=8mm, noitemsep, topsep=3pt plus 3pt]
\item[(z)] For every $\pi\in L_{n-1}$ we have~$c_1(\pi)\in L_n$ and~$c_n(\pi)\in L_n$.
\end{enumerate}

\begin{theorem}
\label{thm:jump}
Given any zigzag language of permutations~$L_n$ and initial permutation $\pi_0 = \ide_n$, Algorithm~J visits every permutation from~$L_n$ exactly once.
\end{theorem}

\begin{remark}
\label{rem:num-zigzag}
It is easy to see that the number of zigzag languages is at least $2^{(n-1)!(n-2)}=2^{2^{\Theta(n\log n)}}$, i.e., it is more than double-exponential in~$n$.
We will see that many of these languages do in fact encode interesting combinatorial objects.
Moreover, minimal jumps as performed by Algorithm~J always translate to small changes on those objects in a provable sense, i.e., our algorithm defines Gray codes for a large variety of combinatorial objects, and Hamilton paths/cycles on the corresponding flip graphs and polytopes.
\end{remark}

Before we present the proof of Theorem~\ref{thm:jump}, we give two equivalent characterizations of zigzag languages.

\subsubsection{Characterization via the tree of permutations}
\label{sec:tree}

There is an intuitive characterization of zigzag languages via the \emph{tree of permutations}.
This is an infinite (unordered) rooted tree which has as nodes all permutations from~$S_n$ at distance~$n$ from the root; see Figure~\ref{fig:tree}.
Specifically, the empty permutation~$\varepsilon$ is at the root, and the children of any node~$\pi\in S_{n-1}$ are exactly the permutations~$c_i(\pi)$, $1\leq i\leq n$, i.e., the permutations obtained by inserting the new largest value~$n$ in all possible positions.
Consequently, the parent of any node~$\pi'\in S_n$ is exactly the permutation $p(\pi')$ obtained by removing the largest value~$n$.
In the figure, for any node~$\pi\in S_{n-1}$, the nodes representing the children~$c_1(\pi)$ and~$c_n(\pi)$ are drawn black, whereas the other children are drawn white.
Any zigzag language of permutations can be obtained from this full tree by pruning subtrees, where by condition~(z) a subtree may be pruned only if its root $\pi'\in S_n$ is neither the child~$c_1(\pi)$ nor the child~$c_n(\pi)$ of its parent~$\pi=p(\pi')\in S_{n-1}$, i.e., only subtrees rooted at white nodes may be pruned.
For any subtree obtained by pruning according to this rule and for any~$n\geq 1$, the remaining permutations of length~$n$ form a zigzag language~$L_n$; see Figure~\ref{fig:ptree}.

\begin{figure}
\includegraphics{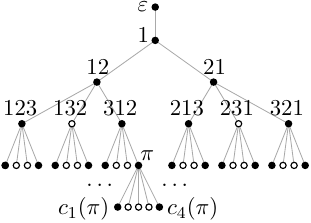}
\caption{Tree of permutations, where the children~$c_1(\pi)$ and~$c_n(\pi)$ of any node $\pi\in S_{n-1}$ are drawn black, all others white.}
\label{fig:tree}
\end{figure}

\begin{figure}
\includegraphics{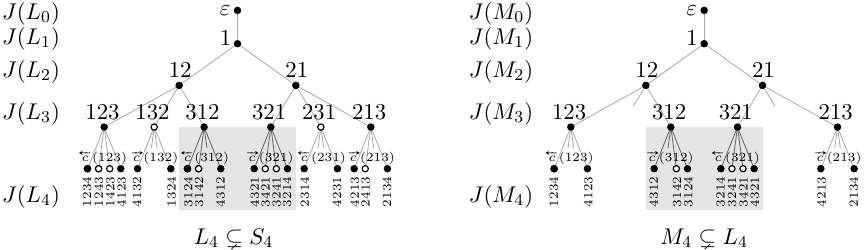}
\caption{Ordered tree representation of two zigzag languages of permutations $L_4$ (left) and~$M_4$ (right) with $M_4\subsetneq L_4\subsetneq S_4$.
Both trees contain the same sets of permutations in the subtrees rooted at~312 and~321 (highlighted in gray), but in the corresponding sequences~$J(L_4)$ and~$J(M_4)$, those permutations appear in different relative order due to the node~132, which was pruned from the right tree.}
\label{fig:ptree}
\end{figure}

Consider all nodes in the tree for which the entire path to the root consists only of black nodes.
Those nodes never get pruned and are therefore contained in any zigzag language.
These are exactly all permutations without peaks.
A \emph{peak} in a permutation $a_1\cdots a_n$ is a triple $a_{i-1} a_i a_{i+1}$ with $a_{i-1}<a_i>a_{i+1}$, and the language of permutations without peaks is generated by the recurrence $P_0:=\{\varepsilon\}$ and $P_n:=\{c_1(\pi),c_n(\pi)\mid \pi\in P_{n-1}\}$ for~$n\geq 1$.
It follows that we have $|P_n|=2^{n-1}$ and $P_n\seq L_n\seq S_n$ for any zigzag language~$L_n$, i.e., $L_n$ is sandwiched between the language of permutations without peaks and the language of all permutations.

\subsubsection{Characterization via nuts}
\label{sec:nuts}

Given a permutation~$\pi$, we may repeatedly remove the largest value from it as long as it is in the leftmost or rightmost position, and obtain what is called the \emph{nut} of~$\pi$.
For example, given $\pi=965214378$, we can remove $9,8,7,6,5$, yielding $2143$ as the nut of~$\pi$.
A left or right jump of some value in a permutation is \emph{maximum} if there is no left jump or right jump of the same value with more steps.
For example, in~$\pi=965214378$ a maximum right jump of the value~$6$ gives $\pi'=952143678$.
By unrolling the recursive definition of zigzag languages from before, we obtain that $L_n\seq S_n$ is a zigzag language if and only if for all $\pi \in L_n$ both the maximum left jump and the maximum right jump of the value~$i$ yield another permutation in~$L_n$ for all $k \leq i \leq n$, where $k$ is the largest value in $\pi$'s nut (with $k = 2$ if the nut is empty).

\subsection{Proof of Theorem~\ref{thm:jump}}

Given a zigzag language~$L_n$, we define a sequence~$J(L_n)$ of all permutations from~$L_n$, and we prove that Algorithm~J generates the permutations of~$L_n$ exactly in this order.
For any $\pi\in L_{n-1}$ we let $\rvec{c}(\pi)$ be the sequence of all $c_i(\pi)\in L_n$ for $i=1,2,\ldots,n$, starting with~$c_1(\pi)$ and ending with~$c_n(\pi)$, and we let $\lvec{c}(\pi)$ denote the reverse sequence, i.e., it starts with~$c_n(\pi)$ and ends with~$c_1(\pi)$.
In words, those sequences are obtained by inserting into~$\pi$ the new largest value~$n$ from left to right, or from right to left, respectively, in all possible positions that yield a permutation from~$L_n$, skipping the positions that yield a permutation that is not in~$L_n$.
The sequence~$J(L_n)$ is defined recursively as follows:
If $n=0$ then we define $J(L_0):=\varepsilon$, and if $n\geq 1$ then we consider the finite sequence $J(L_{n-1})=:\pi_1,\pi_2,\ldots$ and define
\begin{equation}
\label{eq:JLn}
J(L_n):=\lvec{c}(\pi_1),\rvec{c}(\pi_2),\lvec{c}(\pi_3),\rvec{c}(\pi_4),\ldots,
\end{equation}
i.e., this sequence is obtained from the previous sequence by inserting the new largest value~$n$ in all possible positions alternatingly from right to left, or from left to right; see Figure~\ref{fig:ptree}.

\begin{remark}
Algorithm~J thus defines a left-to-right ordering of the nodes at distance~$n$ of the root in the tree representation of the zigzag language~$L_n$ described before, and this ordering is captured by the sequence~$J(L_n)$; see Figure~\ref{fig:ptree}.
Clearly, the same is true for all the zigzag languages~$L_0,L_1,\ldots,L_{n-1}$ that are induced by~$L_n$ through the rule~$L_{k-1}:=\{p(\pi)\mid \pi\in L_k\}$ for $k=n,n-1,\ldots,1$.
The unordered tree is thus turned into an ordered tree, and it is important to realize that pruning operations change the ordering.
Specifically, given two zigzag languages $L_n$ and~$M_n$ with~$M_n\seq L_n$, the tree for~$M_n$ is obtained from the tree for~$L_n$ by pruning, but in general $J(M_n)$ is \emph{not} a subsequence of~$J(L_n)$, as shown by the example in Figure~\ref{fig:ptree}.
This shows that our approach is quite different from the one presented by Vajnovszki and Vernay~\cite{MR2817083}, which considers only subsequences of the Steinhaus-Johnson-Trotter order~$J(S_n)$.
\end{remark}

\begin{proof}[Proof of Theorem~\ref{thm:jump}]
For any $\pi\in L_n$, we let~$J(L_n)_\pi$ denote the subsequence of~$J(L_n)$ that contains all permutations up to and including~$\pi$.
An immediate consequence of the definition of zigzag language is that $L_n$ contains the identity permutation~$\ide_n=c_n(\ide_{n-1})$.
Moreover, the definition~\eqref{eq:JLn} implies that $\ide_n$ is the very first permutation in the sequence~$J(L_n)$.

We now argue by double induction over $n$ and the length of~$J(L_n)$ that Algorithm~J generates all permutations from~$L_n$ exactly in the order described by the sequence~$J(L_n)$, and that when we perform a minimal jump with the largest possible value to create a previously unvisited permutation, then there is only one direction (left or right) to which it can jump.
The induction basis~$n=0$ is clear.
Now suppose the claim holds for the zigzag language~$L_{n-1}:=\{p(\pi)\mid \pi\in L_n\}$.
We proceed to show that it also holds for~$L_n$.

As argued before, the identity permutation~$\ide_n$ is the first permutation in the sequence~$J(L_n)$, and this is indeed the first permutation visited by Algorithm~J in step~J1.
Now let $\pi\in L_n$ be the permutation currently visited by the algorithm in step~J2, and let $\pi':=p(\pi)\in L_{n-1}$.
If $\pi'$ appears at an odd position in~$J(L_{n-1})$, then we define $\bar{c}:=\lvec{c}(\pi')$ and otherwise we define $\bar{c}:=\rvec{c}(\pi')$.
By~\eqref{eq:JLn}, we know that $\pi$ appears in the subsequence $\bar{c}$ within~$J(L_n)$.
We first consider the case that $\pi$ is not the last permutation in~$\bar{c}$.
In this case, the permutation $\rho$ succeeding $\pi$ in~$J(L_n)$ is obtained from~$\pi$ by a minimal jump (w.r.t.~$L_n$) of the largest value~$n$ in some direction~$d$, which is left if $\bar{c}=\lvec{c}(\pi')$ and right if $\bar{c}=\rvec{c}(\pi')$.
Now observe that by the definition of~$\bar{c}$, all permutations in~$L_n$ obtained from~$\pi$ by jumping~$n$ in the direction opposite to~$d$ precede~$\pi$ in~$J(L_n)$ and have been visited by Algorithm~J before by induction.
Consequently, to generate a previously unvisited permutation, the value~$n$ can only jump in direction~$d$ in step~J2 of the algorithm.
Again by the definition of~$\bar{c}$, the permutation~$\rho$ is obtained from~$\pi$ by a minimal jump (w.r.t.~$L_n$), so the next permutation generated by the algorithm will indeed be~$\rho$.

It remains to consider the case that $\pi$ is the last permutation in the subsequence~$\bar{c}$ within~$J(L_n)$.
Let $\rho'$ be the permutation succeeding~$\pi'$ in~$J(L_{n-1})$.
By induction, we have the following property~(*): $\rho'$ is obtained from~$\pi'$ by a minimal jump (w.r.t.~$L_{n-1}$) of the largest possible value~$a$ by $k$~steps in some direction~$d$ (left or right), and~$a$ can jump only into one direction.
As $\pi$ is the last permutation in~$\bar{c}$, the largest value~$n$ of~$\pi$ is at the boundary, which is the left boundary if $\bar{c}=\lvec{c}(\pi')$ or the right boundary if $\bar{c}=\rvec{c}(\pi')$.
By~\eqref{eq:JLn}, the permutation~$\rho$ succeeding~$\pi$ in~$J(L_n)$ also has $n$ at the same boundary, i.e., $\rho$ differs from~$\pi$ by a jump of the value~$a$ by $k$~steps in direction~$d$.
Suppose for the sake of contradiction that when transforming the currently visited permutation~$\pi$ in step~J2, the algorithm does not perform this jump operation, but another one.
This could be a jump of a larger value~$b>a$ to transform~$\pi$ into some permutation~$\tau\in L_n$ that is different from~$\rho$ and not in~$J(L_n)_\pi$, or a jump of the value~$a$ in the direction opposite to~$d$, or a jump of the value~$a$ in direction~$d$ by fewer than~$k$ steps.
But in all those cases the permutation~$\tau':=p(\tau)\in L_{n-1}$ is different from~$\rho'$ and not in~$J(L_{n-1})_{\pi'}$, and it is obtained from~$\pi'$ by a jump of the value~$b>a$, or a jump of the value~$a$ in the direction opposite to~$d$, or a jump of the value~$a$ in direction~$d$ by fewer than~$k$ steps, respectively, a contradiction to property~(*).
This completes the proof.
\end{proof}

\subsection{Further properties of Algorithm~J}

The next lemma captures when Algorithm~J generates a \emph{cyclic} listing of permutations.

\begin{lemma}
\label{lem:cyclic}
In the ordering of permutations~$J(L_n)$ generated by Algorithm~J, the first and last permutation are related by a minimal jump if and only if $|L_i|$ is even for all~$2\leq i\leq n-1$.
\end{lemma}

For example, the conditions described by Lemma~\ref{lem:cyclic} are satisfied for the zigzag languages~$L_n=S_n$ (all permutations) and $L_n=P_n$ (permutations without peaks), and the resulting cyclic orderings~$J(L_n)$ are shown in Figures~\ref{fig:perm4} and~\ref{fig:bits4}, respectively.
Another cyclic Gray code is shown in Figure~\ref{fig:tbaxter4}.
In contrast to that, the Gray codes shown in Figures~\ref{fig:cat4} and~\ref{fig:part4} violate the conditions of the lemma and are therefore not cyclic.

\begin{proof}
Let $\pi_i$ be the last permutation in the ordering $J(L_i)$ for all $i=0,1,\ldots,n$.
For $i\geq 1$, we see from~\eqref{eq:JLn} that $\pi_i=c_i(\pi_{i-1})$ if $|L_{i-1}|$ is even and $\pi_i=c_1(\pi_{i-1})$ if $|L_{i-1}|$ is odd.
As $|L_1|=1$ is odd, we know that 1 and 2 are reversed in~$\pi_n$, and so all numbers $|L_i|$, $2\leq i\leq n-1$, must be even for~$\ide_n$ and~$\pi_n$ to be related by a minimal jump.
\end{proof}

\begin{remark}
\label{rem:cyclic}
It follows from the proof of Theorem~\ref{thm:jump} that instead of initializing the algorithm with the identity permutation~$\pi_0=\ide_n$, we may use any permutation without peaks as a seed~$\pi_0$.
\end{remark}

\subsection{Efficiency considerations}
\label{sec:efficient}

Let us make it very clear that in its stated form, Algorithm~J is not an efficient algorithm to actually generate a particular zigzag language of permutations.
The reason is that it requires storing the list of all previously visited permutations in order to decide which one to generate next.
However, by introducing a few additional arrays, the algorithm can be made memoryless, so that such lookup operations are not needed anymore, and hence no permutations need to be stored at all.
The efficiency of the resulting algorithm is then only determined by the efficiency with which we are able to compute minimal jumps with respect to the input zigzag language~$L_n$ for a given entry of the permutation.
This leads to an algorithm that computes the next permutation to be visited in polynomial time.
In many cases, this can be improved to a loopless algorithm that generates each new permutation in constant worst-case time.
The key insight here is that any jump changes the inversion table of a permutation only in a single entry.
By maintaining only the inversion table, jumps can thus be performed efficiently, even if the number of steps is big.
This extensive and important discussion, however, is not the main focus here, and is deferred to part~IV of this paper series.

\subsection{A general recipe}
\label{sec:recipe}

Here is a step-by-step approach to apply our framework to the generation of a given family~$X_n$ of combinatorial objects.
The first step is to establish a bijection~$f$ that encodes the objects from~$X_n$ as permutations~$L_n\seq S_n$.
If $L_n$ is a zigzag language, which can be checked by verifying the closure property, then we may run Algorithm~J with input~$L_n$, and interpret the resulting ordering~$J(L_n)$ in terms of the combinatorial objects, by applying~$f^{-1}$ to each permutation in~$J(L_n)$, yielding an ordering on~$X_n$.
We may also apply~$f^{-1}$ to Algorithm~J directly, which will yield a simple greedy algorithm for generating~$X_n$.
The final step is to make these algorithms efficient, by introducing additional data structures that allow the change operations on~$X_n$ (which are the preimages of minimal jumps under~$f$) as efficiently as possible.
In the next section we amply illustrate this approach by four examples.

\section{Classical Gray codes from our framework}
\label{sec:examples}

In this section we derive the four classical Gray codes mentioned in the introduction from our framework in a systematic fashion, following the approach outlined in Section~\ref{sec:recipe}.

\subsection{Permutations (Steinhaus-Johnson-Trotter)}

Consider the set~$X_n=L_n=S_n$ of all permutations of~$[n]$.
The bijection~$f$ between~$X_n$ and~$L_n$ here is simply the identity, i.e., $f=\ide$.
In this case, each jump is a jump by 1~step, i.e., it is an adjacent transposition.
Algorithm~J thus yields an ordering of permutations by adjacent transpositions, which coincides with the well-known Steinhaus-Johnson-Trotter order, also known as plain change order~\cite{DBLP:journals/cacm/Trotter62,MR0159764}, which can be implemented efficiently~\cite{MR3444818}.
This ordering is shown in Figure~\ref{fig:perm4}.
Algorithm~J translates into the following simple greedy algorithm to describe this order (see \cite{DBLP:conf/wads/Williams13}):
\textbf{J1.} Visit the identity permutation. \textbf{J2.} Perform a transposition of the largest possible value with an adjacent smaller entry that yields a previously unvisited permutation; then visit this permutation and repeat~J2.

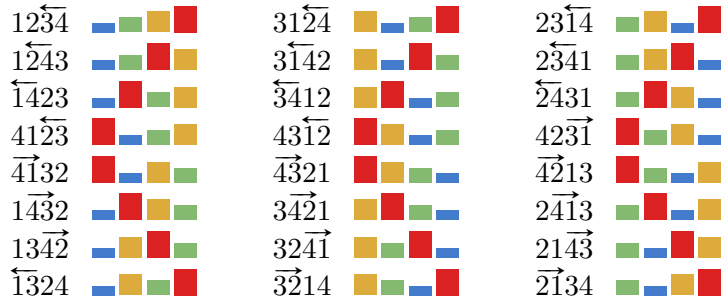
\begin{figure}
\begin{tabular}{c@{\hskip 3mm}c@{\hskip 10mm}c@{\hskip 3mm}c@{\hskip 10mm}c@{\hskip 3mm}c}
$12\lvec{34}$ & \begin{tikzpicture}[scale=30/100]
\fill[fill=p1col1] (0,0) rectangle (1,0.4);
\fill[fill=p1col2] (1.2,0) rectangle (2.2,0.65);
\fill[fill=p1col3] (2.4000000000000004,0) rectangle (3.4000000000000004,0.9);
\fill[fill=p1col4] (3.6000000000000005,0) rectangle (4.6000000000000005,1.15);
\end{tikzpicture}
 & $31\lvec{24}$ & \begin{tikzpicture}[scale=30/100]
\fill[fill=p1col3] (0,0) rectangle (1,0.9);
\fill[fill=p1col1] (1.2,0) rectangle (2.2,0.4);
\fill[fill=p1col2] (2.4000000000000004,0) rectangle (3.4000000000000004,0.65);
\fill[fill=p1col4] (3.6000000000000005,0) rectangle (4.6000000000000005,1.15);
\end{tikzpicture}
 & $23\lvec{14}$ & \begin{tikzpicture}[scale=30/100]
\fill[fill=p1col2] (0,0) rectangle (1,0.65);
\fill[fill=p1col3] (1.2,0) rectangle (2.2,0.9);
\fill[fill=p1col1] (2.4000000000000004,0) rectangle (3.4000000000000004,0.4);
\fill[fill=p1col4] (3.6000000000000005,0) rectangle (4.6000000000000005,1.15);
\end{tikzpicture}
 \\
$1\lvec{24}3$ & \begin{tikzpicture}[scale=30/100]
\fill[fill=p1col1] (0,0) rectangle (1,0.4);
\fill[fill=p1col2] (1.2,0) rectangle (2.2,0.65);
\fill[fill=p1col4] (2.4000000000000004,0) rectangle (3.4000000000000004,1.15);
\fill[fill=p1col3] (3.6000000000000005,0) rectangle (4.6000000000000005,0.9);
\end{tikzpicture}
 & $3\lvec{14}2$ & \begin{tikzpicture}[scale=30/100]
\fill[fill=p1col3] (0,0) rectangle (1,0.9);
\fill[fill=p1col1] (1.2,0) rectangle (2.2,0.4);
\fill[fill=p1col4] (2.4000000000000004,0) rectangle (3.4000000000000004,1.15);
\fill[fill=p1col2] (3.6000000000000005,0) rectangle (4.6000000000000005,0.65);
\end{tikzpicture}
 & $2\lvec{34}1$ & \begin{tikzpicture}[scale=30/100]
\fill[fill=p1col2] (0,0) rectangle (1,0.65);
\fill[fill=p1col3] (1.2,0) rectangle (2.2,0.9);
\fill[fill=p1col4] (2.4000000000000004,0) rectangle (3.4000000000000004,1.15);
\fill[fill=p1col1] (3.6000000000000005,0) rectangle (4.6000000000000005,0.4);
\end{tikzpicture}
 \\
$\lvec{14}23$ & \begin{tikzpicture}[scale=30/100]
\fill[fill=p1col1] (0,0) rectangle (1,0.4);
\fill[fill=p1col4] (1.2,0) rectangle (2.2,1.15);
\fill[fill=p1col2] (2.4000000000000004,0) rectangle (3.4000000000000004,0.65);
\fill[fill=p1col3] (3.6000000000000005,0) rectangle (4.6000000000000005,0.9);
\end{tikzpicture}
 & $\lvec{34}12$ & \begin{tikzpicture}[scale=30/100]
\fill[fill=p1col3] (0,0) rectangle (1,0.9);
\fill[fill=p1col4] (1.2,0) rectangle (2.2,1.15);
\fill[fill=p1col1] (2.4000000000000004,0) rectangle (3.4000000000000004,0.4);
\fill[fill=p1col2] (3.6000000000000005,0) rectangle (4.6000000000000005,0.65);
\end{tikzpicture}
 & $\lvec{24}31$ & \begin{tikzpicture}[scale=30/100]
\fill[fill=p1col2] (0,0) rectangle (1,0.65);
\fill[fill=p1col4] (1.2,0) rectangle (2.2,1.15);
\fill[fill=p1col3] (2.4000000000000004,0) rectangle (3.4000000000000004,0.9);
\fill[fill=p1col1] (3.6000000000000005,0) rectangle (4.6000000000000005,0.4);
\end{tikzpicture}
 \\
$41\lvec{23}$ & \begin{tikzpicture}[scale=30/100]
\fill[fill=p1col4] (0,0) rectangle (1,1.15);
\fill[fill=p1col1] (1.2,0) rectangle (2.2,0.4);
\fill[fill=p1col2] (2.4000000000000004,0) rectangle (3.4000000000000004,0.65);
\fill[fill=p1col3] (3.6000000000000005,0) rectangle (4.6000000000000005,0.9);
\end{tikzpicture}
 & $43\lvec{12}$ & \begin{tikzpicture}[scale=30/100]
\fill[fill=p1col4] (0,0) rectangle (1,1.15);
\fill[fill=p1col3] (1.2,0) rectangle (2.2,0.9);
\fill[fill=p1col1] (2.4000000000000004,0) rectangle (3.4000000000000004,0.4);
\fill[fill=p1col2] (3.6000000000000005,0) rectangle (4.6000000000000005,0.65);
\end{tikzpicture}
 & $42\rvec{31}$ & \begin{tikzpicture}[scale=30/100]
\fill[fill=p1col4] (0,0) rectangle (1,1.15);
\fill[fill=p1col2] (1.2,0) rectangle (2.2,0.65);
\fill[fill=p1col3] (2.4000000000000004,0) rectangle (3.4000000000000004,0.9);
\fill[fill=p1col1] (3.6000000000000005,0) rectangle (4.6000000000000005,0.4);
\end{tikzpicture}
 \\
$\rvec{41}32$ & \begin{tikzpicture}[scale=30/100]
\fill[fill=p1col4] (0,0) rectangle (1,1.15);
\fill[fill=p1col1] (1.2,0) rectangle (2.2,0.4);
\fill[fill=p1col3] (2.4000000000000004,0) rectangle (3.4000000000000004,0.9);
\fill[fill=p1col2] (3.6000000000000005,0) rectangle (4.6000000000000005,0.65);
\end{tikzpicture}
 & $\rvec{43}21$ & \begin{tikzpicture}[scale=30/100]
\fill[fill=p1col4] (0,0) rectangle (1,1.15);
\fill[fill=p1col3] (1.2,0) rectangle (2.2,0.9);
\fill[fill=p1col2] (2.4000000000000004,0) rectangle (3.4000000000000004,0.65);
\fill[fill=p1col1] (3.6000000000000005,0) rectangle (4.6000000000000005,0.4);
\end{tikzpicture}
 & $\rvec{42}13$ & \begin{tikzpicture}[scale=30/100]
\fill[fill=p1col4] (0,0) rectangle (1,1.15);
\fill[fill=p1col2] (1.2,0) rectangle (2.2,0.65);
\fill[fill=p1col1] (2.4000000000000004,0) rectangle (3.4000000000000004,0.4);
\fill[fill=p1col3] (3.6000000000000005,0) rectangle (4.6000000000000005,0.9);
\end{tikzpicture}
 \\
$1\rvec{43}2$ & \begin{tikzpicture}[scale=30/100]
\fill[fill=p1col1] (0,0) rectangle (1,0.4);
\fill[fill=p1col4] (1.2,0) rectangle (2.2,1.15);
\fill[fill=p1col3] (2.4000000000000004,0) rectangle (3.4000000000000004,0.9);
\fill[fill=p1col2] (3.6000000000000005,0) rectangle (4.6000000000000005,0.65);
\end{tikzpicture}
 & $3\rvec{42}1$ & \begin{tikzpicture}[scale=30/100]
\fill[fill=p1col3] (0,0) rectangle (1,0.9);
\fill[fill=p1col4] (1.2,0) rectangle (2.2,1.15);
\fill[fill=p1col2] (2.4000000000000004,0) rectangle (3.4000000000000004,0.65);
\fill[fill=p1col1] (3.6000000000000005,0) rectangle (4.6000000000000005,0.4);
\end{tikzpicture}
 & $2\rvec{41}3$ & \begin{tikzpicture}[scale=30/100]
\fill[fill=p1col2] (0,0) rectangle (1,0.65);
\fill[fill=p1col4] (1.2,0) rectangle (2.2,1.15);
\fill[fill=p1col1] (2.4000000000000004,0) rectangle (3.4000000000000004,0.4);
\fill[fill=p1col3] (3.6000000000000005,0) rectangle (4.6000000000000005,0.9);
\end{tikzpicture}
 \\
$13\rvec{42}$ & \begin{tikzpicture}[scale=30/100]
\fill[fill=p1col1] (0,0) rectangle (1,0.4);
\fill[fill=p1col3] (1.2,0) rectangle (2.2,0.9);
\fill[fill=p1col4] (2.4000000000000004,0) rectangle (3.4000000000000004,1.15);
\fill[fill=p1col2] (3.6000000000000005,0) rectangle (4.6000000000000005,0.65);
\end{tikzpicture}
 & $32\rvec{41}$ & \begin{tikzpicture}[scale=30/100]
\fill[fill=p1col3] (0,0) rectangle (1,0.9);
\fill[fill=p1col2] (1.2,0) rectangle (2.2,0.65);
\fill[fill=p1col4] (2.4000000000000004,0) rectangle (3.4000000000000004,1.15);
\fill[fill=p1col1] (3.6000000000000005,0) rectangle (4.6000000000000005,0.4);
\end{tikzpicture}
 & $21\rvec{43}$ & \begin{tikzpicture}[scale=30/100]
\fill[fill=p1col2] (0,0) rectangle (1,0.65);
\fill[fill=p1col1] (1.2,0) rectangle (2.2,0.4);
\fill[fill=p1col4] (2.4000000000000004,0) rectangle (3.4000000000000004,1.15);
\fill[fill=p1col3] (3.6000000000000005,0) rectangle (4.6000000000000005,0.9);
\end{tikzpicture}
 \\
$\lvec{13}24$ & \begin{tikzpicture}[scale=30/100]
\fill[fill=p1col1] (0,0) rectangle (1,0.4);
\fill[fill=p1col3] (1.2,0) rectangle (2.2,0.9);
\fill[fill=p1col2] (2.4000000000000004,0) rectangle (3.4000000000000004,0.65);
\fill[fill=p1col4] (3.6000000000000005,0) rectangle (4.6000000000000005,1.15);
\end{tikzpicture}
 & $\rvec{32}14$ & \begin{tikzpicture}[scale=30/100]
\fill[fill=p1col3] (0,0) rectangle (1,0.9);
\fill[fill=p1col2] (1.2,0) rectangle (2.2,0.65);
\fill[fill=p1col1] (2.4000000000000004,0) rectangle (3.4000000000000004,0.4);
\fill[fill=p1col4] (3.6000000000000005,0) rectangle (4.6000000000000005,1.15);
\end{tikzpicture}
 & $\rvec{21}34$ & \begin{tikzpicture}[scale=30/100]
\fill[fill=p1col2] (0,0) rectangle (1,0.65);
\fill[fill=p1col1] (1.2,0) rectangle (2.2,0.4);
\fill[fill=p1col3] (2.4000000000000004,0) rectangle (3.4000000000000004,0.9);
\fill[fill=p1col4] (3.6000000000000005,0) rectangle (4.6000000000000005,1.15);
\end{tikzpicture}
 \\
\end{tabular}
\caption{All permutations of length~$n=4$ generated by Algorithm~J, coinciding with the Steinhaus-Johnson-Trotter ordering.
Read the figure column by column, from left to right.
Each arrows indicates an adjacent transposition (jump by 1~step) that creates the next permutation.
The ordering is cyclic, so the last transposition creates the first permutation.
}
\label{fig:perm4}
\end{figure}

\subsection{Binary strings (BRGC)}

Consider the set~$X_n$ of binary strings of length~$n-1$.
We map any binary string $x=x_2\cdots x_n$ to a permutation $f(x)\in S_n$ by setting $f(\varepsilon):=1$ and
\begin{equation*}
 f(x_2\cdots x_n):=\begin{cases}
 c_n\big(f(x_2\cdots x_{n-1})\big) & \text{ if } x_n=0, \\
 c_1\big(f(x_2\cdots x_{n-1})\big) & \text{ if } x_n=1,
                   \end{cases}
\end{equation*}
i.e., we build the permutation~$f(x)$ by inserting the values $i=2,\ldots,n$ one by one, either at the leftmost or rightmost position, depending on the bit~$x_i$; see Figure~\ref{fig:bits4} for examples.
Observe that $f(X_n)$ is exactly the set of permutations without peaks $P_n\seq S_n$ discussed previously in Section~\ref{sec:tree}, and a jump of the value~$i$ in the permutation translates to flipping the bit~$x_i$.
Moreover, $f^{-1}(J(P_n))$ is exactly the well-known reflected Gray code (BRGC) for binary strings of length~$n-1$~\cite{gray_1953}, which can be implemented efficiently~\cite{MR0424386}.
This ordering is shown in Figure~\ref{fig:bits4}.
Applying $f^{-1}$ to Algorithm~J yields the following simple greedy algorithm to describe the BRGC (see \cite{DBLP:conf/wads/Williams13}):
\textbf{J1.} Visit the all-zero string. \textbf{J2.} Flip the rightmost bit that yields a previously unvisited string; then visit this string and repeat~J2.

\begin{figure}
\begin{tabular}{ccc@{}c@{}c}
 & & \colorbox{p1col2!50}{\!$x_2$\!} & \colorbox{p1col3!50}{\!$x_3$\!} & \colorbox{p1col4!50}{\!$x_4$\!} \\ \cline{3-5}
$\lvec{1234}$ & \begin{tikzpicture}[scale=30/100]
\fill[fill=p1col1] (0,0) rectangle (1,0.4);
\fill[fill=p1col2] (1.2,0) rectangle (2.2,0.65);
\fill[fill=p1col3] (2.4000000000000004,0) rectangle (3.4000000000000004,0.9);
\fill[fill=p1col4] (3.6000000000000005,0) rectangle (4.6000000000000005,1.15);
\end{tikzpicture}
 & 0 & 0 & 0 \\
$4\lvec{123}$ & \begin{tikzpicture}[scale=30/100]
\fill[fill=p1col4] (0,0) rectangle (1,1.15);
\fill[fill=p1col1] (1.2,0) rectangle (2.2,0.4);
\fill[fill=p1col2] (2.4000000000000004,0) rectangle (3.4000000000000004,0.65);
\fill[fill=p1col3] (3.6000000000000005,0) rectangle (4.6000000000000005,0.9);
\end{tikzpicture}
 & 0 & 0 & 1 \\
$\rvec{4312}$ & \begin{tikzpicture}[scale=30/100]
\fill[fill=p1col4] (0,0) rectangle (1,1.15);
\fill[fill=p1col3] (1.2,0) rectangle (2.2,0.9);
\fill[fill=p1col1] (2.4000000000000004,0) rectangle (3.4000000000000004,0.4);
\fill[fill=p1col2] (3.6000000000000005,0) rectangle (4.6000000000000005,0.65);
\end{tikzpicture}
 & 0 & 1 & 1 \\
$3\lvec{12}4$ & \begin{tikzpicture}[scale=30/100]
\fill[fill=p1col3] (0,0) rectangle (1,0.9);
\fill[fill=p1col1] (1.2,0) rectangle (2.2,0.4);
\fill[fill=p1col2] (2.4000000000000004,0) rectangle (3.4000000000000004,0.65);
\fill[fill=p1col4] (3.6000000000000005,0) rectangle (4.6000000000000005,1.15);
\end{tikzpicture}
 & 0 & 1 & 0 \\
$\lvec{3214}$ & \begin{tikzpicture}[scale=30/100]
\fill[fill=p1col3] (0,0) rectangle (1,0.9);
\fill[fill=p1col2] (1.2,0) rectangle (2.2,0.65);
\fill[fill=p1col1] (2.4000000000000004,0) rectangle (3.4000000000000004,0.4);
\fill[fill=p1col4] (3.6000000000000005,0) rectangle (4.6000000000000005,1.15);
\end{tikzpicture}
 & 1 & 1 & 0 \\
$4\rvec{321}$ & \begin{tikzpicture}[scale=30/100]
\fill[fill=p1col4] (0,0) rectangle (1,1.15);
\fill[fill=p1col3] (1.2,0) rectangle (2.2,0.9);
\fill[fill=p1col2] (2.4000000000000004,0) rectangle (3.4000000000000004,0.65);
\fill[fill=p1col1] (3.6000000000000005,0) rectangle (4.6000000000000005,0.4);
\end{tikzpicture}
 & 1 & 1 & 1 \\
$\rvec{4213}$ & \begin{tikzpicture}[scale=30/100]
\fill[fill=p1col4] (0,0) rectangle (1,1.15);
\fill[fill=p1col2] (1.2,0) rectangle (2.2,0.65);
\fill[fill=p1col1] (2.4000000000000004,0) rectangle (3.4000000000000004,0.4);
\fill[fill=p1col3] (3.6000000000000005,0) rectangle (4.6000000000000005,0.9);
\end{tikzpicture}
 & 1 & 0 & 1 \\
$\rvec{21}34$ & \begin{tikzpicture}[scale=30/100]
\fill[fill=p1col2] (0,0) rectangle (1,0.65);
\fill[fill=p1col1] (1.2,0) rectangle (2.2,0.4);
\fill[fill=p1col3] (2.4000000000000004,0) rectangle (3.4000000000000004,0.9);
\fill[fill=p1col4] (3.6000000000000005,0) rectangle (4.6000000000000005,1.15);
\end{tikzpicture}
 & 1 & 0 & 0
\end{tabular}
\caption{Permutations without peaks of length~$n=4$ generated by Algorithm~J, and the resulting Gray code for binary strings, coinciding with the BRGC.
In this and subsequent figures, an arrow in a permutation indicates a jump, where the value below the tail of the arrow is the value that jumps, and the tip of the arrow shows the position of the value after the jump.}
\label{fig:bits4}
\end{figure}

\subsection{Binary trees (Lucas-Roelants van Baronaigien-Ruskey)}

Consider the set~$X_n$ of binary trees with $n$ nodes, labelled with $n$ distinct integers from~$[n]$, that have the search tree property, i.e., for every node, all nodes in the left subtree are smaller than all nodes in the right subtree.
We recursively map any such tree~$x$ with root node~$i$, left subtree~$x_L$, and right subtree~$x_R$ to a permutation~$f(x)\in S_n$ by setting $f(x):=(i,f(x_L),f(x_R))$ and $f(\emptyset):=\varepsilon$ if $x=\emptyset$ has no nodes; see Figure~\ref{fig:cat4} for examples.
By the search tree property, $f(X_n)$ is exactly the set of permutations that avoid the pattern~231, i.e., we have $f(X_n)=S_n(231)$, which we will prove to be a zigzag language.
A jump of the value~$i$ in the permutation translates to a tree rotation involving the node~$i$.
Specifically, a right jump of the value~$i$ in the permutation corresponds to a right rotation at node~$i$, and a left jump corresponds to a left rotation at the parent of node~$i$, where node~$i$ is the right child of this parent, and these two operations are inverse to each other.
Moreover, $f^{-1}(J(S_n(231)))$ is exactly the ordering of binary trees described by Lucas, Roelants van Baronaigien, and Ruskey~\cite{MR1239499}, which they showed can be implemented efficiently.
This ordering is shown in Figure~\ref{fig:cat4}.
Applying $f^{-1}$ to Algorithm~J yields the following simple greedy algorithm to describe this order (see \cite{DBLP:conf/wads/Williams13}):
\textbf{J1.} Visit the all-right tree, i.e., every node has exactly one child, namely a right child. \textbf{J2.} Perform a rotation involving the largest possible node that yields a previously unvisited tree; then visit this tree and repeat~J2.

Via standard bijections, binary trees are equivalent to many other Catalan families (such as triangulations, Dyck paths, etc.), so we also obtain Gray codes for all these other objects; see Figure~\ref{fig:cat4}.

\begin{figure}
\begin{tabular}{c@{\hskip 2mm}c@{\hskip 1mm}c@{\hskip 1mm}c@{\hskip 8mm}c@{\hskip 2mm}c@{\hskip 1mm}c@{\hskip 1mm}c}
\parbox{13mm}{$12\lvec{34}$ \\ \begin{tikzpicture}[scale=30/100]
\fill[fill=p1col1] (0,0) rectangle (1,0.4);
\fill[fill=p1col2] (1.2,0) rectangle (2.2,0.65);
\fill[fill=p1col3] (2.4000000000000004,0) rectangle (3.4000000000000004,0.9);
\fill[fill=p1col4] (3.6000000000000005,0) rectangle (4.6000000000000005,1.15);
\end{tikzpicture}
 \\[2mm]} & \begin{tikzpicture}[scale=30/100]
\draw[tree edge] (7.50,-3) -- (7.25,-4);
\draw[tree edge] (7.50,-3) -- (7.75,-4);
\draw[tree edge] (7.00,-2) -- (6.50,-3);
\draw[tree edge] (7.00,-2) -- (7.50,-3);
\draw[tree edge] (6.00,-1) -- (5.00,-2);
\draw[tree edge] (6.00,-1) -- (7.00,-2);
\draw[tree edge] (4.00,0) -- (2.00,-1);
\draw[tree edge] (4.00,0) -- (6.00,-1);
\node[tree vertex,fill=p1col4] at (7.50,-3) {};
\node[tree vertex,fill=p1col3] at (7.00,-2) {};
\node[tree vertex,fill=p1col2] at (6.00,-1) {};
\node[tree vertex,fill=p1col1] at (4.00,0) {};
\end{tikzpicture}
 & \begin{tikzpicture}[scale=65/100]
\fill[p1col1] (-1.00,0.00)--(-0.50,0.87)--(0.50,0.87)--cycle;
\fill[p1col2] (-0.50,-0.87)--(-1.00,0.00)--(0.50,0.87)--cycle;
\fill[p1col3] (0.50,-0.87)--(-0.50,-0.87)--(0.50,0.87)--cycle;
\fill[p1col4] (1.00,-0.00)--(0.50,-0.87)--(0.50,0.87)--cycle;
\node[triang vertex] at (-0.50,0.87) {};
\node[triang vertex] at (-1.00,0.00) {};
\node[triang vertex] at (-0.50,-0.87) {};
\node[triang vertex] at (0.50,-0.87) {};
\node[triang vertex] at (1.00,-0.00) {};
\node[triang vertex] at (0.50,0.87) {};
\end{tikzpicture}
 & \begin{tikzpicture}[scale=30/100]
\fill[p1col4] (6,0)--(7,1)--(7,1)--(8,0)--cycle;
\fill[p1col3] (4,0)--(5,1)--(5,1)--(6,0)--cycle;
\fill[p1col2] (2,0)--(3,1)--(3,1)--(4,0)--cycle;
\fill[p1col1] (0,0)--(1,1)--(1,1)--(2,0)--cycle;
\end{tikzpicture}
 &
\parbox{13mm}{$\lvec{3124}$ \\ \begin{tikzpicture}[scale=30/100]
\fill[fill=p1col3] (0,0) rectangle (1,0.9);
\fill[fill=p1col1] (1.2,0) rectangle (2.2,0.4);
\fill[fill=p1col2] (2.4000000000000004,0) rectangle (3.4000000000000004,0.65);
\fill[fill=p1col4] (3.6000000000000005,0) rectangle (4.6000000000000005,1.15);
\end{tikzpicture}
 \\[2mm]} & \begin{tikzpicture}[scale=30/100]
\draw[tree edge] (3.00,-2) -- (2.50,-3);
\draw[tree edge] (3.00,-2) -- (3.50,-3);
\draw[tree edge] (2.00,-1) -- (1.00,-2);
\draw[tree edge] (2.00,-1) -- (3.00,-2);
\draw[tree edge] (6.00,-1) -- (5.00,-2);
\draw[tree edge] (6.00,-1) -- (7.00,-2);
\draw[tree edge] (4.00,0) -- (2.00,-1);
\draw[tree edge] (4.00,0) -- (6.00,-1);
\node[tree vertex,fill=p1col2] at (3.00,-2) {};
\node[tree vertex,fill=p1col1] at (2.00,-1) {};
\node[tree vertex,fill=p1col4] at (6.00,-1) {};
\node[tree vertex,fill=p1col3] at (4.00,0) {};
\end{tikzpicture}
 & \begin{tikzpicture}[scale=65/100]
\fill[p1col3] (0.50,-0.87)--(-0.50,0.87)--(0.50,0.87)--cycle;
\fill[p1col1] (-1.00,0.00)--(-0.50,0.87)--(0.50,-0.87)--cycle;
\fill[p1col2] (-0.50,-0.87)--(-1.00,0.00)--(0.50,-0.87)--cycle;
\fill[p1col4] (1.00,-0.00)--(0.50,-0.87)--(0.50,0.87)--cycle;
\node[triang vertex] at (-0.50,0.87) {};
\node[triang vertex] at (-1.00,0.00) {};
\node[triang vertex] at (-0.50,-0.87) {};
\node[triang vertex] at (0.50,-0.87) {};
\node[triang vertex] at (1.00,-0.00) {};
\node[triang vertex] at (0.50,0.87) {};
\end{tikzpicture}
 & \begin{tikzpicture}[scale=30/100]
\fill[p1col2] (3,1)--(4,2)--(4,2)--(5,1)--cycle;
\fill[p1col1] (1,1)--(2,2)--(2,2)--(3,1)--cycle;
\fill[p1col4] (6,0)--(7,1)--(7,1)--(8,0)--cycle;
\fill[p1col3] (0,0)--(1,1)--(5,1)--(6,0)--cycle;
\end{tikzpicture}
 \vspace{-4mm} \\
\parbox{13mm}{$1\lvec{24}3$ \\ \begin{tikzpicture}[scale=30/100]
\fill[fill=p1col1] (0,0) rectangle (1,0.4);
\fill[fill=p1col2] (1.2,0) rectangle (2.2,0.65);
\fill[fill=p1col4] (2.4000000000000004,0) rectangle (3.4000000000000004,1.15);
\fill[fill=p1col3] (3.6000000000000005,0) rectangle (4.6000000000000005,0.9);
\end{tikzpicture}
 \\[2mm]} & \begin{tikzpicture}[scale=30/100]
\draw[tree edge] (6.50,-3) -- (6.25,-4);
\draw[tree edge] (6.50,-3) -- (6.75,-4);
\draw[tree edge] (7.00,-2) -- (6.50,-3);
\draw[tree edge] (7.00,-2) -- (7.50,-3);
\draw[tree edge] (6.00,-1) -- (5.00,-2);
\draw[tree edge] (6.00,-1) -- (7.00,-2);
\draw[tree edge] (4.00,0) -- (2.00,-1);
\draw[tree edge] (4.00,0) -- (6.00,-1);
\node[tree vertex,fill=p1col3] at (6.50,-3) {};
\node[tree vertex,fill=p1col4] at (7.00,-2) {};
\node[tree vertex,fill=p1col2] at (6.00,-1) {};
\node[tree vertex,fill=p1col1] at (4.00,0) {};
\end{tikzpicture}
 & \begin{tikzpicture}[scale=65/100]
\fill[p1col1] (-1.00,0.00)--(-0.50,0.87)--(0.50,0.87)--cycle;
\fill[p1col2] (-0.50,-0.87)--(-1.00,0.00)--(0.50,0.87)--cycle;
\fill[p1col4] (1.00,-0.00)--(-0.50,-0.87)--(0.50,0.87)--cycle;
\fill[p1col3] (0.50,-0.87)--(-0.50,-0.87)--(1.00,-0.00)--cycle;
\node[triang vertex] at (-0.50,0.87) {};
\node[triang vertex] at (-1.00,0.00) {};
\node[triang vertex] at (-0.50,-0.87) {};
\node[triang vertex] at (0.50,-0.87) {};
\node[triang vertex] at (1.00,-0.00) {};
\node[triang vertex] at (0.50,0.87) {};
\end{tikzpicture}
 & \begin{tikzpicture}[scale=30/100]
\fill[p1col3] (5,1)--(6,2)--(6,2)--(7,1)--cycle;
\fill[p1col4] (4,0)--(5,1)--(7,1)--(8,0)--cycle;
\fill[p1col2] (2,0)--(3,1)--(3,1)--(4,0)--cycle;
\fill[p1col1] (0,0)--(1,1)--(1,1)--(2,0)--cycle;
\end{tikzpicture}
 &
\parbox{13mm}{$43\lvec{12}$ \\ \begin{tikzpicture}[scale=30/100]
\fill[fill=p1col4] (0,0) rectangle (1,1.15);
\fill[fill=p1col3] (1.2,0) rectangle (2.2,0.9);
\fill[fill=p1col1] (2.4000000000000004,0) rectangle (3.4000000000000004,0.4);
\fill[fill=p1col2] (3.6000000000000005,0) rectangle (4.6000000000000005,0.65);
\end{tikzpicture}
 \\[2mm]} & \begin{tikzpicture}[scale=30/100]
\draw[tree edge] (1.50,-3) -- (1.25,-4);
\draw[tree edge] (1.50,-3) -- (1.75,-4);
\draw[tree edge] (1.00,-2) -- (0.50,-3);
\draw[tree edge] (1.00,-2) -- (1.50,-3);
\draw[tree edge] (2.00,-1) -- (1.00,-2);
\draw[tree edge] (2.00,-1) -- (3.00,-2);
\draw[tree edge] (4.00,0) -- (2.00,-1);
\draw[tree edge] (4.00,0) -- (6.00,-1);
\node[tree vertex,fill=p1col2] at (1.50,-3) {};
\node[tree vertex,fill=p1col1] at (1.00,-2) {};
\node[tree vertex,fill=p1col3] at (2.00,-1) {};
\node[tree vertex,fill=p1col4] at (4.00,0) {};
\end{tikzpicture}
 & \begin{tikzpicture}[scale=65/100]
\fill[p1col4] (1.00,-0.00)--(-0.50,0.87)--(0.50,0.87)--cycle;
\fill[p1col3] (0.50,-0.87)--(-0.50,0.87)--(1.00,-0.00)--cycle;
\fill[p1col1] (-1.00,0.00)--(-0.50,0.87)--(0.50,-0.87)--cycle;
\fill[p1col2] (-0.50,-0.87)--(-1.00,0.00)--(0.50,-0.87)--cycle;
\node[triang vertex] at (-0.50,0.87) {};
\node[triang vertex] at (-1.00,0.00) {};
\node[triang vertex] at (-0.50,-0.87) {};
\node[triang vertex] at (0.50,-0.87) {};
\node[triang vertex] at (1.00,-0.00) {};
\node[triang vertex] at (0.50,0.87) {};
\end{tikzpicture}
 & \begin{tikzpicture}[scale=30/100]
\fill[p1col2] (4,2)--(5,3)--(5,3)--(6,2)--cycle;
\fill[p1col1] (2,2)--(3,3)--(3,3)--(4,2)--cycle;
\fill[p1col3] (1,1)--(2,2)--(6,2)--(7,1)--cycle;
\fill[p1col4] (0,0)--(1,1)--(7,1)--(8,0)--cycle;
\end{tikzpicture}
 \vspace{-4mm} \\
\parbox{13mm}{$\lvec{14}23$ \\ \begin{tikzpicture}[scale=30/100]
\fill[fill=p1col1] (0,0) rectangle (1,0.4);
\fill[fill=p1col4] (1.2,0) rectangle (2.2,1.15);
\fill[fill=p1col2] (2.4000000000000004,0) rectangle (3.4000000000000004,0.65);
\fill[fill=p1col3] (3.6000000000000005,0) rectangle (4.6000000000000005,0.9);
\end{tikzpicture}
 \\[2mm]} & \begin{tikzpicture}[scale=30/100]
\draw[tree edge] (5.50,-3) -- (5.25,-4);
\draw[tree edge] (5.50,-3) -- (5.75,-4);
\draw[tree edge] (5.00,-2) -- (4.50,-3);
\draw[tree edge] (5.00,-2) -- (5.50,-3);
\draw[tree edge] (6.00,-1) -- (5.00,-2);
\draw[tree edge] (6.00,-1) -- (7.00,-2);
\draw[tree edge] (4.00,0) -- (2.00,-1);
\draw[tree edge] (4.00,0) -- (6.00,-1);
\node[tree vertex,fill=p1col3] at (5.50,-3) {};
\node[tree vertex,fill=p1col2] at (5.00,-2) {};
\node[tree vertex,fill=p1col4] at (6.00,-1) {};
\node[tree vertex,fill=p1col1] at (4.00,0) {};
\end{tikzpicture}
 & \begin{tikzpicture}[scale=65/100]
\fill[p1col1] (-1.00,0.00)--(-0.50,0.87)--(0.50,0.87)--cycle;
\fill[p1col4] (1.00,-0.00)--(-1.00,0.00)--(0.50,0.87)--cycle;
\fill[p1col2] (-0.50,-0.87)--(-1.00,0.00)--(1.00,-0.00)--cycle;
\fill[p1col3] (0.50,-0.87)--(-0.50,-0.87)--(1.00,-0.00)--cycle;
\node[triang vertex] at (-0.50,0.87) {};
\node[triang vertex] at (-1.00,0.00) {};
\node[triang vertex] at (-0.50,-0.87) {};
\node[triang vertex] at (0.50,-0.87) {};
\node[triang vertex] at (1.00,-0.00) {};
\node[triang vertex] at (0.50,0.87) {};
\end{tikzpicture}
 & \begin{tikzpicture}[scale=30/100]
\fill[p1col3] (5,1)--(6,2)--(6,2)--(7,1)--cycle;
\fill[p1col2] (3,1)--(4,2)--(4,2)--(5,1)--cycle;
\fill[p1col4] (2,0)--(3,1)--(7,1)--(8,0)--cycle;
\fill[p1col1] (0,0)--(1,1)--(1,1)--(2,0)--cycle;
\end{tikzpicture}
 &
\parbox{13mm}{$\rvec{4321}$ \\ \begin{tikzpicture}[scale=30/100]
\fill[fill=p1col4] (0,0) rectangle (1,1.15);
\fill[fill=p1col3] (1.2,0) rectangle (2.2,0.9);
\fill[fill=p1col2] (2.4000000000000004,0) rectangle (3.4000000000000004,0.65);
\fill[fill=p1col1] (3.6000000000000005,0) rectangle (4.6000000000000005,0.4);
\end{tikzpicture}
 \\[2mm]} & \begin{tikzpicture}[scale=30/100]
\draw[tree edge] (0.50,-3) -- (0.25,-4);
\draw[tree edge] (0.50,-3) -- (0.75,-4);
\draw[tree edge] (1.00,-2) -- (0.50,-3);
\draw[tree edge] (1.00,-2) -- (1.50,-3);
\draw[tree edge] (2.00,-1) -- (1.00,-2);
\draw[tree edge] (2.00,-1) -- (3.00,-2);
\draw[tree edge] (4.00,0) -- (2.00,-1);
\draw[tree edge] (4.00,0) -- (6.00,-1);
\node[tree vertex,fill=p1col1] at (0.50,-3) {};
\node[tree vertex,fill=p1col2] at (1.00,-2) {};
\node[tree vertex,fill=p1col3] at (2.00,-1) {};
\node[tree vertex,fill=p1col4] at (4.00,0) {};
\end{tikzpicture}
 & \begin{tikzpicture}[scale=65/100]
\fill[p1col4] (1.00,-0.00)--(-0.50,0.87)--(0.50,0.87)--cycle;
\fill[p1col3] (0.50,-0.87)--(-0.50,0.87)--(1.00,-0.00)--cycle;
\fill[p1col2] (-0.50,-0.87)--(-0.50,0.87)--(0.50,-0.87)--cycle;
\fill[p1col1] (-1.00,0.00)--(-0.50,0.87)--(-0.50,-0.87)--cycle;
\node[triang vertex] at (-0.50,0.87) {};
\node[triang vertex] at (-1.00,0.00) {};
\node[triang vertex] at (-0.50,-0.87) {};
\node[triang vertex] at (0.50,-0.87) {};
\node[triang vertex] at (1.00,-0.00) {};
\node[triang vertex] at (0.50,0.87) {};
\end{tikzpicture}
 & \begin{tikzpicture}[scale=30/100]
\fill[p1col1] (3,3)--(4,4)--(4,4)--(5,3)--cycle;
\fill[p1col2] (2,2)--(3,3)--(5,3)--(6,2)--cycle;
\fill[p1col3] (1,1)--(2,2)--(6,2)--(7,1)--cycle;
\fill[p1col4] (0,0)--(1,1)--(7,1)--(8,0)--cycle;
\end{tikzpicture}
 \vspace{-4mm} \\
\parbox{13mm}{$41\lvec{23}$ \\ \begin{tikzpicture}[scale=30/100]
\fill[fill=p1col4] (0,0) rectangle (1,1.15);
\fill[fill=p1col1] (1.2,0) rectangle (2.2,0.4);
\fill[fill=p1col2] (2.4000000000000004,0) rectangle (3.4000000000000004,0.65);
\fill[fill=p1col3] (3.6000000000000005,0) rectangle (4.6000000000000005,0.9);
\end{tikzpicture}
 \\[2mm]} & \begin{tikzpicture}[scale=30/100]
\draw[tree edge] (3.50,-3) -- (3.25,-4);
\draw[tree edge] (3.50,-3) -- (3.75,-4);
\draw[tree edge] (3.00,-2) -- (2.50,-3);
\draw[tree edge] (3.00,-2) -- (3.50,-3);
\draw[tree edge] (2.00,-1) -- (1.00,-2);
\draw[tree edge] (2.00,-1) -- (3.00,-2);
\draw[tree edge] (4.00,0) -- (2.00,-1);
\draw[tree edge] (4.00,0) -- (6.00,-1);
\node[tree vertex,fill=p1col3] at (3.50,-3) {};
\node[tree vertex,fill=p1col2] at (3.00,-2) {};
\node[tree vertex,fill=p1col1] at (2.00,-1) {};
\node[tree vertex,fill=p1col4] at (4.00,0) {};
\end{tikzpicture}
 & \begin{tikzpicture}[scale=65/100]
\fill[p1col4] (1.00,-0.00)--(-0.50,0.87)--(0.50,0.87)--cycle;
\fill[p1col1] (-1.00,0.00)--(-0.50,0.87)--(1.00,-0.00)--cycle;
\fill[p1col2] (-0.50,-0.87)--(-1.00,0.00)--(1.00,-0.00)--cycle;
\fill[p1col3] (0.50,-0.87)--(-0.50,-0.87)--(1.00,-0.00)--cycle;
\node[triang vertex] at (-0.50,0.87) {};
\node[triang vertex] at (-1.00,0.00) {};
\node[triang vertex] at (-0.50,-0.87) {};
\node[triang vertex] at (0.50,-0.87) {};
\node[triang vertex] at (1.00,-0.00) {};
\node[triang vertex] at (0.50,0.87) {};
\end{tikzpicture}
 & \begin{tikzpicture}[scale=30/100]
\fill[p1col3] (5,1)--(6,2)--(6,2)--(7,1)--cycle;
\fill[p1col2] (3,1)--(4,2)--(4,2)--(5,1)--cycle;
\fill[p1col1] (1,1)--(2,2)--(2,2)--(3,1)--cycle;
\fill[p1col4] (0,0)--(1,1)--(7,1)--(8,0)--cycle;
\end{tikzpicture}
 &
\parbox{13mm}{$\rvec{321}4$ \\ \begin{tikzpicture}[scale=30/100]
\fill[fill=p1col3] (0,0) rectangle (1,0.9);
\fill[fill=p1col2] (1.2,0) rectangle (2.2,0.65);
\fill[fill=p1col1] (2.4000000000000004,0) rectangle (3.4000000000000004,0.4);
\fill[fill=p1col4] (3.6000000000000005,0) rectangle (4.6000000000000005,1.15);
\end{tikzpicture}
 \\[2mm]} & \begin{tikzpicture}[scale=30/100]
\draw[tree edge] (1.00,-2) -- (0.50,-3);
\draw[tree edge] (1.00,-2) -- (1.50,-3);
\draw[tree edge] (2.00,-1) -- (1.00,-2);
\draw[tree edge] (2.00,-1) -- (3.00,-2);
\draw[tree edge] (6.00,-1) -- (5.00,-2);
\draw[tree edge] (6.00,-1) -- (7.00,-2);
\draw[tree edge] (4.00,0) -- (2.00,-1);
\draw[tree edge] (4.00,0) -- (6.00,-1);
\node[tree vertex,fill=p1col1] at (1.00,-2) {};
\node[tree vertex,fill=p1col2] at (2.00,-1) {};
\node[tree vertex,fill=p1col4] at (6.00,-1) {};
\node[tree vertex,fill=p1col3] at (4.00,0) {};
\end{tikzpicture}
 & \begin{tikzpicture}[scale=65/100]
\fill[p1col3] (0.50,-0.87)--(-0.50,0.87)--(0.50,0.87)--cycle;
\fill[p1col2] (-0.50,-0.87)--(-0.50,0.87)--(0.50,-0.87)--cycle;
\fill[p1col1] (-1.00,0.00)--(-0.50,0.87)--(-0.50,-0.87)--cycle;
\fill[p1col4] (1.00,-0.00)--(0.50,-0.87)--(0.50,0.87)--cycle;
\node[triang vertex] at (-0.50,0.87) {};
\node[triang vertex] at (-1.00,0.00) {};
\node[triang vertex] at (-0.50,-0.87) {};
\node[triang vertex] at (0.50,-0.87) {};
\node[triang vertex] at (1.00,-0.00) {};
\node[triang vertex] at (0.50,0.87) {};
\end{tikzpicture}
 & \begin{tikzpicture}[scale=30/100]
\fill[p1col1] (2,2)--(3,3)--(3,3)--(4,2)--cycle;
\fill[p1col2] (1,1)--(2,2)--(4,2)--(5,1)--cycle;
\fill[p1col4] (6,0)--(7,1)--(7,1)--(8,0)--cycle;
\fill[p1col3] (0,0)--(1,1)--(5,1)--(6,0)--cycle;
\end{tikzpicture}
 \vspace{-4mm} \\
\parbox{13mm}{$\rvec{41}32$ \\ \begin{tikzpicture}[scale=30/100]
\fill[fill=p1col4] (0,0) rectangle (1,1.15);
\fill[fill=p1col1] (1.2,0) rectangle (2.2,0.4);
\fill[fill=p1col3] (2.4000000000000004,0) rectangle (3.4000000000000004,0.9);
\fill[fill=p1col2] (3.6000000000000005,0) rectangle (4.6000000000000005,0.65);
\end{tikzpicture}
 \\[2mm]} & \begin{tikzpicture}[scale=30/100]
\draw[tree edge] (2.50,-3) -- (2.25,-4);
\draw[tree edge] (2.50,-3) -- (2.75,-4);
\draw[tree edge] (3.00,-2) -- (2.50,-3);
\draw[tree edge] (3.00,-2) -- (3.50,-3);
\draw[tree edge] (2.00,-1) -- (1.00,-2);
\draw[tree edge] (2.00,-1) -- (3.00,-2);
\draw[tree edge] (4.00,0) -- (2.00,-1);
\draw[tree edge] (4.00,0) -- (6.00,-1);
\node[tree vertex,fill=p1col2] at (2.50,-3) {};
\node[tree vertex,fill=p1col3] at (3.00,-2) {};
\node[tree vertex,fill=p1col1] at (2.00,-1) {};
\node[tree vertex,fill=p1col4] at (4.00,0) {};
\end{tikzpicture}
 & \begin{tikzpicture}[scale=65/100]
\fill[p1col4] (1.00,-0.00)--(-0.50,0.87)--(0.50,0.87)--cycle;
\fill[p1col1] (-1.00,0.00)--(-0.50,0.87)--(1.00,-0.00)--cycle;
\fill[p1col3] (0.50,-0.87)--(-1.00,0.00)--(1.00,-0.00)--cycle;
\fill[p1col2] (-0.50,-0.87)--(-1.00,0.00)--(0.50,-0.87)--cycle;
\node[triang vertex] at (-0.50,0.87) {};
\node[triang vertex] at (-1.00,0.00) {};
\node[triang vertex] at (-0.50,-0.87) {};
\node[triang vertex] at (0.50,-0.87) {};
\node[triang vertex] at (1.00,-0.00) {};
\node[triang vertex] at (0.50,0.87) {};
\end{tikzpicture}
 & \begin{tikzpicture}[scale=30/100]
\fill[p1col2] (4,2)--(5,3)--(5,3)--(6,2)--cycle;
\fill[p1col3] (3,1)--(4,2)--(6,2)--(7,1)--cycle;
\fill[p1col1] (1,1)--(2,2)--(2,2)--(3,1)--cycle;
\fill[p1col4] (0,0)--(1,1)--(7,1)--(8,0)--cycle;
\end{tikzpicture}
 &
\parbox{13mm}{$21\lvec{34}$ \\ \begin{tikzpicture}[scale=30/100]
\fill[fill=p1col2] (0,0) rectangle (1,0.65);
\fill[fill=p1col1] (1.2,0) rectangle (2.2,0.4);
\fill[fill=p1col3] (2.4000000000000004,0) rectangle (3.4000000000000004,0.9);
\fill[fill=p1col4] (3.6000000000000005,0) rectangle (4.6000000000000005,1.15);
\end{tikzpicture}
 \\[2mm]} & \begin{tikzpicture}[scale=30/100]
\draw[tree edge] (2.00,-1) -- (1.00,-2);
\draw[tree edge] (2.00,-1) -- (3.00,-2);
\draw[tree edge] (7.00,-2) -- (6.50,-3);
\draw[tree edge] (7.00,-2) -- (7.50,-3);
\draw[tree edge] (6.00,-1) -- (5.00,-2);
\draw[tree edge] (6.00,-1) -- (7.00,-2);
\draw[tree edge] (4.00,0) -- (2.00,-1);
\draw[tree edge] (4.00,0) -- (6.00,-1);
\node[tree vertex,fill=p1col1] at (2.00,-1) {};
\node[tree vertex,fill=p1col4] at (7.00,-2) {};
\node[tree vertex,fill=p1col3] at (6.00,-1) {};
\node[tree vertex,fill=p1col2] at (4.00,0) {};
\end{tikzpicture}
 & \begin{tikzpicture}[scale=65/100]
\fill[p1col2] (-0.50,-0.87)--(-0.50,0.87)--(0.50,0.87)--cycle;
\fill[p1col1] (-1.00,0.00)--(-0.50,0.87)--(-0.50,-0.87)--cycle;
\fill[p1col3] (0.50,-0.87)--(-0.50,-0.87)--(0.50,0.87)--cycle;
\fill[p1col4] (1.00,-0.00)--(0.50,-0.87)--(0.50,0.87)--cycle;
\node[triang vertex] at (-0.50,0.87) {};
\node[triang vertex] at (-1.00,0.00) {};
\node[triang vertex] at (-0.50,-0.87) {};
\node[triang vertex] at (0.50,-0.87) {};
\node[triang vertex] at (1.00,-0.00) {};
\node[triang vertex] at (0.50,0.87) {};
\end{tikzpicture}
 & \begin{tikzpicture}[scale=30/100]
\fill[p1col1] (1,1)--(2,2)--(2,2)--(3,1)--cycle;
\fill[p1col4] (6,0)--(7,1)--(7,1)--(8,0)--cycle;
\fill[p1col3] (4,0)--(5,1)--(5,1)--(6,0)--cycle;
\fill[p1col2] (0,0)--(1,1)--(3,1)--(4,0)--cycle;
\end{tikzpicture}
 \vspace{-4mm} \\
\parbox{13mm}{$1\rvec{432}$ \\ \begin{tikzpicture}[scale=30/100]
\fill[fill=p1col1] (0,0) rectangle (1,0.4);
\fill[fill=p1col4] (1.2,0) rectangle (2.2,1.15);
\fill[fill=p1col3] (2.4000000000000004,0) rectangle (3.4000000000000004,0.9);
\fill[fill=p1col2] (3.6000000000000005,0) rectangle (4.6000000000000005,0.65);
\end{tikzpicture}
 \\[2mm]} & \begin{tikzpicture}[scale=30/100]
\draw[tree edge] (4.50,-3) -- (4.25,-4);
\draw[tree edge] (4.50,-3) -- (4.75,-4);
\draw[tree edge] (5.00,-2) -- (4.50,-3);
\draw[tree edge] (5.00,-2) -- (5.50,-3);
\draw[tree edge] (6.00,-1) -- (5.00,-2);
\draw[tree edge] (6.00,-1) -- (7.00,-2);
\draw[tree edge] (4.00,0) -- (2.00,-1);
\draw[tree edge] (4.00,0) -- (6.00,-1);
\node[tree vertex,fill=p1col2] at (4.50,-3) {};
\node[tree vertex,fill=p1col3] at (5.00,-2) {};
\node[tree vertex,fill=p1col4] at (6.00,-1) {};
\node[tree vertex,fill=p1col1] at (4.00,0) {};
\end{tikzpicture}
 & \begin{tikzpicture}[scale=65/100]
\fill[p1col1] (-1.00,0.00)--(-0.50,0.87)--(0.50,0.87)--cycle;
\fill[p1col4] (1.00,-0.00)--(-1.00,0.00)--(0.50,0.87)--cycle;
\fill[p1col3] (0.50,-0.87)--(-1.00,0.00)--(1.00,-0.00)--cycle;
\fill[p1col2] (-0.50,-0.87)--(-1.00,0.00)--(0.50,-0.87)--cycle;
\node[triang vertex] at (-0.50,0.87) {};
\node[triang vertex] at (-1.00,0.00) {};
\node[triang vertex] at (-0.50,-0.87) {};
\node[triang vertex] at (0.50,-0.87) {};
\node[triang vertex] at (1.00,-0.00) {};
\node[triang vertex] at (0.50,0.87) {};
\end{tikzpicture}
 & \begin{tikzpicture}[scale=30/100]
\fill[p1col2] (4,2)--(5,3)--(5,3)--(6,2)--cycle;
\fill[p1col3] (3,1)--(4,2)--(6,2)--(7,1)--cycle;
\fill[p1col4] (2,0)--(3,1)--(7,1)--(8,0)--cycle;
\fill[p1col1] (0,0)--(1,1)--(1,1)--(2,0)--cycle;
\end{tikzpicture}
 &
\parbox{13mm}{$\lvec{214}3$ \\ \begin{tikzpicture}[scale=30/100]
\fill[fill=p1col2] (0,0) rectangle (1,0.65);
\fill[fill=p1col1] (1.2,0) rectangle (2.2,0.4);
\fill[fill=p1col4] (2.4000000000000004,0) rectangle (3.4000000000000004,1.15);
\fill[fill=p1col3] (3.6000000000000005,0) rectangle (4.6000000000000005,0.9);
\end{tikzpicture}
 \\[2mm]} & \begin{tikzpicture}[scale=30/100]
\draw[tree edge] (2.00,-1) -- (1.00,-2);
\draw[tree edge] (2.00,-1) -- (3.00,-2);
\draw[tree edge] (5.00,-2) -- (4.50,-3);
\draw[tree edge] (5.00,-2) -- (5.50,-3);
\draw[tree edge] (6.00,-1) -- (5.00,-2);
\draw[tree edge] (6.00,-1) -- (7.00,-2);
\draw[tree edge] (4.00,0) -- (2.00,-1);
\draw[tree edge] (4.00,0) -- (6.00,-1);
\node[tree vertex,fill=p1col1] at (2.00,-1) {};
\node[tree vertex,fill=p1col3] at (5.00,-2) {};
\node[tree vertex,fill=p1col4] at (6.00,-1) {};
\node[tree vertex,fill=p1col2] at (4.00,0) {};
\end{tikzpicture}
 & \begin{tikzpicture}[scale=65/100]
\fill[p1col1] (-1.00,0.00)--(-0.50,0.87)--(0.50,0.87)--cycle;
\fill[p1col4] (1.00,-0.00)--(-1.00,0.00)--(0.50,0.87)--cycle;
\fill[p1col3] (0.50,-0.87)--(-1.00,0.00)--(1.00,-0.00)--cycle;
\fill[p1col2] (-0.50,-0.87)--(-1.00,0.00)--(0.50,-0.87)--cycle;
\node[triang vertex] at (-0.50,0.87) {};
\node[triang vertex] at (-1.00,0.00) {};
\node[triang vertex] at (-0.50,-0.87) {};
\node[triang vertex] at (0.50,-0.87) {};
\node[triang vertex] at (1.00,-0.00) {};
\node[triang vertex] at (0.50,0.87) {};
\end{tikzpicture}
 & \begin{tikzpicture}[scale=30/100]
\fill[p1col1] (1,1)--(2,2)--(2,2)--(3,1)--cycle;
\fill[p1col3] (5,1)--(6,2)--(6,2)--(7,1)--cycle;
\fill[p1col4] (4,0)--(5,1)--(7,1)--(8,0)--cycle;
\fill[p1col2] (0,0)--(1,1)--(3,1)--(4,0)--cycle;
\end{tikzpicture}
 \vspace{-4mm} \\
\parbox{13mm}{$\lvec{13}24$ \\ \begin{tikzpicture}[scale=30/100]
\fill[fill=p1col1] (0,0) rectangle (1,0.4);
\fill[fill=p1col3] (1.2,0) rectangle (2.2,0.9);
\fill[fill=p1col2] (2.4000000000000004,0) rectangle (3.4000000000000004,0.65);
\fill[fill=p1col4] (3.6000000000000005,0) rectangle (4.6000000000000005,1.15);
\end{tikzpicture}
 \\[2mm]} & \begin{tikzpicture}[scale=30/100]
\draw[tree edge] (5.00,-2) -- (4.50,-3);
\draw[tree edge] (5.00,-2) -- (5.50,-3);
\draw[tree edge] (7.00,-2) -- (6.50,-3);
\draw[tree edge] (7.00,-2) -- (7.50,-3);
\draw[tree edge] (6.00,-1) -- (5.00,-2);
\draw[tree edge] (6.00,-1) -- (7.00,-2);
\draw[tree edge] (4.00,0) -- (2.00,-1);
\draw[tree edge] (4.00,0) -- (6.00,-1);
\node[tree vertex,fill=p1col2] at (5.00,-2) {};
\node[tree vertex,fill=p1col4] at (7.00,-2) {};
\node[tree vertex,fill=p1col3] at (6.00,-1) {};
\node[tree vertex,fill=p1col1] at (4.00,0) {};
\end{tikzpicture}
 & \begin{tikzpicture}[scale=65/100]
\fill[p1col1] (-1.00,0.00)--(-0.50,0.87)--(0.50,0.87)--cycle;
\fill[p1col3] (0.50,-0.87)--(-1.00,0.00)--(0.50,0.87)--cycle;
\fill[p1col2] (-0.50,-0.87)--(-1.00,0.00)--(0.50,-0.87)--cycle;
\fill[p1col4] (1.00,-0.00)--(0.50,-0.87)--(0.50,0.87)--cycle;
\node[triang vertex] at (-0.50,0.87) {};
\node[triang vertex] at (-1.00,0.00) {};
\node[triang vertex] at (-0.50,-0.87) {};
\node[triang vertex] at (0.50,-0.87) {};
\node[triang vertex] at (1.00,-0.00) {};
\node[triang vertex] at (0.50,0.87) {};
\end{tikzpicture}
 & \begin{tikzpicture}[scale=30/100]
\fill[p1col2] (3,1)--(4,2)--(4,2)--(5,1)--cycle;
\fill[p1col4] (6,0)--(7,1)--(7,1)--(8,0)--cycle;
\fill[p1col3] (2,0)--(3,1)--(5,1)--(6,0)--cycle;
\fill[p1col1] (0,0)--(1,1)--(1,1)--(2,0)--cycle;
\end{tikzpicture}
 &
\parbox{13mm}{$4213$ \\ \begin{tikzpicture}[scale=30/100]
\fill[fill=p1col4] (0,0) rectangle (1,1.15);
\fill[fill=p1col2] (1.2,0) rectangle (2.2,0.65);
\fill[fill=p1col1] (2.4000000000000004,0) rectangle (3.4000000000000004,0.4);
\fill[fill=p1col3] (3.6000000000000005,0) rectangle (4.6000000000000005,0.9);
\end{tikzpicture}
 \\[2mm]} & \begin{tikzpicture}[scale=30/100]
\draw[tree edge] (1.00,-2) -- (0.50,-3);
\draw[tree edge] (1.00,-2) -- (1.50,-3);
\draw[tree edge] (3.00,-2) -- (2.50,-3);
\draw[tree edge] (3.00,-2) -- (3.50,-3);
\draw[tree edge] (2.00,-1) -- (1.00,-2);
\draw[tree edge] (2.00,-1) -- (3.00,-2);
\draw[tree edge] (4.00,0) -- (2.00,-1);
\draw[tree edge] (4.00,0) -- (6.00,-1);
\node[tree vertex,fill=p1col1] at (1.00,-2) {};
\node[tree vertex,fill=p1col3] at (3.00,-2) {};
\node[tree vertex,fill=p1col2] at (2.00,-1) {};
\node[tree vertex,fill=p1col4] at (4.00,0) {};
\end{tikzpicture}
 & \begin{tikzpicture}[scale=65/100]
\fill[p1col4] (1.00,-0.00)--(-0.50,0.87)--(0.50,0.87)--cycle;
\fill[p1col2] (-0.50,-0.87)--(-0.50,0.87)--(1.00,-0.00)--cycle;
\fill[p1col1] (-1.00,0.00)--(-0.50,0.87)--(-0.50,-0.87)--cycle;
\fill[p1col3] (0.50,-0.87)--(-0.50,-0.87)--(1.00,-0.00)--cycle;
\node[triang vertex] at (-0.50,0.87) {};
\node[triang vertex] at (-1.00,0.00) {};
\node[triang vertex] at (-0.50,-0.87) {};
\node[triang vertex] at (0.50,-0.87) {};
\node[triang vertex] at (1.00,-0.00) {};
\node[triang vertex] at (0.50,0.87) {};
\end{tikzpicture}
 & \begin{tikzpicture}[scale=30/100]
\fill[p1col1] (2,2)--(3,3)--(3,3)--(4,2)--cycle;
\fill[p1col3] (5,1)--(6,2)--(6,2)--(7,1)--cycle;
\fill[p1col2] (1,1)--(2,2)--(4,2)--(5,1)--cycle;
\fill[p1col4] (0,0)--(1,1)--(7,1)--(8,0)--cycle;
\end{tikzpicture}
 \vspace{-4mm} \\
\end{tabular}
\caption{231-avoiding permutations of length $n=4$ generated by Algorithm~J and the resulting Gray codes for Catalan families (binary trees, triangulations, Dyck paths), coinciding with the ordering described by Lucas, Roelants van Baronaigien, and Ruskey.}
\label{fig:cat4}
\end{figure}
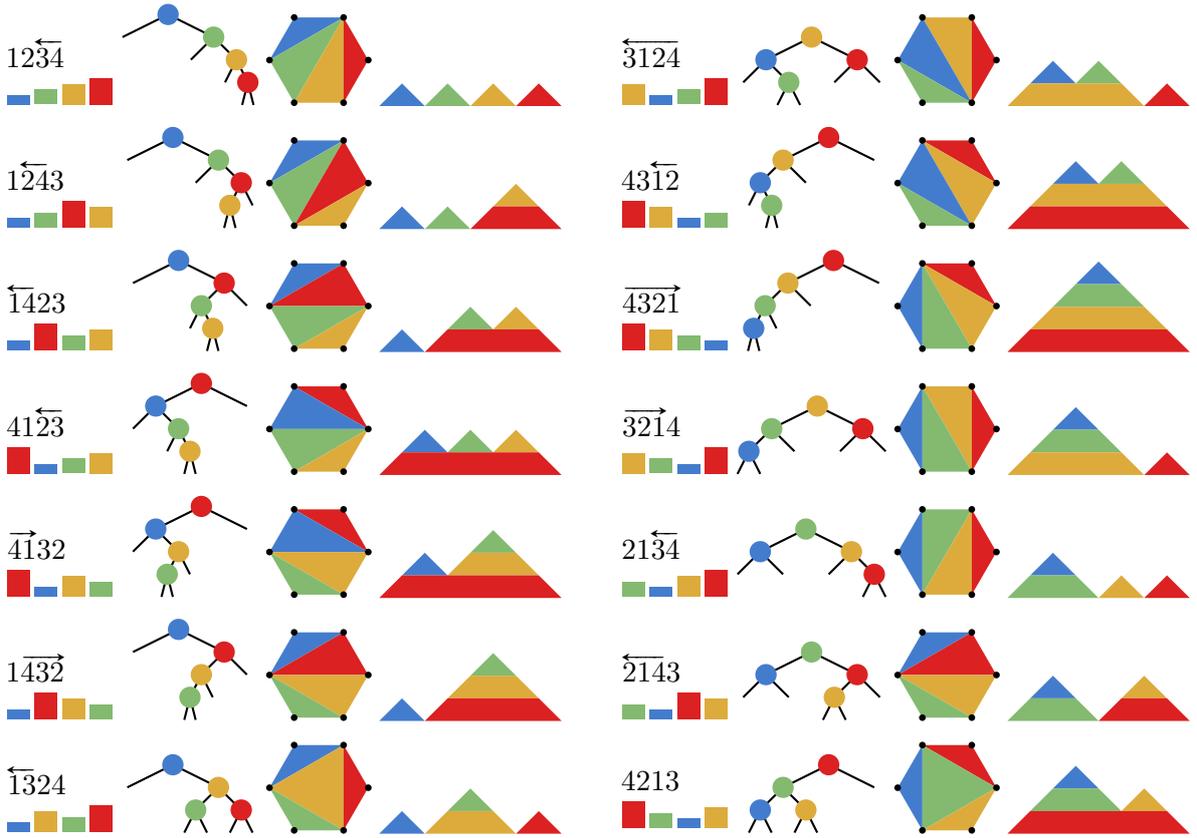

\subsection{Set partitions (Kaye)}

Consider the set~$X_n$ of partitions of the set~$[n]$ into nonempty subsets.
We can represent any such partition~$x=\{x_1,\ldots,x_k\}$, $x_i\subseteq [n]$, in a canonic way, by sorting all subsets in decreasing order of their minimum element, and the elements of each subset in increasing order.
Then $x$ is mapped to a permutation $f(x)\in S_n$ by writing out the elements from left to right in this canonic representation of~$x$.
For instance, the set partition $x=\{\{9\},\{6\},\{3,4,7\},\{1,2,5,8\}\}$ is encoded as the permutation $f(x)=963471258$.
Observe that $f(X_n)$ is the set of permutations with the property that for every descent~$a_i a_{i+1}$, $a_i>a_{i+1}$, in the permutation, no value left of~$a_i$ is smaller than~$a_{i+1}$.
This notion of pattern-avoidance can be described concisely by the vincular permutation pattern $1\ul{32}$, i.e., we have $f(X_n)=S_n(1\ul{32})$ (the formal definition of vincular patterns is given in the next section), which we will prove to be a zigzag language.
A jump of the value~$i$ in the permutation corresponds to moving the element~$i$ from its subset to the previous or next subset in the canonic representation, possibly creating a singleton set~$\{i\}$. 
Moreover, $f^{-1}(J(S_n(1\ul{32})))$ is exactly the ordering of set partitions described by Kaye~\cite{MR0443423}, which he showed can be implemented efficiently.
This ordering is shown in Figure~\ref{fig:part4}.
Applying $f^{-1}$ to Algorithm~J yields the following simple greedy algorithm to describe this order (see \cite{DBLP:conf/wads/Williams13}):
\textbf{J1.} Visit the set partition $\{\{1,\ldots,n\}\}$. \textbf{J2.} Move the largest possible element from its subset to the previous or next subset so as to obtain a previously unvisited partition; then visit this partition and repeat~J2.

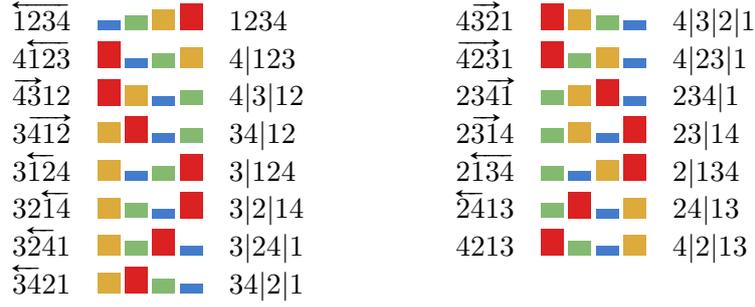
\begin{figure}
\begin{tabular}{ccl@{\hskip 20mm}ccl}
$\lvec{1234}$ & \begin{tikzpicture}[scale=30/100]
\fill[fill=p1col1] (0,0) rectangle (1,0.4);
\fill[fill=p1col2] (1.2,0) rectangle (2.2,0.65);
\fill[fill=p1col3] (2.4000000000000004,0) rectangle (3.4000000000000004,0.9);
\fill[fill=p1col4] (3.6000000000000005,0) rectangle (4.6000000000000005,1.15);
\end{tikzpicture}
 & $1234$ & $4\rvec{32}1$ & \begin{tikzpicture}[scale=30/100]
\fill[fill=p1col4] (0,0) rectangle (1,1.15);
\fill[fill=p1col3] (1.2,0) rectangle (2.2,0.9);
\fill[fill=p1col2] (2.4000000000000004,0) rectangle (3.4000000000000004,0.65);
\fill[fill=p1col1] (3.6000000000000005,0) rectangle (4.6000000000000005,0.4);
\end{tikzpicture}
 & $4|3|2|1$ \\
$4\lvec{123}$ & \begin{tikzpicture}[scale=30/100]
\fill[fill=p1col4] (0,0) rectangle (1,1.15);
\fill[fill=p1col1] (1.2,0) rectangle (2.2,0.4);
\fill[fill=p1col2] (2.4000000000000004,0) rectangle (3.4000000000000004,0.65);
\fill[fill=p1col3] (3.6000000000000005,0) rectangle (4.6000000000000005,0.9);
\end{tikzpicture}
 & $4|123$ & $\rvec{423}1$ & \begin{tikzpicture}[scale=30/100]
\fill[fill=p1col4] (0,0) rectangle (1,1.15);
\fill[fill=p1col2] (1.2,0) rectangle (2.2,0.65);
\fill[fill=p1col3] (2.4000000000000004,0) rectangle (3.4000000000000004,0.9);
\fill[fill=p1col1] (3.6000000000000005,0) rectangle (4.6000000000000005,0.4);
\end{tikzpicture}
 & $4|23|1$  \\
$\rvec{43}12$ & \begin{tikzpicture}[scale=30/100]
\fill[fill=p1col4] (0,0) rectangle (1,1.15);
\fill[fill=p1col3] (1.2,0) rectangle (2.2,0.9);
\fill[fill=p1col1] (2.4000000000000004,0) rectangle (3.4000000000000004,0.4);
\fill[fill=p1col2] (3.6000000000000005,0) rectangle (4.6000000000000005,0.65);
\end{tikzpicture}
 & $4|3|12$ & $23\rvec{41}$ & \begin{tikzpicture}[scale=30/100]
\fill[fill=p1col2] (0,0) rectangle (1,0.65);
\fill[fill=p1col3] (1.2,0) rectangle (2.2,0.9);
\fill[fill=p1col4] (2.4000000000000004,0) rectangle (3.4000000000000004,1.15);
\fill[fill=p1col1] (3.6000000000000005,0) rectangle (4.6000000000000005,0.4);
\end{tikzpicture}
 & $234|1$ \\
$3\rvec{412}$ & \begin{tikzpicture}[scale=30/100]
\fill[fill=p1col3] (0,0) rectangle (1,0.9);
\fill[fill=p1col4] (1.2,0) rectangle (2.2,1.15);
\fill[fill=p1col1] (2.4000000000000004,0) rectangle (3.4000000000000004,0.4);
\fill[fill=p1col2] (3.6000000000000005,0) rectangle (4.6000000000000005,0.65);
\end{tikzpicture}
 & $34|12$ & $2\rvec{31}4$ & \begin{tikzpicture}[scale=30/100]
\fill[fill=p1col2] (0,0) rectangle (1,0.65);
\fill[fill=p1col3] (1.2,0) rectangle (2.2,0.9);
\fill[fill=p1col1] (2.4000000000000004,0) rectangle (3.4000000000000004,0.4);
\fill[fill=p1col4] (3.6000000000000005,0) rectangle (4.6000000000000005,1.15);
\end{tikzpicture}
 & $23|14$  \\
$3\lvec{12}4$ & \begin{tikzpicture}[scale=30/100]
\fill[fill=p1col3] (0,0) rectangle (1,0.9);
\fill[fill=p1col1] (1.2,0) rectangle (2.2,0.4);
\fill[fill=p1col2] (2.4000000000000004,0) rectangle (3.4000000000000004,0.65);
\fill[fill=p1col4] (3.6000000000000005,0) rectangle (4.6000000000000005,1.15);
\end{tikzpicture}
 & $3|124$ & $2\lvec{134}$ & \begin{tikzpicture}[scale=30/100]
\fill[fill=p1col2] (0,0) rectangle (1,0.65);
\fill[fill=p1col1] (1.2,0) rectangle (2.2,0.4);
\fill[fill=p1col3] (2.4000000000000004,0) rectangle (3.4000000000000004,0.9);
\fill[fill=p1col4] (3.6000000000000005,0) rectangle (4.6000000000000005,1.15);
\end{tikzpicture}
 & $2|134$  \\
$32\lvec{14}$ & \begin{tikzpicture}[scale=30/100]
\fill[fill=p1col3] (0,0) rectangle (1,0.9);
\fill[fill=p1col2] (1.2,0) rectangle (2.2,0.65);
\fill[fill=p1col1] (2.4000000000000004,0) rectangle (3.4000000000000004,0.4);
\fill[fill=p1col4] (3.6000000000000005,0) rectangle (4.6000000000000005,1.15);
\end{tikzpicture}
 & $3|2|14$ & $\lvec{24}13$ & \begin{tikzpicture}[scale=30/100]
\fill[fill=p1col2] (0,0) rectangle (1,0.65);
\fill[fill=p1col4] (1.2,0) rectangle (2.2,1.15);
\fill[fill=p1col1] (2.4000000000000004,0) rectangle (3.4000000000000004,0.4);
\fill[fill=p1col3] (3.6000000000000005,0) rectangle (4.6000000000000005,0.9);
\end{tikzpicture}
 & $24|13$ \\
$3\lvec{24}1$ & \begin{tikzpicture}[scale=30/100]
\fill[fill=p1col3] (0,0) rectangle (1,0.9);
\fill[fill=p1col2] (1.2,0) rectangle (2.2,0.65);
\fill[fill=p1col4] (2.4000000000000004,0) rectangle (3.4000000000000004,1.15);
\fill[fill=p1col1] (3.6000000000000005,0) rectangle (4.6000000000000005,0.4);
\end{tikzpicture}
 & $3|24|1$ & $4213$ & \begin{tikzpicture}[scale=30/100]
\fill[fill=p1col4] (0,0) rectangle (1,1.15);
\fill[fill=p1col2] (1.2,0) rectangle (2.2,0.65);
\fill[fill=p1col1] (2.4000000000000004,0) rectangle (3.4000000000000004,0.4);
\fill[fill=p1col3] (3.6000000000000005,0) rectangle (4.6000000000000005,0.9);
\end{tikzpicture}
 & $4|2|13$ \\
$\lvec{34}21$ & \begin{tikzpicture}[scale=30/100]
\fill[fill=p1col3] (0,0) rectangle (1,0.9);
\fill[fill=p1col4] (1.2,0) rectangle (2.2,1.15);
\fill[fill=p1col2] (2.4000000000000004,0) rectangle (3.4000000000000004,0.65);
\fill[fill=p1col1] (3.6000000000000005,0) rectangle (4.6000000000000005,0.4);
\end{tikzpicture}
 & $34|2|1$ & & & \\
\end{tabular}
\caption{$1\ul{32}$-avoiding permutations of length $n=4$ generated by Algorithm~J and resulting Gray code for set partitions, coinciding with Kaye's Gray code.
Set partitions are denoted compactly, omitting curly brackets and commas in the canonic representation, and using vertical bars to separate subsets.}
\label{fig:part4}
\end{figure}

\begin{table}
\tikzstyle{mesh node}=[circle,minimum size=1.6pt,inner sep=0pt,draw=black,fill=black]
\small
\caption{Tame permutation patterns and corresponding combinatorial objects and orderings generated by Algorithm~J.
Some patterns are interesting in their own right and have no `natural' associated combinatorial objects, in which case the first two columns are merged.
The patterns with underscores, bars, boxes, and those drawn as grids with some shaded cells, are defined and explained in Sections~\ref{sec:vinc}--\ref{sec:mesh}.
See Table~\ref{tab:tame2} for more patterns.
}
\label{tab:tame1}
\renewcommand{\arraystretch}{0.9}
\begin{tabular}{|l|l|l|}
\hline
\textbf{Tame patterns} & \textbf{Combinatorial objects and ordering} & \textbf{References/OEIS} \cite{oeis} \\ \hline
none    & permutations by adjacent & \cite{MR0159764,DBLP:journals/cacm/Trotter62}, A000142 \\
        & transpositions \textrightarrow{} plain change order & \\ \hline
$231=2\ul{31}$   & \textit{Catalan families} & A000108 \\
        & \textbullet{} binary trees by rotations \textrightarrow{} Lucas- & \cite{MR1239499} \\
        & \hspace{4mm} -Roelants van Baronaigien-Ruskey order & \\
        & \textbullet{} triangulations by edge flips & \\
        & \textbullet{} Dyck paths by hill flips & \\ \hline
$1\ul{32}$ & \textit{Bell families} & A000110 \\
           & \textbullet{} set partitions by element & \cite{MR0443423, DBLP:conf/wads/Williams13} \\
           & \hspace{2.6mm} exchanges \textrightarrow{} Kaye's order & \\ \hline
$132\wedge 231=\ul{13}2\wedge 2\ul{31}$: & binary strings by bitflips & \cite{gray_1953}, A011782 \\
permutations without peaks & \textrightarrow{} reflected Gray code order (BRGC) & \\ \hline
$1342$ & forests of $\beta(0,1)$-trees & \cite{MR1485138,MR3441169}, \\
       &                               & A022558 \\ \hline
\multicolumn{2}{|l|}{$2143$: vexillary permutations} & \cite{MR815233}, A005802 \\ \hline
\multicolumn{2}{|l|}{conjunction of~$v_k$ tame patterns with $v_2=35, v_3=91, v_4=2346$} & \cite{MR3252657}, A224318, \\
\multicolumn{2}{|l|}{(see \cite{billey_website}): $k$-vexillary permutations $(k\geq 1)$} & A223034, A223905 \\ \hline
\multicolumn{2}{|l|}{$2143\wedge 3412$: skew-merged permutations} & \cite{MR1297387,MR1490467}, A029759 \\ \hline
\multicolumn{2}{|l|}{$2143\wedge 2413\wedge 3142$} & \cite{MR2609478,MR3211776}, A033321 \\ \hline 
\multicolumn{2}{|l|}{$2143\wedge 2413\wedge 3142\wedge 3412$: X-shaped permutations} & \cite{waton_2007,MR2785755}, A006012 \\ \hline
$2413\wedge 3142$: & \textit{Schr\"oder families} & A006318 \\
separable permutations & \textbullet{} slicing floorplans (=guillotine & \cite{MR624050, MR1620935, MR2233287} \\
                       & \hspace{2.6mm} partitions) & \\
                       & \textbullet{} topological drawings of $K_{2,n}$ & \cite{MR3840257} \\ \hline
$2\ul{41}3\wedge 3\ul{14}2$: Baxter & mosaic floorplans (=diagonal & \cite{DBLP:journals/todaes/YaoCCG03,MR2233287} \\
$2\ul{41}3\wedge 3\ul{41}2$: twisted Baxter & rectangulations=R-equivalent & \cite{MR2871762,MR3878132} \\
$2\ul{14}3\wedge 3\ul{14}2$ &  rectangulations) & A001181 \\ \hline
$2\ul{14}3\wedge 3\ul{41}2$ & S-equivalent rectangulations & \cite{MR3084577}, A214358 \\ \hline
$2\ul{14}3\wedge 3\ul{41}2\wedge 2413\wedge 3142$ & S-equivalent guillotine rectangulations & \cite{MR3084577}, A078482 \\ \hline
$3\ul{51}24\wedge 3\ul{51}42\wedge 24\ul{51}3\wedge 42\ul{51}3$: & generic rectangulations & \cite{MR2864445} \\ 
2-clumped permutations                                                                     & (=rectangular drawings) & \\ \hline
\multicolumn{2}{|l|}{conjunction of~$c_k$ tame patterns with $c_k=2(k/2)!(k/2+1)!$ for $k$ even} & \cite{MR2864445} \\ 
\multicolumn{2}{|l|}{and $c_k=2((k+1)/2)!^2$ for $k$ odd: $k$-clumped permutations} & \\ \hline
\multicolumn{2}{|l|}{$12543\!\wedge\! 13254\!\wedge\! 13524\!\wedge\! 13542\!\wedge\! 21543\!\wedge\! 125364\!\wedge\! 125634\!\wedge\! 215364\!\wedge\! 215634$} & \cite{MR4229590} \\
\multicolumn{2}{|l|}{${}\wedge 315264\wedge 315624\wedge 315642$: permutations with 0-1 Schubert polynomial} & \\ \hline
\multicolumn{2}{|l|}{$2143\wedge 2413\wedge 3412\wedge 314562\wedge 412563\wedge 415632\wedge 431562\wedge 512364$} & \cite{MR3881276} \\
\multicolumn{2}{|l|}{${}\wedge 512643\wedge 516432\wedge 541263\wedge 541632\wedge 543162$: widdershins permutations} & \\ \hline
\multicolumn{2}{|l|}{$243\ol{1}$ (A051295), $25\ol{3}14$ (A117106), $352\ol{4}1$ (A137534), $4\ol{2}513$ (A137535)} & \cite{MR2595489} (OEIS \\
\multicolumn{2}{|l|}{$425\ol{1}3$ (A110447), $\ol{4}2153$ (A137536), $25\ol{1}34$ (A137538), $\ol{4}1523$ (A137539)} &  shown on the left) \\
\multicolumn{2}{|l|}{$\ol{4}1253$ (A137540), $3524\ol{1}$ (A137542)} & \\ \hline
\multicolumn{2}{|l|}{$3\ol{1}52\ol{4}=3\ol{1}42\wedge 241\ol{3}$} & \cite{MR2712381,MR2652101}, \\
\multicolumn{2}{|l|}{}                                            & A098569 \\ \hline
\multicolumn{2}{|l|}{$\vcenter{\hbox{\boxp{2143}, \boxp{3142}}}$} & \cite{MR2973348} \\\hline
\multicolumn{2}{|l|}{$\vcenter{\hbox{\begin{tikzpicture}[scale=0.16]
\path [fill=gray!60] (2,1) rectangle (3,5);
\path [fill=gray!60] (1,3) rectangle (5,4);
\path [fill=white] (0,0) rectangle (1,-0.3);
\path [fill=white] (0,6) rectangle (1,6.3);
\mesh{5}{3,1,5,2,4}
\end{tikzpicture}
}}
{}\wedge\,{}
\vcenter{\hbox{\begin{tikzpicture}[scale=0.16]
\path [fill=gray!60] (3,1) rectangle (4,5);
\path [fill=gray!60] (1,2) rectangle (5,3);
\mesh{5}{2,4,1,5,3}
\end{tikzpicture}}}$\hspace{1ex}: \parbox{7.3cm}{permutations that characterize Schubert varieties which are Gorenstein}} & \cite{MR2264071}, A097483 \\ \hline
$\vcenter{\hbox{\begin{tikzpicture}[scale=0.16]
\path [fill=gray!60] (0,1) rectangle (4,2);
\path [fill=gray!60] (1,0) rectangle (2,4);
\mesh{3}{2,3,1}
\end{tikzpicture}
}}$
& $(\mathbf{2}+\mathbf{2})$-free posets & \parbox{4cm}{\begin{flushleft}\cite{parviainen_2009,MR2652101} A022493\end{flushleft}} \\ \hline
\end{tabular}
\end{table}

\section{Pattern-avoiding permutations}
\label{sec:pattern}

The first main application of our framework is the generation of pattern-avoiding permutations.
Our main results in this section are summarized in Theorem~\ref{thm:tame}, Theorem~\ref{thm:mesh} (and its corollaries Lemmas~\ref{lem:class}--\ref{lem:bivinc}), and in Table~\ref{tab:tame1}.
We emphasize that all our results can be generalized to bounding the number of appearances of patterns, where the special case with a bound of~0 appearances is pattern-avoidance; see Section~\ref{sec:mult} below.

\subsection{Preliminaries}
\label{sec:prelim-perm}

The following simple but powerful lemma follows immediately from the definition of zigzag languages given in Section~\ref{sec:zigzag}.
For any set~$L_n\seq S_n$ we define $p(L_n):=\{p(\pi)\mid \pi\in L_n\}$.

\begin{lemma}
\label{lem:zigzag-ops}
Let $L_n,M_n\seq S_n$, $n\geq 1$, be two zigzag languages of permutations.
Then $L_n\cup M_n$ and $L_n\cap M_n$ are also zigzag languages of permutations, and we have $p(L_n\cup M_n)=p(L_n)\cup p(M_n)$ and $p(L_n\cap M_n)=p(L_n)\cap p(M_n)$.
\end{lemma}

We say that two sequences of integers $\sigma$ and~$\tau$ are \emph{order-isomorphic}, if their elements appear in the same relative order in both sequences.
For instance, $2576$ and $1243$ are order-isomorphic.
Given two permutations~$\pi\in S_n$ and~$\tau\in S_k$, we say that $\pi$ \emph{contains the pattern $\tau$}, if and only if $\pi=a_1\cdots a_n$ contains a subpermutation $a_{i_1}\cdots a_{i_k}$, $i_1<\cdots<i_k$, that is order-isomorphic to $\tau$.
We refer to such a subpermutation as a \emph{match} of~$\tau$ in~$\pi$.
If $\pi$ does not contain the pattern~$\tau$, then we say that \emph{$\pi$ avoids~$\tau$}.
For example, $\pi=6\colorbox{black!30}{\!35\!}41\colorbox{black!30}{\!2\!}$ contains the pattern $\tau=231$, as the highlighted entries form a match of~$\tau$ in~$\pi$.
On the other hand, $\pi=654123$ avoids $\tau=231$.
For any permutation~$\tau$, we let~$S_n(\tau)$ denote all permutations from~$S_n$ avoiding the pattern~$\tau$.

For propositional formulas~$F$ and~$G$ consisting of logical ANDs~$\wedge$, ORs~$\vee$, and patterns as variables, we define
\begin{align}
\begin{split}
\label{eq:zigzag-ops}
S_n(F\wedge G) &:= S_n(F)\cap S_n(G), \\
S_n(F\vee G)   &:= S_n(F)\cup S_n(G).
\end{split}
\end{align}
For instance, $S_n(\tau_1\wedge\cdots\wedge \tau_\ell)$ is the set of permutations avoiding each of the patterns~$\tau_1,\ldots,\tau_\ell$, and $S_n(\tau_1\vee\cdots\vee \tau_\ell)$ is the set of permutations avoiding at least one of the patterns~$\tau_1,\ldots,\tau_\ell$.

\begin{remark}
From the point of view of counting, we clearly have $|L_n\cup M_n|=|L_n|+|M_n|-|L_n\cap M_n|$, so the problem of counting the union of two zigzag languages can be reduced to counting the individual languages and the intersection.
However, from the point of view of exhaustive generation, we clearly do not want to take this approach, namely generate all permutations in~$L_n$, all permutations in~$M_n$, all permutations in~$L_n\cap M_n$, and then combine and reduce those lists.
This shows that the problem of generating languages like $S_n(\tau_1\vee\cdots\vee \tau_\ell)$ or $S_n(F)$ for more general formulas $F$ is genuinely interesting in our context.
We will see a few applications of this general setting below, and we feel that this direction of generalization deserves further investigation by the pattern-avoidance community.
\end{remark}

\subsection{Tame patterns}
\label{sec:tame}

We say that an infinite sequence of sets $L_0,L_1,\ldots$ is \emph{hereditary}, if $L_{i-1}=p(L_i)$ holds for all $i\geq 1$.
We say that a permutation pattern~$\tau$ is \emph{tame}, if $S_n(\tau)$, $n\geq 0$, is a hereditary sequence of zigzag languages.
The hereditary property ensures that for a given set $S_n(\tau)=:L_n$, we can check whether a permutation~$\pi$ is in the sets $L_{i-1}:=p(L_i)$ for $i=n,n-1,\ldots,1$ simply by checking for matches of the pattern~$\tau$ in~$\pi$.
See also the discussion in Section~\ref{sec:limits}.

The following theorem is an immediate consequence of Lemma~\ref{lem:zigzag-ops} and the definition~\eqref{eq:zigzag-ops}.

\begin{theorem}
\label{thm:tame}
Let $F$ be an arbitrary propositional formula consisting of logical ANDs~$\wedge$, ORs~$\vee$, and tame patterns as variables, then $S_n(F)$, $n\geq 0$, is a hereditary sequence of zigzag languages.
Consequently, all of these languages can be generated by Algorithm~J.
\end{theorem}

In the following we provide simple sufficient conditions guaranteeing that a pattern is tame (cf.~Section~\ref{sec:limits}).

\begin{lemma}
\label{lem:class}
If a pattern $\tau\in S_k$, $k\geq 3$, does not have the largest value~$k$ at the leftmost or rightmost position, then it is tame.
\end{lemma}

We prove Lemma~\ref{lem:class} in Section~\ref{sec:mesh-proofs}.

Table~\ref{tab:tame1} lists several tame patterns and the combinatorial objects encoded by the corresponding zigzag languages.
The bijections between those permutations and the combinatorial objects are well-known and are described in the listed papers (recall also Section~\ref{sec:examples}).
The ordering of $231$-avoiding permutations of length~$n=4$ generated by Algorithm~J, and the corresponding Gray codes for three different Catalan objects are shown in Figure~\ref{fig:cat4}.
We refer to the permutation patterns discussed so far as \emph{classical} patterns.
In the following we discuss some other important variants of permutation patterns appearing in the literature.

\subsection{Vincular patterns}
\label{sec:vinc}

Vincular patterns were introduced by Babson and Steingr\'{i}msson~\cite{MR1758852}.
In a \emph{vincular} pattern~$\tau$, there is exactly one underlined pair of consecutive entries, with the interpretation that a match of~$\tau$ in~$\pi$ requires that the underlined entries match adjacent positions in~$\pi$.
For instance, the permutation $\pi=\colorbox{black!30}{\!3\!}1\colorbox{black!30}{\!4\!}\colorbox{black!30}{\!2\!}$ contains the pattern~$\tau=231$, but it avoids the vincular pattern~$\tau=\ul{23}1$.

\begin{lemma}
\label{lem:vinc}
If a vincular pattern $\tau\in S_k$, $k\geq 3$, does not have the largest value~$k$ at the leftmost or rightmost position, and the largest value~$k$ is part of the vincular pair, then it is tame.
\end{lemma}

We prove Lemma~\ref{lem:vinc} in Section~\ref{sec:mesh-proofs}.

Table~\ref{tab:tame1} also lists several tame vincular patterns and the combinatorial objects encoded by the corresponding zigzag languages, namely set partitions and different kinds of rectangulations.
The ordering of $1\ul{32}$-avoiding permutations of length $n=4$ generated by Algorithm~J, and the resulting Gray code for set partitions, is shown in Figure~\ref{fig:part4}.
The generated ordering of twisted Baxter permutations of length~$n=4$, and the resulting Gray code for diagonal rectangulations, is shown in Figure~\ref{fig:tbaxter4}.

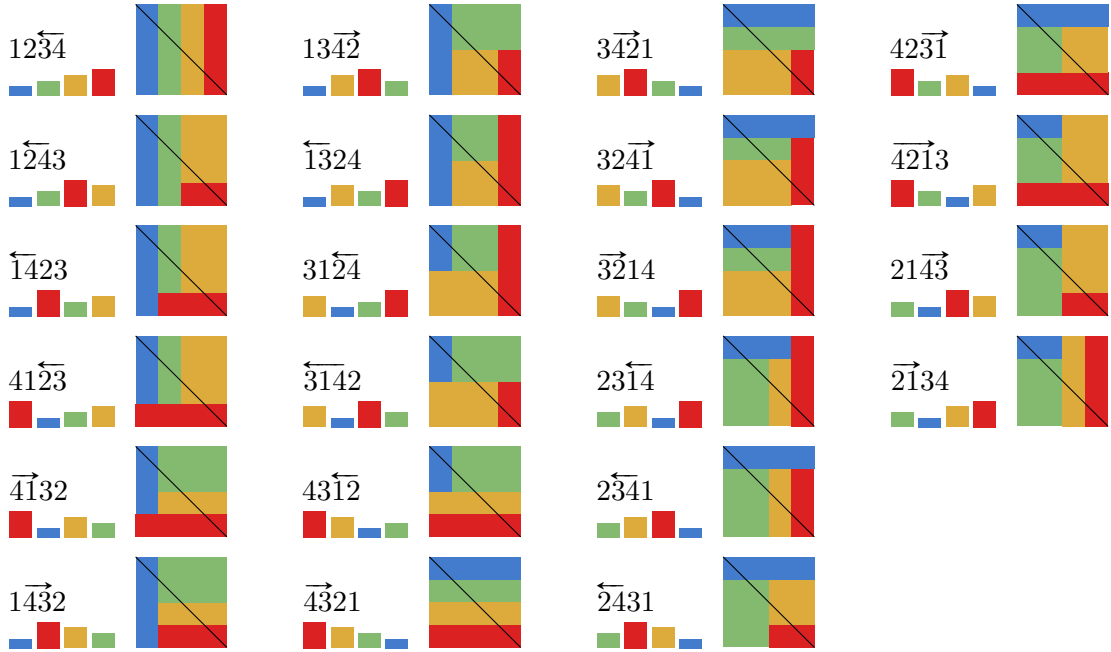
\begin{figure}
\begin{tabular}{cc@{\hskip 10mm}cc@{\hskip 10mm}cc@{\hskip 10mm}cc}
\parbox{13mm}{12\lvec{34} \\ \begin{tikzpicture}[scale=30/100]
\fill[fill=p1col1] (0,0) rectangle (1,0.4);
\fill[fill=p1col2] (1.2,0) rectangle (2.2,0.65);
\fill[fill=p1col3] (2.4000000000000004,0) rectangle (3.4000000000000004,0.9);
\fill[fill=p1col4] (3.6000000000000005,0) rectangle (4.6000000000000005,1.15);
\end{tikzpicture}
 \\[1.5mm]} & \begin{tikzpicture}[scale=30/100,font=\tiny]
\fill[p1col1] (0,4) rectangle (1,0);
\fill[p1col2] (1,4) rectangle (2,0);
\fill[p1col3] (2,4) rectangle (3,0);
\fill[p1col4] (3,4) rectangle (4,0);
\draw (0,4)--(4,0);
\end{tikzpicture}
 & 
\parbox{13mm}{13\rvec{42} \\ \begin{tikzpicture}[scale=30/100]
\fill[fill=p1col1] (0,0) rectangle (1,0.4);
\fill[fill=p1col3] (1.2,0) rectangle (2.2,0.9);
\fill[fill=p1col4] (2.4000000000000004,0) rectangle (3.4000000000000004,1.15);
\fill[fill=p1col2] (3.6000000000000005,0) rectangle (4.6000000000000005,0.65);
\end{tikzpicture}
 \\[1.5mm]} & \begin{tikzpicture}[scale=30/100,font=\tiny]
\fill[p1col1] (0,4) rectangle (1,0);
\fill[p1col3] (1,2) rectangle (3,0);
\fill[p1col4] (3,2) rectangle (4,0);
\fill[p1col2] (1,4) rectangle (4,2);
\draw (0,4)--(4,0);
\end{tikzpicture}
 & 
\parbox{13mm}{3\rvec{42}1 \\ \begin{tikzpicture}[scale=30/100]
\fill[fill=p1col3] (0,0) rectangle (1,0.9);
\fill[fill=p1col4] (1.2,0) rectangle (2.2,1.15);
\fill[fill=p1col2] (2.4000000000000004,0) rectangle (3.4000000000000004,0.65);
\fill[fill=p1col1] (3.6000000000000005,0) rectangle (4.6000000000000005,0.4);
\end{tikzpicture}
 \\[1.5mm]} & \begin{tikzpicture}[scale=30/100,font=\tiny]
\fill[p1col3] (0,2) rectangle (3,0);
\fill[p1col4] (3,2) rectangle (4,0);
\fill[p1col2] (0,3) rectangle (4,2);
\fill[p1col1] (0,4) rectangle (4,3);
\draw (0,4)--(4,0);
\end{tikzpicture}
 & 
\parbox{13mm}{42\rvec{31} \\ \begin{tikzpicture}[scale=30/100]
\fill[fill=p1col4] (0,0) rectangle (1,1.15);
\fill[fill=p1col2] (1.2,0) rectangle (2.2,0.65);
\fill[fill=p1col3] (2.4000000000000004,0) rectangle (3.4000000000000004,0.9);
\fill[fill=p1col1] (3.6000000000000005,0) rectangle (4.6000000000000005,0.4);
\end{tikzpicture}
 \\[1.5mm]} & \begin{tikzpicture}[scale=30/100,font=\tiny]
\fill[p1col4] (0,1) rectangle (4,0);
\fill[p1col2] (0,3) rectangle (2,1);
\fill[p1col3] (2,3) rectangle (4,1);
\fill[p1col1] (0,4) rectangle (4,3);
\draw (0,4)--(4,0);
\end{tikzpicture}
 \vspace{-4mm} \\
\parbox{13mm}{1\lvec{24}3 \\ \begin{tikzpicture}[scale=30/100]
\fill[fill=p1col1] (0,0) rectangle (1,0.4);
\fill[fill=p1col2] (1.2,0) rectangle (2.2,0.65);
\fill[fill=p1col4] (2.4000000000000004,0) rectangle (3.4000000000000004,1.15);
\fill[fill=p1col3] (3.6000000000000005,0) rectangle (4.6000000000000005,0.9);
\end{tikzpicture}
 \\[1.5mm]} & \begin{tikzpicture}[scale=30/100,font=\tiny]
\fill[p1col1] (0,4) rectangle (1,0);
\fill[p1col2] (1,4) rectangle (2,0);
\fill[p1col4] (2,1) rectangle (4,0);
\fill[p1col3] (2,4) rectangle (4,1);
\draw (0,4)--(4,0);
\end{tikzpicture}
 &
\parbox{13mm}{\lvec{13}24 \\ \begin{tikzpicture}[scale=30/100]
\fill[fill=p1col1] (0,0) rectangle (1,0.4);
\fill[fill=p1col3] (1.2,0) rectangle (2.2,0.9);
\fill[fill=p1col2] (2.4000000000000004,0) rectangle (3.4000000000000004,0.65);
\fill[fill=p1col4] (3.6000000000000005,0) rectangle (4.6000000000000005,1.15);
\end{tikzpicture}
 \\[1.5mm]} & \begin{tikzpicture}[scale=30/100,font=\tiny]
\fill[p1col1] (0,4) rectangle (1,0);
\fill[p1col3] (1,2) rectangle (3,0);
\fill[p1col2] (1,4) rectangle (3,2);
\fill[p1col4] (3,4) rectangle (4,0);
\draw (0,4)--(4,0);
\end{tikzpicture}
 &
\parbox{13mm}{32\rvec{41} \\ \begin{tikzpicture}[scale=30/100]
\fill[fill=p1col3] (0,0) rectangle (1,0.9);
\fill[fill=p1col2] (1.2,0) rectangle (2.2,0.65);
\fill[fill=p1col4] (2.4000000000000004,0) rectangle (3.4000000000000004,1.15);
\fill[fill=p1col1] (3.6000000000000005,0) rectangle (4.6000000000000005,0.4);
\end{tikzpicture}
 \\[1.5mm]} & \begin{tikzpicture}[scale=30/100,font=\tiny]
\fill[p1col3] (0,2) rectangle (3,0);
\fill[p1col2] (0,3) rectangle (3,2);
\fill[p1col4] (3,3) rectangle (4,0);
\fill[p1col1] (0,4) rectangle (4,3);
\draw (0,4)--(4,0);
\end{tikzpicture}
 &
\parbox{13mm}{\rvec{421}3 \\ \begin{tikzpicture}[scale=30/100]
\fill[fill=p1col4] (0,0) rectangle (1,1.15);
\fill[fill=p1col2] (1.2,0) rectangle (2.2,0.65);
\fill[fill=p1col1] (2.4000000000000004,0) rectangle (3.4000000000000004,0.4);
\fill[fill=p1col3] (3.6000000000000005,0) rectangle (4.6000000000000005,0.9);
\end{tikzpicture}
 \\[1.5mm]} & \begin{tikzpicture}[scale=30/100,font=\tiny]
\fill[p1col4] (0,1) rectangle (4,0);
\fill[p1col2] (0,3) rectangle (2,1);
\fill[p1col1] (0,4) rectangle (2,3);
\fill[p1col3] (2,4) rectangle (4,1);
\draw (0,4)--(4,0);
\end{tikzpicture}
 \vspace{-4mm} \\
\parbox{13mm}{\lvec{14}23 \\ \begin{tikzpicture}[scale=30/100]
\fill[fill=p1col1] (0,0) rectangle (1,0.4);
\fill[fill=p1col4] (1.2,0) rectangle (2.2,1.15);
\fill[fill=p1col2] (2.4000000000000004,0) rectangle (3.4000000000000004,0.65);
\fill[fill=p1col3] (3.6000000000000005,0) rectangle (4.6000000000000005,0.9);
\end{tikzpicture}
 \\[1.5mm]} & \begin{tikzpicture}[scale=30/100,font=\tiny]
\fill[p1col1] (0,4) rectangle (1,0);
\fill[p1col4] (1,1) rectangle (4,0);
\fill[p1col2] (1,4) rectangle (2,1);
\fill[p1col3] (2,4) rectangle (4,1);
\draw (0,4)--(4,0);
\end{tikzpicture}
 &
\parbox{13mm}{31\lvec{24} \\ \begin{tikzpicture}[scale=30/100]
\fill[fill=p1col3] (0,0) rectangle (1,0.9);
\fill[fill=p1col1] (1.2,0) rectangle (2.2,0.4);
\fill[fill=p1col2] (2.4000000000000004,0) rectangle (3.4000000000000004,0.65);
\fill[fill=p1col4] (3.6000000000000005,0) rectangle (4.6000000000000005,1.15);
\end{tikzpicture}
 \\[1.5mm]} & \begin{tikzpicture}[scale=30/100,font=\tiny]
\fill[p1col3] (0,2) rectangle (3,0);
\fill[p1col1] (0,4) rectangle (1,2);
\fill[p1col2] (1,4) rectangle (3,2);
\fill[p1col4] (3,4) rectangle (4,0);
\draw (0,4)--(4,0);
\end{tikzpicture}
 &
\parbox{13mm}{\rvec{32}14 \\ \begin{tikzpicture}[scale=30/100]
\fill[fill=p1col3] (0,0) rectangle (1,0.9);
\fill[fill=p1col2] (1.2,0) rectangle (2.2,0.65);
\fill[fill=p1col1] (2.4000000000000004,0) rectangle (3.4000000000000004,0.4);
\fill[fill=p1col4] (3.6000000000000005,0) rectangle (4.6000000000000005,1.15);
\end{tikzpicture}
 \\[1.5mm]} & \begin{tikzpicture}[scale=30/100,font=\tiny]
\fill[p1col3] (0,2) rectangle (3,0);
\fill[p1col2] (0,3) rectangle (3,2);
\fill[p1col1] (0,4) rectangle (3,3);
\fill[p1col4] (3,4) rectangle (4,0);
\draw (0,4)--(4,0);
\end{tikzpicture}
 &
\parbox{13mm}{21\rvec{43} \\ \begin{tikzpicture}[scale=30/100]
\fill[fill=p1col2] (0,0) rectangle (1,0.65);
\fill[fill=p1col1] (1.2,0) rectangle (2.2,0.4);
\fill[fill=p1col4] (2.4000000000000004,0) rectangle (3.4000000000000004,1.15);
\fill[fill=p1col3] (3.6000000000000005,0) rectangle (4.6000000000000005,0.9);
\end{tikzpicture}
 \\[1.5mm]} & \begin{tikzpicture}[scale=30/100,font=\tiny]
\fill[p1col2] (0,3) rectangle (2,0);
\fill[p1col1] (0,4) rectangle (2,3);
\fill[p1col4] (2,1) rectangle (4,0);
\fill[p1col3] (2,4) rectangle (4,1);
\draw (0,4)--(4,0);
\end{tikzpicture}
 \vspace{-4mm} \\
\parbox{13mm}{41\lvec{23} \\ \begin{tikzpicture}[scale=30/100]
\fill[fill=p1col4] (0,0) rectangle (1,1.15);
\fill[fill=p1col1] (1.2,0) rectangle (2.2,0.4);
\fill[fill=p1col2] (2.4000000000000004,0) rectangle (3.4000000000000004,0.65);
\fill[fill=p1col3] (3.6000000000000005,0) rectangle (4.6000000000000005,0.9);
\end{tikzpicture}
 \\[1.5mm]} & \begin{tikzpicture}[scale=30/100,font=\tiny]
\fill[p1col4] (0,1) rectangle (4,0);
\fill[p1col1] (0,4) rectangle (1,1);
\fill[p1col2] (1,4) rectangle (2,1);
\fill[p1col3] (2,4) rectangle (4,1);
\draw (0,4)--(4,0);
\end{tikzpicture}
 &
\parbox{13mm}{\lvec{314}2 \\ \begin{tikzpicture}[scale=30/100]
\fill[fill=p1col3] (0,0) rectangle (1,0.9);
\fill[fill=p1col1] (1.2,0) rectangle (2.2,0.4);
\fill[fill=p1col4] (2.4000000000000004,0) rectangle (3.4000000000000004,1.15);
\fill[fill=p1col2] (3.6000000000000005,0) rectangle (4.6000000000000005,0.65);
\end{tikzpicture}
 \\[1.5mm]} & \begin{tikzpicture}[scale=30/100,font=\tiny]
\fill[p1col3] (0,2) rectangle (3,0);
\fill[p1col1] (0,4) rectangle (1,2);
\fill[p1col4] (3,2) rectangle (4,0);
\fill[p1col2] (1,4) rectangle (4,2);
\draw (0,4)--(4,0);
\end{tikzpicture}
 &
\parbox{13mm}{23\lvec{14} \\ \begin{tikzpicture}[scale=30/100]
\fill[fill=p1col2] (0,0) rectangle (1,0.65);
\fill[fill=p1col3] (1.2,0) rectangle (2.2,0.9);
\fill[fill=p1col1] (2.4000000000000004,0) rectangle (3.4000000000000004,0.4);
\fill[fill=p1col4] (3.6000000000000005,0) rectangle (4.6000000000000005,1.15);
\end{tikzpicture}
 \\[1.5mm]} & \begin{tikzpicture}[scale=30/100,font=\tiny]
\fill[p1col2] (0,3) rectangle (2,0);
\fill[p1col3] (2,3) rectangle (3,0);
\fill[p1col1] (0,4) rectangle (3,3);
\fill[p1col4] (3,4) rectangle (4,0);
\draw (0,4)--(4,0);
\end{tikzpicture}
 &
\parbox{13mm}{\rvec{21}34 \\ \begin{tikzpicture}[scale=30/100]
\fill[fill=p1col2] (0,0) rectangle (1,0.65);
\fill[fill=p1col1] (1.2,0) rectangle (2.2,0.4);
\fill[fill=p1col3] (2.4000000000000004,0) rectangle (3.4000000000000004,0.9);
\fill[fill=p1col4] (3.6000000000000005,0) rectangle (4.6000000000000005,1.15);
\end{tikzpicture}
 \\[1.5mm]} & \begin{tikzpicture}[scale=30/100,font=\tiny]
\fill[p1col2] (0,3) rectangle (2,0);
\fill[p1col1] (0,4) rectangle (2,3);
\fill[p1col3] (2,4) rectangle (3,0);
\fill[p1col4] (3,4) rectangle (4,0);
\draw (0,4)--(4,0);
\end{tikzpicture}
 \vspace{-4mm} \\
\parbox{13mm}{\rvec{41}32 \\ \begin{tikzpicture}[scale=30/100]
\fill[fill=p1col4] (0,0) rectangle (1,1.15);
\fill[fill=p1col1] (1.2,0) rectangle (2.2,0.4);
\fill[fill=p1col3] (2.4000000000000004,0) rectangle (3.4000000000000004,0.9);
\fill[fill=p1col2] (3.6000000000000005,0) rectangle (4.6000000000000005,0.65);
\end{tikzpicture}
 \\[1.5mm]} & \begin{tikzpicture}[scale=30/100,font=\tiny]
\fill[p1col4] (0,1) rectangle (4,0);
\fill[p1col1] (0,4) rectangle (1,1);
\fill[p1col3] (1,2) rectangle (4,1);
\fill[p1col2] (1,4) rectangle (4,2);
\draw (0,4)--(4,0);
\end{tikzpicture}
 &
\parbox{13mm}{43\lvec{12} \\ \begin{tikzpicture}[scale=30/100]
\fill[fill=p1col4] (0,0) rectangle (1,1.15);
\fill[fill=p1col3] (1.2,0) rectangle (2.2,0.9);
\fill[fill=p1col1] (2.4000000000000004,0) rectangle (3.4000000000000004,0.4);
\fill[fill=p1col2] (3.6000000000000005,0) rectangle (4.6000000000000005,0.65);
\end{tikzpicture}
 \\[1.5mm]} & \begin{tikzpicture}[scale=30/100,font=\tiny]
\fill[p1col4] (0,1) rectangle (4,0);
\fill[p1col3] (0,2) rectangle (4,1);
\fill[p1col1] (0,4) rectangle (1,2);
\fill[p1col2] (1,4) rectangle (4,2);
\draw (0,4)--(4,0);
\end{tikzpicture}
 &
\parbox{13mm}{2\lvec{34}1 \\ \begin{tikzpicture}[scale=30/100]
\fill[fill=p1col2] (0,0) rectangle (1,0.65);
\fill[fill=p1col3] (1.2,0) rectangle (2.2,0.9);
\fill[fill=p1col4] (2.4000000000000004,0) rectangle (3.4000000000000004,1.15);
\fill[fill=p1col1] (3.6000000000000005,0) rectangle (4.6000000000000005,0.4);
\end{tikzpicture}
 \\[1.5mm]} & \begin{tikzpicture}[scale=30/100,font=\tiny]
\fill[p1col2] (0,3) rectangle (2,0);
\fill[p1col3] (2,3) rectangle (3,0);
\fill[p1col4] (3,3) rectangle (4,0);
\fill[p1col1] (0,4) rectangle (4,3);
\draw (0,4)--(4,0);
\end{tikzpicture}
 & & \vspace{-4mm} \\
\parbox{13mm}{1\rvec{43}2 \\ \begin{tikzpicture}[scale=30/100]
\fill[fill=p1col1] (0,0) rectangle (1,0.4);
\fill[fill=p1col4] (1.2,0) rectangle (2.2,1.15);
\fill[fill=p1col3] (2.4000000000000004,0) rectangle (3.4000000000000004,0.9);
\fill[fill=p1col2] (3.6000000000000005,0) rectangle (4.6000000000000005,0.65);
\end{tikzpicture}
 \\[1.5mm]} & \begin{tikzpicture}[scale=30/100,font=\tiny]
\fill[p1col1] (0,4) rectangle (1,0);
\fill[p1col4] (1,1) rectangle (4,0);
\fill[p1col3] (1,2) rectangle (4,1);
\fill[p1col2] (1,4) rectangle (4,2);
\draw (0,4)--(4,0);
\end{tikzpicture}
 &
\parbox{13mm}{\rvec{43}21 \\ \begin{tikzpicture}[scale=30/100]
\fill[fill=p1col4] (0,0) rectangle (1,1.15);
\fill[fill=p1col3] (1.2,0) rectangle (2.2,0.9);
\fill[fill=p1col2] (2.4000000000000004,0) rectangle (3.4000000000000004,0.65);
\fill[fill=p1col1] (3.6000000000000005,0) rectangle (4.6000000000000005,0.4);
\end{tikzpicture}
 \\[1.5mm]} & \begin{tikzpicture}[scale=30/100,font=\tiny]
\fill[p1col4] (0,1) rectangle (4,0);
\fill[p1col3] (0,2) rectangle (4,1);
\fill[p1col2] (0,3) rectangle (4,2);
\fill[p1col1] (0,4) rectangle (4,3);
\draw (0,4)--(4,0);
\end{tikzpicture}
 &
\parbox{13mm}{\lvec{24}31 \\ \begin{tikzpicture}[scale=30/100]
\fill[fill=p1col2] (0,0) rectangle (1,0.65);
\fill[fill=p1col4] (1.2,0) rectangle (2.2,1.15);
\fill[fill=p1col3] (2.4000000000000004,0) rectangle (3.4000000000000004,0.9);
\fill[fill=p1col1] (3.6000000000000005,0) rectangle (4.6000000000000005,0.4);
\end{tikzpicture}
 \\[1.5mm]} & \begin{tikzpicture}[scale=30/100,font=\tiny]
\fill[p1col2] (0,3) rectangle (2,0);
\fill[p1col4] (2,1) rectangle (4,0);
\fill[p1col3] (2,3) rectangle (4,1);
\fill[p1col1] (0,4) rectangle (4,3);
\draw (0,4)--(4,0);
\end{tikzpicture}
 & & \vspace{-4mm} \\
\end{tabular}
\caption{Twisted Baxter permutations ($2\ul{41}3\wedge 3\ul{41}2$-avoiding) for $n=4$ generated by Algorithm~J and resulting Gray code for diagonal rectangulations.}
\label{fig:tbaxter4}
\end{figure}

\subsection{Barred patterns}
\label{sec:barred}

Barred permutation patterns were first considered by West~\cite{MR2716312}.
A \emph{barred} pattern is a pattern~$\tau$ with a number of overlined entries, e.g., $\tau=25\ol{3}41$.
Let $\tau'$ be the permutation obtained by removing the bars in~$\tau$, and let $\tau^-$ be the permutation that is order-isomorphic to the non-barred entries in~$\tau$.
In our example, we have $\tau'=25341$ and $\tau^-=2431$.
A permutation~$\pi$ \emph{contains} a barred pattern~$\tau$ if and only if it contains a match of~$\tau^-$ that cannot be extended to a match of~$\tau'$ by adding entries of~$\pi$ at the positions specified by the barred entries.
For instance, $\pi=\colorbox{black!30}{\!35\!}2\colorbox{black!30}{\!41\!}$ contains $\tau=25\ol{3}41$, as the highlighted entries form a match of $\tau^-=2431$ that cannot be extended to a match of $\tau'=25341$.
We clearly have $S_n(\tau^-)\seq S_n(\tau)$.

The following lemma gives a sufficient condition for a single-barred pattern to be tame.

\begin{lemma}
\label{lem:barred}
If for a single-barred pattern $\tau\in S_k$, $k\geq 4$, the permutation $\tau^-\in S_{k-1}$ does not have the largest value~$k-1$ at the leftmost or rightmost position, and the barred entry in~$\tau$ is smaller than~$k$ or at a position next to the entry~$k-1$, then $\tau$ is tame.
\end{lemma}

We prove Lemma~\ref{lem:barred} in Section~\ref{sec:mesh-proofs}.
See Table~\ref{tab:tame1} for several examples of tame single-barred patterns that were studied by Pudwell~\cite{MR2595489}.
As we will show in Section~\ref{sec:mbarred} below, in many cases patterns with multiple bars can be reduced to single-barred patterns.

\subsection{Boxed patterns}

\emph{Boxed patterns} were introduced in the paper by Avgustinovich, Kitaev, and Valyuzhenich~\cite{MR2973348}.
A permutation $\pi$ \emph{contains} the boxed pattern~$\boxp{\tau}$ if and only if it contains a match of~$\tau$ such that no entry of~$\pi$ at a position between the matched ones has a value between the smallest and largest value of the match.
For example, the permutation $\pi=\colorbox{black!30}{\!43\!}1\colorbox{black!30}{\!7\!}928\colorbox{black!30}{\!6\!}5$ contains the boxed pattern~$\boxp{2143}$, as the highlighted entries of~$\pi$ form a match of~$2143$, and the entries~$1,9,2,8$ are either smaller than~3 or larger than~7.
On the other hand, the permutation $\pi'=\colorbox{black!30}{\!3\!}5\colorbox{black!30}{\!16\!}2\colorbox{black!30}{\!4\!}$ avoids~$\boxp{2143}$, as the only possible match of~$2143$ is at the highlighted positions in~$\pi'$, but the entries~5 and~2 are both between~1 and~6. 

\begin{lemma}
\label{lem:boxed}
Given a boxed pattern~$\boxp{\tau}$ with $\tau\in S_k$ and $k\geq 3$, if $\tau$ does not have the largest value~$k$ at the leftmost or rightmost position, then it is tame.
\end{lemma}

We prove Lemma~\ref{lem:boxed} in Section~\ref{sec:mesh-proofs}.
Table~\ref{tab:tame1} shows two tame boxed patterns studied in~\cite{MR2973348}.

\subsection{Patterns with Bruhat restrictions}

\emph{Patterns with Bruhat restrictions} were introduced by Woo and Yong~\cite{MR2264071}.
Such a pattern is a pair $(\tau,B)$, where $\tau\in S_k$ and $B\seq [k]^2$ is a set of pairs of indices $(a,b)$ with $a<b$ and $\tau(a)<\tau(b)$ such that for all $i\in\{a+1,\ldots,b-1\}$ we either have $\tau(i)<\tau(a)$ or $\tau(i)>\tau(b)$.
A permutation~$\pi$ \emph{contains} this pattern if and only if it contains a match of~$\tau$, and for any pair of entries~$\pi(i_a)$ and~$\pi(i_b)$ that are matched by a corresponding pair of entries~$\tau(a)$ and~$\tau(b)$ with $(a,b)\in B$, we have that $\pi(i)<\pi(i_a)$ or $\pi(i)>\pi(i_b)$ for all $i\in\{i_a+1,\ldots,i_b-1\}$.

\begin{lemma}
\label{lem:bruhat}
Given a pattern with Bruhat restrictions $(\tau,B)$ with $\tau\in S_k$ and $k\geq 3$, if $\tau$ does not have the largest value~$k$ at the leftmost or rightmost position, then it is tame.
\end{lemma}

We prove Lemma~\ref{lem:bruhat} in Section~\ref{sec:mesh-proofs}.
Note that Lemma~\ref{lem:bruhat} does not impose any restrictions on the set~$B$, and that it hence generalizes Lemma~\ref{lem:class}, which corresponds to the case $B=\emptyset$.
Table~\ref{tab:tame1} shows two patterns with Bruhat restriction studied in~\cite{MR2264071}.

\subsection{Bivincular patterns}

\emph{Bivincular patterns} were introduced by Bousquet-M\'{e}lou, Claesson, Dukes, and Kitaev~\cite{MR2652101}.
Such a pattern is a pair $(\tau,B)$, where $\tau\in S_k$ is a vincular pattern and $B\seq [k-1]$.
A permutation~$\pi$ \emph{contains} this pattern if and only if it contains a match of the vincular pattern~$\tau$ (respecting the adjacency condition for the vincular pair), and in this match, the entries at positions $\tau^{-1}(i)$ and $\tau^{-1}(i+1)$ are consecutive values in~$\pi$ for all $i\in B$.

\begin{lemma}
\label{lem:bivinc}
Given a bivincular pattern $(\tau,B)$ with $\tau\in S_k$ and $k\geq 3$, if the vincular pattern~$\tau$ satisfies the conditions in Lemma~\ref{lem:vinc} and if $k-1\notin B$, then it is tame.
\end{lemma}

We prove Lemma~\ref{lem:bivinc} in Section~\ref{sec:mesh-proofs}.
Note that Lemma~\ref{lem:bivinc} generalizes Lemma~\ref{lem:vinc}, which corresponds to the case $B=\emptyset$.
In Table~\ref{tab:tame1}, $(\mathbf{2}+\mathbf{2})$-free posets are mentioned as an example of a combinatorial class that is in bijection to permutations avoiding a tame bivincular pattern.

\subsection{Mesh patterns}
\label{sec:mesh}

In the following we take a geometric viewpoint on permutations.
For any pair of real numbers $P=(a,b)$, we define $P_x:=a$ and $P_y:=b$.
Moreover, for any mapping $\alpha:[n]\rightarrow B$ and any subset $I\seq [n]$ with $|I|=k$ we write $\alpha|_I:[k]\rightarrow B$ for the function defined by $\alpha|_I(i):=\alpha(j)$, where $j$ is the $i$th smallest element in~$I$, for all $i\in[k]$.
\begin{wrapfigure}{r}{0.26\textwidth}
\flushright
\begin{tikzpicture}[scale=0.4]
\mesh{5}{1,4,3,5,2}
\node at (-0.6,0.3) {\small $c_\alpha(0,0)$};
\node at (6.7,0.3) {\small $c_\alpha(5,0)$};
\node at (-0.6,5.6) {\small $c_\alpha(0,5)$};
\node at (6.7,5.6) {\small $c_\alpha(5,5)$};
\node at (1,1.6) {\tiny $\alpha(1)$};
\node at (2,4.6) {\tiny $\alpha(2)$};
\node at (3,3.6) {\tiny $\alpha(3)$};
\node at (4,5.6) {\tiny $\alpha(4)$};
\node at (5,2.6) {\tiny $\alpha(5)$};
\end{tikzpicture}
\end{wrapfigure}
The \emph{grid representation} of a permutation $\pi\in S_n$ is a mapping $\alpha:[n]\rightarrow \mathbb{Z}^2$ such that the sequence $\alpha(1)_x,\ldots,\alpha(n)_x$ is strictly increasing, and the sequence $\alpha(1)_y,\ldots,\alpha(n)_y$ is order-isomorphic to~$\pi$.
This representation is unique up to shifts that preserve the relative order of the values $\alpha(1)_y,\ldots,\alpha(n)_y$.
We can think of the grid representation of~$\pi$ as a graphical representation of the permutation matrix.
For instance, the permutation $\pi=14352$ has the grid representation shown on the right.
A \emph{cell} in the grid representation~$\alpha$ is a connected region from $\mathbb{R}^2\setminus \big(\alpha([n])_x\times \mathbb{R} \;\cup\; \mathbb{R}\times \alpha([n])_y\big)$, and we denote these cells by $c_\alpha(i,j)$, $i,j\in \{0,\ldots,n\}$, where the first index~$i$ increases with~$x$, and the second index~$j$ increases with~$y$, as shown in the figure on the right.
By definition, every cell is a Cartesian product of two open intervals.
For instance, we have $c_\alpha(1,0)=\left]\alpha(1)_x,\alpha(2)_x\right[\times\left]-\infty,\alpha(1)_y\right[$.
If the grid representation~$\alpha$ is clear from the context, we sometimes refer to a cell~$c_\alpha(i,j)$ simply by its index~$(i,j)$.

\begin{wrapfigure}{r}{0.2\textwidth}
\vspace{-7mm}
\flushright
\begin{tikzpicture}[scale=0.4]
\path [fill=gray!60] (0,1) rectangle (1,2);
\path [fill=gray!60] (4,3) rectangle (5,4);
\mesh{5}{1,4,3,5,2}
\node[anchor=west] at (0,-0.6) {$\sigma=(\tau,C)$};
\node[anchor=west] at (0,-1.68) {$\tau=14352$};
\node[anchor=west] at (0,-2.9) {$C=\{(0,1),(4,3)\}$};
\end{tikzpicture}
\vspace{-13mm}
\end{wrapfigure}
Mesh patterns were introduced by Br\"{a}nd\'{e}n and Claesson~\cite{MR2795782}, and they generalize all the aforementioned types of patterns.

A \emph{mesh pattern} is a pair $\sigma=(\tau,C)$, $\tau\in S_k$, with $C\seq \{0,\ldots,k\}^2$.
Each pair~$(a,b)\in C$ encodes the cell with index~$(a,b)$ in the grid representation of~$\tau$, and in our figures we visualize~$\sigma$ by drawing the grid representation of~$\tau$ and by shading exactly the cells indexed by~$C$.
For instance, the mesh pattern $\sigma=(\tau,C)=(14352,\{(0,1),(4,3)\})$ has the graphical representation shown on the right.

A permutation~$\pi\in S_n$ \emph{contains} the mesh pattern~$\sigma=(\tau,C)$, $\tau\in S_k$, if the grid representation~$\alpha$ of~$\pi$ admits a subset~$I\seq [n]$, $|I|=k$, such that $\beta:=\alpha|_I$ is the grid representation of~$\tau$, and the cell $c_\beta(i,j)$ contains no points from $\alpha([n])$ for all $(i,j)\in C$.
In this case, we refer to $\beta$ as a \emph{match} of the mesh pattern~$\sigma$ in the grid representation~$\alpha$ of~$\pi$.
Note that the first condition in the definition of mesh pattern containment is equivalent to requiring that the subpermutation of~$\pi$ on the indices in~$I$ is order-isomorphic to~$\tau$, while the second condition requires that the cells of~$C$ in the grid representation of~$\tau$ in~$\pi$ contain no points from~$\pi$.

\begin{wrapfigure}{r}{0.38\textwidth}
\vspace{-7mm}
\flushright
\begin{tikzpicture}[scale=0.4]
\mesh{5}{1,4,3,5,2}
\node[anchor=west] at (0.5,-0.6) {$\pi=14352$};
\end{tikzpicture}
\hspace{2mm}
\begin{tikzpicture}[scale=0.4]
\path [fill=gray!60] (4,3) rectangle (6,5);
\mesh{6}{1,5,3,6,4,2}
\node[anchor=west] at (0.5,-0.6) {$\pi'=153642$};
\node[marked mesh node] (1) at (1,1) {};
\node[marked mesh node] (1) at (2,5) {};
\node[marked mesh node] (1) at (3,3) {};
\node[marked mesh node] (1) at (4,6) {};
\node[marked mesh node] (1) at (6,2) {};
\end{tikzpicture}
\vspace{-5mm}
\end{wrapfigure}
For example, consider the grid representation~$\alpha$ of each of the two permutations~$\pi=14352$ and~$\pi'=153642$ shown on the right.
While $\pi$ clearly contains the mesh pattern~$\sigma$ from before, the permutation $\pi'$ avoids it, as the only choice for $I\seq[5]$ such that $\beta:=\alpha|_I$ is the grid representation of~$\tau$ is $I=\{1,2,3,4,6\}$ (marked points in the figure), but then the cell $c_\beta(4,3)$ (shaded gray) contains the point~$\alpha(5)$ (non-marked).

The following main theorem of this section implies all the lemmas about classical, vincular, barred patterns, etc.\ stated in the previous sections.

\begin{theorem}
\label{thm:mesh}
Let $\sigma=(\tau,C)$, $\tau\in S_k$, $k\geq 3$, be a mesh pattern, and let $i$ be the position of the largest value~$k$ in $\tau$.
If the pattern satisfies each of the following four conditions, then it is tame:
\begin{enumerate}[label=(\roman*),leftmargin=8mm, noitemsep, topsep=3pt plus 3pt]
\item $i$ is different from~1 and~$k$.
\item For all $a\in \{0,\ldots,k\}\setminus \{i-1,i\}$, we have $(a,k)\notin C$.
\item If $(i-1,k)\in C$, then for all $a\in \{0,\ldots,k\}\setminus\{i-1\}$ we have $(a,k-1)\notin C$ and for all $b\in \{0,\ldots,k-2\}$ we have that $(i,b)\in C$ implies $(i-1,b)\in C$.
\item If $(i,k)\in C$, then for all $a\in \{0,\ldots,k\}\setminus\{i\}$ we have $(a,k-1)\notin C$ and for all $b\in \{0,\ldots,k-2\}$ we have that $(i-1,b)\in C$ implies $(i,b)\in C$.
\end{enumerate}
\end{theorem}

The conditions in Theorem~\ref{thm:mesh} can be understood in the grid representation of~$(\tau,C)$ as follows; see the left hand side of Figure~\ref{fig:mesh}:
Condition~(i) asserts that the highest point of~$\tau$ must not be the leftmost or rightmost point (the two crossed out grid points in the figure are forbidden).
Condition~(ii) asserts that none of the cells in the topmost row (above the points) must be shaded, with the possible exception of the cells next to the highest point (solid crossed out cells in the figure).
Condition~(iii) asserts that if the cell~$(i-1,k)$ to the top left of the highest point is shaded (dark gray cell in the figure), then none of the cells in the row below except possibly~$(i-1,k-1)$ must be shaded (dotted crossed out cells), and if one of the cells strictly below~$(i,k-1)$ is shaded (dark gray questions marks), then the cell to the left of it must also be shaded (indicated by an arrow to the left).
Symmetrically, condition~(iv) asserts that if the cell~$(i,k)$ to the top right of the highest point is shaded (light gray cell in the figure), then none of the cells in the row below except possibly~$(i,k-1)$ must be shaded (dashed crossed out cells), and if one of the cells strictly below~$(i-1,k-1)$ is shaded (light gray questions marks), then the cell to the right of it must also be shaded (indicated by an arrow to the right).

To illustrate the conditions in Theorem~\ref{thm:mesh} further, consider the following mesh patterns~$\sigma_k=(3241,C_k)$, $k=1,\ldots,7$, with different sets~$C_k$:
\vspace{2mm}

\begin{center}
\begin{tikzpicture}[scale=0.4]
\path [fill=gray!60] (1,4) rectangle (2,5);
\mesh{4}{3,2,4,1}
\draw[mesh line] (1.1,4.1)--(1.9,4.9) (1.1,4.9)--(1.9,4.1);
\node[anchor=west] at (1.5,-0.6) {$\sigma_1$};
\end{tikzpicture} \hspace{1mm}
\begin{tikzpicture}[scale=0.4]
\path [fill=gray!60] (0,0) rectangle (1,2);
\path [fill=gray!60] (0,3) rectangle (2,4);
\path [fill=gray!60] (3,3) rectangle (4,4);
\path [fill=gray!60] (2,2) rectangle (3,3);
\path [fill=gray!60] (2,0) rectangle (5,1);
\mesh{4}{3,2,4,1}
\node[anchor=west] at (1.5,-0.6) {$\sigma_2$};
\end{tikzpicture} \hspace{1mm}
\begin{tikzpicture}[scale=0.4]
\path [fill=gray!60] (1,3) rectangle (2,4);
\path [fill=gray!60] (2,4) rectangle (3,5);
\mesh{4}{3,2,4,1}
\draw[mesh line] (1.1,3.1)--(1.9,3.9) (1.1,3.9)--(1.9,3.1);
\node[anchor=west] at (1.5,-0.6) {$\sigma_3$};
\end{tikzpicture} \hspace{1mm}
\begin{tikzpicture}[scale=0.4]
\path [fill=gray!60] (2,3) rectangle (3,5);
\path [fill=gray!60] (2,2) rectangle (4,3);
\path [fill=gray!60] (2,1) rectangle (3,2);
\path [fill=gray!60] (3,0) rectangle (4,1);
\mesh{4}{3,2,4,1}
\draw [->,>=stealth] (3.5,0.5) -- (2.5,0.5);
\node[anchor=west] at (1.5,-0.6) {$\sigma_4$};
\end{tikzpicture} \hspace{1mm}
\begin{tikzpicture}[scale=0.4]
\path [fill=gray!60] (2,3) rectangle (3,5);
\path [fill=gray!60] (2,1) rectangle (3,2);
\path [fill=gray!60] (2,0) rectangle (4,1);
\mesh{4}{3,2,4,1}
\node[anchor=west] at (1.5,-0.6) {$\sigma_5$};
\end{tikzpicture} \hspace{1mm}
\begin{tikzpicture}[scale=0.4]
\path [fill=gray!60] (2,4) rectangle (4,5);
\path [fill=gray!60] (2,1) rectangle (3,2);
\path [fill=gray!60] (2,0) rectangle (4,1);
\mesh{4}{3,2,4,1}
\draw [->,>=stealth] (2.5,1.5) -- (3.5,1.5);
\node[anchor=west] at (1.5,-0.6) {$\sigma_6$};
\end{tikzpicture} \hspace{1mm}
\begin{tikzpicture}[scale=0.4]
\path [fill=gray!60] (2,4) rectangle (4,5);
\path [fill=gray!60] (2,1) rectangle (5,2);
\path [fill=gray!60] (0,0) rectangle (4,1);
\mesh{4}{3,2,4,1}
\node[anchor=west] at (1.5,-0.6) {$\sigma_7$};
\end{tikzpicture}
\end{center}

All of these mesh patterns satisfy condition~(i).
The pattern~$\sigma_1$ violates condition~(ii) due to the cell~$(1,4)\in C_1$ (crossed in the figure), while all other patterns satisfy this condition.
The pattern~$\sigma_2$ satisfies conditions~(iii) and~(iv) trivially, as the premises of each of the two conditions are not satisfied.
The pattern~$\sigma_3$ violates the first part of condition~(iii) due to the cell~$(1,3)\in C_3$ (crossed).
The pattern~$\sigma_4$ satisfies the first part of condition~(iii), but violates the second part due to the cells~$(3,0)\in C_4$ and~$(2,0)\notin C_4$ (connected by an arrow in the figure).
The pattern~$\sigma_5$ satisfies conditions~(iii) and~(iv).
The pattern~$\sigma_6$ violates condition~(iv) due to the cells~$(2,1)\in C_6$ and $(3,1)\notin C_6$ (arrow).
Finally, the pattern~$\sigma_7$ satisfies conditions~(iii) and~(iv).

\begin{figure}
\includegraphics{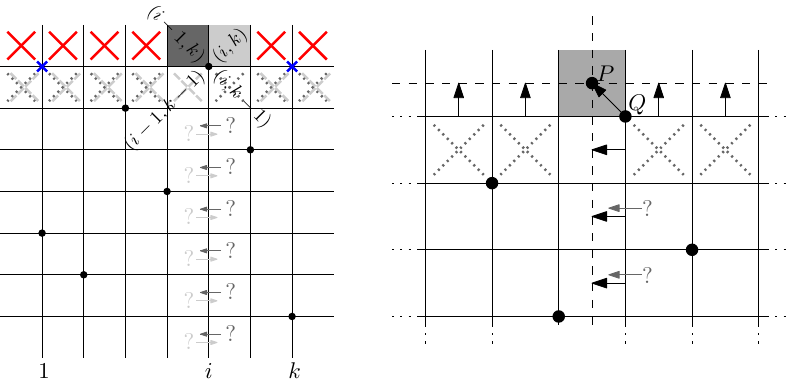}
\caption{Illustration of the four conditions in Theorem~\ref{thm:mesh} (left) and how they are used in the proof of the theorem (right).}
\label{fig:mesh}
\end{figure}

\begin{proof}
We show that if $\sigma=(\tau,C)$ satisfies the four conditions of the theorem, then $S_n(\sigma)$, $n\geq 0$, is a hereditary sequence of zigzag languages.
We argue by induction on~$n$.
Note that $S_0(\sigma)=S_0=\{\varepsilon\}$ is a zigzag language by definition, so the induction basis is clear.
For the induction step let $n\geq 1$.
We first show that if $\pi\in S_{n-1}(\sigma)$, then $c_1(\pi),c_n(\pi)\in S_n(\sigma)$.
As $c_1(\pi)$ and~$c_n(\pi)$ are obtained from~$\pi$ by inserting the new largest value~$n$ at the leftmost or rightmost position, respectively, the grid representation of these two permutations differs from the grid representation of~$\pi$ by adding a new highest point at the leftmost or rightmost position.
However, as $\pi$ avoids~$\sigma$ by assumption, condition~(i) guarantees that both $c_1(\pi)$ and~$c_n(\pi)$ also avoid~$\sigma$, which is what we wanted to show.

To complete the induction step, we now show that if $\pi\in S_n(\sigma)$, then $p(\pi)\in S_{n-1}(\sigma)$.
Recall that $p(\pi)$ is obtained from $\pi$ by removing the largest value~$n$, so in the grid representation, we remove the highest point~$P$.
Our assumption is that $\pi$ avoids the pattern~$\sigma$, and we need to show that removing the highest point does not create a match of the pattern~$\sigma$.
For the sake of contradiction, suppose that removing~$P$ creates a match of the pattern~$\sigma$ in~$p(\pi)$.
Let $Q$ be the highest point in this match of the pattern~$\sigma$ in~$p(\pi)$.
By condition~(ii), we are in exactly one of the following two symmetric cases:
(a) the cell~$(i-1,k)$ is in~$C$ and $P$ lies inside this cell of~$\sigma$ in this match of the pattern;
(b) the cell~$(i,k)$ is in~$C$ and $P$ lies inside this cell of~$\sigma$ in this match of the pattern.
We first consider case~(a), which is illustrated on the right hand side of Figure~\ref{fig:mesh}:
We claim that we can exchange the point~$Q$ for the point~$P$ in the match of the pattern~$\sigma$, and obtain another match of~$\sigma$ in~$\pi$, which would contradict the assumption that $\pi$ avoids~$\sigma$.
Indeed, this exchange operation strictly enlarges only the cells~$(a,k-1)$ for all $a\in\{0,\ldots,k\}\setminus\{i-1\}$ and the cells~$(i,b)$ for all $b\in\{0,\ldots,k-2\}$.
The first set of cells are not in~$C$ by the first part of condition~(iii).
The second set of cells are either not in~$C$, or if they are, then the corresponding cells to the left of it are also in~$C$ by the second part of condition~(iii).
Moreover, after the exchange the cell~$(i,k)$ contains no point from~$\pi$, as $P$ is the highest point (this is of course only relevant if $(i,k)\in C$).
Furthermore, after the exchange the cell~$(i-1,k-1)$ contains at most those points from~$\pi$ that were in the same cell before the exchange (clearly $P$ is the only point inside the cell~$(i-1,k)$).
So we indeed obtain a match of~$\sigma$ in $\pi$, a contradiction.

In the symmetric case~(b), we apply the same exchange argument, using condition~(iv) instead of~(iii).
This completes the proof.
\end{proof}

\subsection{Proof of Lemmas~\ref{lem:class}--\ref{lem:bivinc}}
\label{sec:mesh-proofs}

With Theorem~\ref{thm:mesh} in hand, the proofs of Lemmas~\ref{lem:class}--\ref{lem:bivinc} are straightforward.
As noted before, Lemma~\ref{lem:bruhat} generalizes Lemma~\ref{lem:class}, and Lemma~\ref{lem:bivinc} generalizes Lemma~\ref{lem:vinc}, so we only need to prove Lemmas~\ref{lem:barred}--\ref{lem:bivinc}.

\begin{proof}[Proof of Lemma~\ref{lem:barred}]
Note that a barred pattern~$\tau\in S_k$ with a single barred entry~$b$ at position~$a$ corresponds to the mesh pattern $\sigma=(\tau^-,\{(a-1,b-1)\})$, i.e., in the grid representation of $\sigma$ a single cell is shaded.
It follows that conditions~(iii) and (iv) of Theorem~\ref{thm:mesh} are trivially satisfied, and conditions~(i) and~(ii) translate into the conditions in the lemma.
\end{proof}

\begin{proof}[Proof of Lemma~\ref{lem:boxed}]
A boxed pattern~$\boxp{\tau}$, $\tau\in S_k$, corresponds to the mesh pattern $\sigma=(\tau,C)$ with $C:=\{(i,j)\mid 1\leq i,j\leq k-1\}$, i.e., in the grid representation of~$\sigma$, all cells inside the bounding box of the points from~$\tau$ are shaded.
It follows that conditions~(ii)--(iv) of Theorem~\ref{thm:mesh} are trivially satisfied, and condition~(i) corresponds exactly to the condition in the lemma.
\end{proof}

\begin{proof}[Proof of Lemma~\ref{lem:bruhat}]
A pattern with Bruhat restrictions~$(\tau,B)$ corresponds to the mesh pattern $\sigma=(\tau,C)$ where $C$ is the union of all the sets $R(a,b):=\{(i,j)\mid a\leq i<b\wedge \tau(a)\leq j<\tau(b)\}$ for $(a,b)\in B$, i.e., in the grid representation of~$\sigma$, certain rectangles of cells inside the bounding box of the points from~$\tau$ are shaded.
It follows that conditions~(ii)--(iv) of Theorem~\ref{thm:mesh} are trivially satisfied, and condition~(i) corresponds exactly to the condition in the lemma.
\end{proof}

\begin{proof}[Proof of Lemma~\ref{lem:bivinc}]
A vincular pattern~$\tau\in S_k$ where the entries at positions~$a$ and $a+1$ are underlined corresponds to the mesh pattern $\sigma=(\tau,C)$ with $C:=\{a\}\times\{0,\ldots,k\}$, i.e., in the grid representation of~$\sigma$, an entire column of cells is shaded.
For the bivincular pattern~$(\tau,B)$ we also have to add the sets $\{0,\ldots,k\}\times \{b\}$ for all $b\in B$ to the set of cells~$C$, i.e., in the grid representation we also have to shade the corresponding rows of cells.
By the conditions stated in Lemma~\ref{lem:vinc}, conditions~(i) and~(ii) of Theorem~\ref{thm:mesh} are satisfied.
By the condition~$k-1\notin B$, conditions~(iii) and~(iv) of the theorem are also satisfied, proving that the mesh pattern~$\sigma$ is tame.
\end{proof}

\subsection{Patterns with multiplicities}
\label{sec:mult}

All the aforementioned notions and results in this section generalize straightforwardly to bounding the number of appearances of a pattern.
Formally, a \emph{counted pattern} is a pair $\sigma=(\tau,c)$, where $\tau$ is a mesh pattern, and~$c$ is a non-negative integer.
Moreover, $S_n(\sigma)$ denotes the set of all permutations from~$S_n$ that contain \emph{at most}~$c$ matches of the pattern~$\tau$, where the special case $c=0$ is pattern-avoidance (cf.~\cite{MR1422065}).

By Theorem~\ref{thm:tame}, we can form propositional formulas~$F$ consisting of logical ANDs~$\wedge$, ORs~$\vee$, and tame counted patterns $(\tau_i,c_i)$ as variables, with possibly different counts $c_i$ for each variable.
The tameness of each~$(\tau_i,c_i)$ can be checked by verifying whether the patterns~$\tau_i$ satisfy the conditions stated in Theorem~\ref{thm:mesh} or its corollaries Lemmas~\ref{lem:class}--\ref{lem:bivinc}.
We then obtain a hereditary zigzag language~$S_n(F)$ that can be generated by Algorithm~J.

A somewhat contrived example for such a language would be $F=\big((231,3)\wedge (2143,5)\big)\vee (3\ul{14}2,2)$, the language of permutations that contain at most~3 matches of the pattern~$231$ AND at most~5 matches of the pattern~$2143$, OR at most~2 matches of the vincular pattern~$3\ul{14}2$.

\section{Algebra with patterns}
\label{sec:algebra}

In this section we significantly extend the methods described in the previous section, by applying geometric transformations to permutation patterns, and by describing some other types of patterns as conjunctions and disjunctions of suitable mesh patterns (recall Theorem~\ref{thm:tame}).
Two particularly relevant additional types of permutations covered in this section are monotone and geometric grid classes; see~Theorem~\ref{thm:geo-grid} below.

\subsection{Elementary transformations}
\label{sec:trans}

We now consider three important \emph{elementary transformations} of permutations that are important in the context of pattern-avoidance, as they preserve the cardinality of the set $S_n(F)$.
Each of them corresponds to a geometric transformation of the grid representation of each of the patterns~$\tau=a_1\cdots a_k$ in the formula~$F$, and together these transformations form the dihedral group $D_4$ of symmetries of a regular 4-gon:
\begin{itemize}[itemsep=0ex,parsep=0.5ex,leftmargin=2.5ex]
\item \emph{Reversal}, defined as $\rev(\tau):=a_k\cdots a_1$.
This corresponds to a vertical reflection of the grid representation.
\item \emph{Complementation}, defined as $\cpl(\tau)_i=k+1-a_i$ for all $i=1,\ldots,k$.
This corresponds to a horizontal reflection of the grid representation.
\item \emph{Inversion}, defined by $\inv(\tau)_{\tau(i)}=i$ for all $i=1,\ldots,k$.
This corresponds to a diagonal reflection of the grid representation along the south-west to north-east diagonal.
\end{itemize}
Note that a clockwise 90-degree rotation is obtained as $\rot(\tau):=\inv(\rev(\tau))=\cpl(\inv(\tau))$.
Clearly, all these operations generalize to mesh patterns~$(\tau,C)$, by applying the aforementioned geometric transformations to the grid representation of~$(\tau,C)$.
These operations and their relations are illustrated in Figure~\ref{fig:trans} for $(\tau,C)=(14352,\{(1,0),(1,1),(3,3),(4,3)\})$.

\begin{figure}
\includegraphics{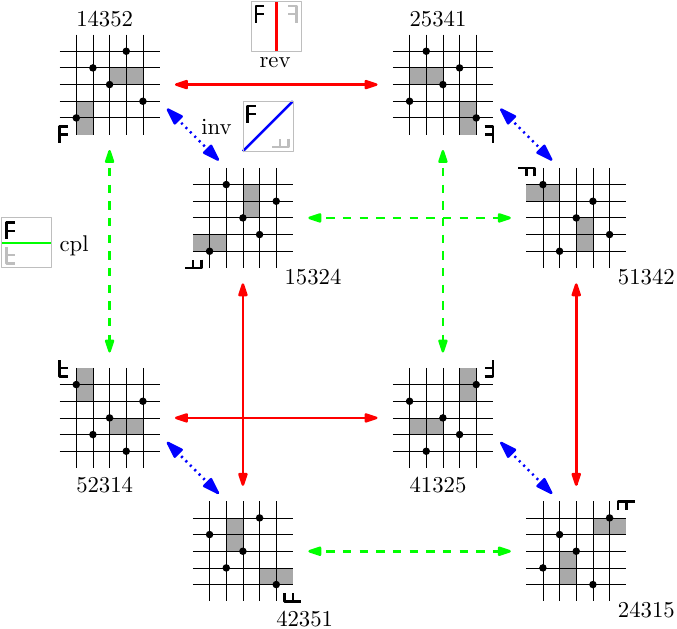}
\caption{Elementary transformations between permutations.}
\label{fig:trans}
\end{figure}

The following lemma is immediate.

\begin{lemma}
\label{lem:trans}
Given any composition~$h$ of the elementary transformations reversal, complementation and inversion, and any propositional formula~$F$ consisting of logical ANDs~$\wedge$, ORs~$\vee$, and mesh patterns $\tau_1,\ldots,\tau_\ell$ as variables, then the sets of pattern-avoiding permutations~$S_n(F)$ and~$S_n(h(F))$ are in bijection under~$h$ for all~$n\geq 1$, where the formula~$h(F)$ is obtained from~$F$ by replacing every pattern~$\tau_i$ by~$h(\tau_i)$ for all $i=1,\ldots,\ell$.
\end{lemma}

Lemma~\ref{lem:trans} is very useful for the purpose of exhaustive generation, because even if $\tau_i$ is not tame, then maybe $h(\tau_i)$ is.
So even if we cannot apply Algorithm~J to generate $S_n(\tau)$ directly, we may be able to generate $S_n(h(\tau_i))$, and then apply $h^{-1}$ to the resulting permutations.
For instance, $\tau=213$ is not tame, as the largest entry appears at the rightmost position.
However, $\cpl(\tau)=231$ is tame by Lemma~\ref{lem:class}, and so we can use Algorithm~J to generate $S_n(\cpl(\tau))$.

As another example, consider so-called \emph{2-stack sortable permutations} introduced by West~\cite{MR2716312} and later counted in~\cite{MR1168135, MR1401000, MR1630680}.
These permutations are characterized by the pattern-avoidance formula $F=\tau_1\wedge \tau_2$ with $\tau_1:=2341$ and $\tau_2:=3\ol{5}241$ ($\tau_2$ is a barred pattern).
Unfortunately, $\tau_2$ is not tame (the barred entry~$\ol{5}$ is not at a position next to the entry~4; recall Lemma~\ref{lem:barred}), so Algorithm~J cannot be used directly for generating~$S_n(F)$.
However, applying rotation, $h(\tau):=\rot(\tau)=\inv(\rev(\tau))$, yields two tame patterns $h(\tau_1)=1432$ and $h(\tau_2)=1352\ol{4}$ and the formula $h(F)=h(\tau_1)\wedge h(\tau_2)$, which can be used for generating $S_n(h(F))$ via Algorithm~J:
\vspace{2mm}

\begin{center}
\begin{tikzpicture}[scale=0.4]
\begin{scope}
\mesh{4}{2,3,4,1}
\node[anchor=west] at (0,-0.6) {$\tau_1=2341$};
\end{scope}
\begin{scope}[shift={(7,0)}]
\path [fill=gray!60] (1,4) rectangle (2,5);
\mesh{4}{3,2,4,1}
\node[anchor=west] at (0,-0.6) {$\tau_2=3\ol{5}241$};
\end{scope}
\node at (6,2) {$\wedge$};
\node at (24,2) {$\wedge$};
\draw[->] (13,2) -- node[above] {$h=\rot$} ++ (4,0);
\begin{scope}[shift={(18,0)}]
\mesh{4}{1,4,3,2}
\node[anchor=west] at (-1,-0.6) {$\rot(\tau_1)=1432$};
\end{scope}
\begin{scope}[shift={(25,0)}]
\path [fill=gray!60] (4,3) rectangle (5,4);
\mesh{4}{1,3,4,2}
\node[anchor=west] at (0,-0.6) {$\rot(\tau_2)=1352\ol{4}$};
\end{scope}
\end{tikzpicture}
\end{center}

Table~\ref{tab:tame2} lists several more permutations patterns that have been studied in the literature and that can be turned into tame patterns by such elementary transformations.

\begin{table}[h!]
\tikzstyle{mesh node}=[circle,minimum size=1.6pt,inner sep=0pt,draw=black,fill=black]
\small
\caption{Permutation patterns that become tame through elementary transformations, plus corresponding combinatorial objects.
}
\label{tab:tame2}
\renewcommand{\arraystretch}{0.9}
\begin{tabular}{|l|l|l|}
\hline
\textbf{Tame patterns} & \textbf{Combinatorial objects} & \textbf{References/OEIS} \cite{oeis} \\ \hline
$\rot(2341\!\wedge\! 3\ol{5}241)\!=\!1432\!\wedge\! 1352\ol{4}$: & rooted non-separable planar maps & \cite{MR2716312, MR1168135, MR1401000} \\
2-stack sortable permutations, & & \cite{MR1630680}, A000139 \\
$2413\wedge 41\ol{3}52$, $2413\wedge 45\ol{3}12$, & & \cite{MR1394948} \\ 
$2413\wedge 21\ol{3}54$, $3241\wedge \ol{2}4153$, & & \\
$\rot^{-1}(2413\wedge 5\ol{1}324)=2413\wedge 1534\ol{2}$, & & \\
$\cpl(2413\wedge \ol{4}2315)=3142\wedge\ol{2}4351$, & & \\
$\cpl(2314\wedge \ol{4}2513)=3241\wedge\ol{2}4153$, & & \\
$\cpl(3214\wedge \ol{2}4135)=2341\wedge \ol{4}2531$, & & \\
$\rev(2413\wedge 41\ol{3}52)=3142\wedge 2\ul{41}3$ & & \cite{MR2502604} \\ \hline
\multicolumn{2}{|l|}{conjunction of 20 patterns $\tau_i$ with tame $\cpl(\tau_i)$:} & \cite{MR2240774}, A245233 \\
\multicolumn{2}{|l|}{permutations generated by a stack of depth two and an infinite stack} & \\ \hline
\multicolumn{2}{|l|}{$\inv(132\wedge 312)=132\wedge 231$: Gilbreath permutations} & \cite{MR2028287,MR2858033} \\ \hline
\multicolumn{2}{|l|}{$\rot(3\ol{1}\ul{42}\wedge 3\ol{1}\ul{24})=\rot\Big(
\vcenter{\hbox{\begin{tikzpicture}[scale=0.16]
\path [fill=gray!60] (1,0) rectangle (2,1);
\path [fill=gray!60] (2,0) rectangle (3,4);
\path [fill=white] (0,0) rectangle (1,-0.3);
\path [fill=white] (0,4) rectangle (1,4.3);
\mesh{3}{2,3,1}
\end{tikzpicture}
}}
{}\!\wedge\,{}
\vcenter{\hbox{\begin{tikzpicture}[scale=0.16]
\path [fill=gray!60] (1,0) rectangle (2,1);
\path [fill=gray!60] (2,0) rectangle (3,4);
\mesh{3}{2,1,3}
\end{tikzpicture}}}
\Big)={}
\vcenter{\hbox{\begin{tikzpicture}[scale=0.16]
\path [fill=gray!60] (0,2) rectangle (1,3);
\path [fill=gray!60] (0,1) rectangle (4,2);
\mesh{3}{1,3,2}
\end{tikzpicture}}}
{}\,\wedge\,{}
\vcenter{\hbox{\begin{tikzpicture}[scale=0.16]
\path [fill=gray!60] (0,2) rectangle (1,3);
\path [fill=gray!60] (0,1) rectangle (4,2);
\mesh{3}{2,3,1}
\end{tikzpicture}}}
$\hspace{1ex}:} & \cite{MR2710744,MR2732833}, A129698 \\
\multicolumn{2}{|l|}{permutations that uniquely encode pile configurations in patience sorting} & \\ \hline
\end{tabular}
\end{table}

\subsection{Partially ordered patterns}

Partially ordered patterns were introduced by Kitaev~\cite{MR2163450}.
A \emph{partially ordered pattern (POP)} is a partially ordered set $P=([k],\prec)$, and we say that a permutation~$\pi$ \emph{contains} this pattern if and only if it contains a subpermutation $a_{i_1}\cdots a_{i_k}$, $i_1<\cdots<i_k$, such that $k\prec l$ in the partial order implies that $a_{i_k}<a_{i_l}$.
In particular, if $\prec$ is a linear order, then this is equivalent to classical pattern avoidance.
However, some other constraints can be expressed much more conveniently using POPs.
For instance, avoiding the POP
\begin{center}
\begin{tikzpicture}
\node at (-1,-0.4) {$P_1=$};
\node[poset node] (2) at (0,0) {2};
\node[poset node] (1) at (-0.3,-0.8) {1};
\node[poset node] (3) at (0.3,-0.8) {3};
\draw (2) -- (1);
\draw (2) -- (3);
\end{tikzpicture}
\end{center}
is equivalent to avoiding peaks in the permutation, so $S_n(P_1)$ is the set of permutations without peaks discussed before, which satisfies~$|S_n(P_1)|=2^{n-1}$.

More generally, the POP~$P_k$ defined by
\begin{center}
\begin{tikzpicture}
\node at (-1,-0.4) {$P_1=$};
\node[poset node] (2) at (0,0) {2};
\node[poset node] (1) at (-0.3,-0.8) {1};
\node[poset node] (3) at (0.3,-0.8) {3};
\draw (2) -- (1);
\draw (2) -- (3);
\end{tikzpicture}\raisebox{5mm}[5mm][0mm]{,}\hspace{1mm}
\begin{tikzpicture}
\node at (-1,-0.4) {$P_2=$};
\node[poset node] (2) at (0,0) {2};
\node[poset node] (4) at (0.6,0) {4};
\node[poset node] (1) at (-0.3,-0.8) {1};
\node[poset node] (3) at (0.3,-0.8) {3};
\node[poset node] (5) at (0.9,-0.8) {5};
\draw (2) -- (1);
\draw (2) -- (3);
\draw (4) -- (3);
\draw (4) -- (5);
\end{tikzpicture}\raisebox{5mm}[5mm][0mm]{,}\hspace{1mm}
\begin{tikzpicture}
\node at (-1,-0.4) {$P_3=$};
\node[poset node] (2) at (0,0) {2};
\node[poset node] (4) at (0.6,0) {4};
\node[poset node] (6) at (1.2,0) {6};
\node[poset node] (1) at (-0.3,-0.8) {1};
\node[poset node] (3) at (0.3,-0.8) {3};
\node[poset node] (5) at (0.9,-0.8) {5};
\node[poset node] (7) at (1.5,-0.8) {7};
\draw (2) -- (1);
\draw (2) -- (3);
\draw (4) -- (3);
\draw (4) -- (5);
\draw (6) -- (5);
\draw (6) -- (7);
\end{tikzpicture}\raisebox{5mm}[5mm][0mm]{$,\ldots,$}\hspace{1mm}
\begin{tikzpicture}
\node at (-1,-0.4) {$P_k=$};
\node[poset node] (2) at (0,0) {2};
\node[poset node] (4) at (0.6,0) {4};
\node[poset node] (1) at (-0.3,-0.8) {1};
\node[poset node] (3) at (0.3,-0.8) {3};
\node[poset node] (5) at (0.9,-0.8) {5};
\node at (1.35,-0.8) {$\ldots$};
\node[poset node] (20) at (3.0,0) {$2k$};
\node[poset node] (10) at (2.3,-0.8) {$2k-1$};
\node[poset node] (30) at (3.7,-0.8) {$2k+1$};
\draw (2) -- (1);
\draw (2) -- (3);
\draw (4) -- (3);
\draw (4) -- (5);
\draw (20) -- (10);
\draw (20) -- (30);
\end{tikzpicture}
\end{center}
realizes the language~$S_n(P_k)$ of permutations with at most $k-1$ peaks.

We let $L(P)$ denote the set of all linear extensions of the poset~$P$, and for any linear extension~$x\in L(P)$, we consider the inverse permutation of~$x$, as the $i$th entry of~$\inv(x)$ denotes the position of~$i$ in~$x$.
Moreover, $\inv(x)\in S_k$, so $\inv(x)$ is a classical pattern.

\begin{lemma}
\label{lem:pop}
For any partially ordered pattern~$P=([k],\prec)$, we have
\begin{equation}
S_n(P)=\bigcap_{x\in L(P)} S_n(\inv(x))=S_n\Big(\bigwedge\nolimits_{x\in L(P)}\inv(x)\Big).
\end{equation}
In particular, if the poset~$P$ does not have~1 or~$k$ as a maximal element, then $P$ is tame.
\end{lemma}

\begin{proof}
The first part of the lemma follows immediately from the definition of POPs and from~\eqref{eq:zigzag-ops}.
To prove the second part, suppose that $P$ does not have~1 or~$k$ as a maximal element.
Then in any linear extension~$x\in L(P)$, 1 and $k$ will not appear at the last position, and so in the inverse permutation~$\inv(x)$, the largest entry~$k$ will neither be at position~1 nor at position~$k$.
We can hence apply Lemma~\ref{lem:class}, and using Theorem~\ref{thm:tame} we obtain that $P$ is tame.
\end{proof}

For instance, for the POP~$P_1$ from before we have $L(P_1)=\{132,312\}$, and so $P_1=132\wedge 231$, and for the POP~$P_2$ we have $L(P_2)=\{13254,13524,13542,15324,15342,31254,\ldots\}$, a set of 16 linear extensions in total, so $P_2=13254\wedge 14253\wedge 15243\wedge 14352\wedge 15342\wedge 23145\wedge\cdots$.

Moreover, we can create counted POPs with multiplicity~$c$ (recall Section~\ref{sec:mult}), by taking the OR of conjunctions of counted classical patterns as described by Lemma~\ref{lem:pop}, over all number partitions of~$c$ into the corresponding number of parts.
For instance, the counted POP~$\sigma=(P_1,c)$, which realizes the zigzag language $S_n(\sigma)$ of permutations with at most $c$ triples of values forming a peak, is obtained by considering the partitions $c=c+0=(c-1)+1=\cdots=1+(c-1)=0+c$, resulting in the formula
\begin{equation*}
 (P_1,c)=\big((132,c)\wedge (231,0)\big)\vee \big((132,c-1)\wedge (231,1)\big)\vee \cdots \vee\big((132,0)\wedge (231,c)\big)
\end{equation*}
with counted classical patterns on the right-hand side.

\subsection{Barred patterns with multiple bars}
\label{sec:mbarred}

Some patterns with multiple bars can be reduced to single-barred patterns (to which Lemma~\ref{lem:barred} applies) as shown by the following lemma.

\begin{lemma}[cf.~\cite{MR2924751}]
\label{lem:mbarred}
Let $\tau\in S_k$, $k\geq 5$, be a pattern with $b\geq 2$ bars, such that no two barred entries are at neighboring positions or have adjacent values.
Let $\widetilde{\tau}_1,\ldots,\widetilde{\tau}_b\in S_{k-b+1}$ be the permutations with a single barred entry that are order-isomorphic to the sequences obtained from $\tau$ by removing all but one barred entry.
Then we have
\begin{equation*}
S_n(\tau)=\bigcap_{1\leq i\leq b} S_n(\widetilde{\tau}_i)=S_n\Big(\bigwedge\nolimits_{1\leq i\leq b} \widetilde{\tau}_i \Big).
\end{equation*}
Consequently, if $\tau^-\in S_{k-b}$ does not have the largest value~$k-b$ at the leftmost or rightmost position, and the largest barred entry in~$\tau$ is smaller than~$k$ or at a position next to the entry~$k-1$, then $\tau$ is tame.
\end{lemma}

\begin{proof}
To prove the first part, observe that when no two barred entries are at neighboring positions or have adjacent values, then the definition of barred pattern avoidance is equivalent to avoiding each of the single-barred patterns $\widetilde{\tau}_1,\ldots,\widetilde{\tau}_b$, so the claim follows using~\eqref{eq:zigzag-ops}.

To prove the second part we show that each of the single-barred patterns $\widetilde{\tau}_1,\ldots,\widetilde{\tau}_b$ satisfies the conditions of Lemma~\ref{lem:barred}.
Indeed, we know that $\tau^-=(\widetilde{\tau}_i)^-\in S_{k-b}$, $1\leq i\leq b$, does not have the largest value~$k-b$ at the leftmost or rightmost position.
Moreover, if $\widetilde{\tau}_i$ is obtained from $\tau$ by removing all but the largest barred entry, then the barred entry in $\widetilde{\tau}_i\in S_{k-b+1}$ is either smaller than $k-b+1$ or at a position next to the entry $k-b$.
To see this note that if the largest entry~$k$ in~$\tau$ is barred, then the second largest entry~$k-1$ is not barred by the assumption that no two barred entries have adjacent values.
For the same reason, if $\widetilde{\tau}_i$ is obtained from $\tau$ by removing barred entries including the largest one, then the barred entry in $\widetilde{\tau}_i\in S_{k-b+1}$ is smaller than $k-b+1$.
Consequently, we can apply Lemma~\ref{lem:barred} to each of the patterns $\widetilde{\tau}_1,\ldots,\widetilde{\tau}_b$, and complete the proof by applying Theorem~\ref{thm:tame}.
\end{proof}

Lemma~\ref{lem:mbarred} applies for instance to the tame pattern~$3\ol{1}52\ol{4}=3\ol{1}42\wedge 241\ol{3}$ listed in Table~\ref{tab:tame1}.

\subsection{Weak avoidance of barred patterns and dotted patterns}
\label{sec:weak}

Weak pattern avoidance and dotted patterns were introduced by Baril~\cite{MR2836813} (see also~\cite{MR3147205}).
Given a single-barred pattern~$\tau$, we define $\tau'$ and $\tau^-$ as in Section~\ref{sec:barred}, and we say that a permutation~$\pi$ \emph{weakly contains}~$\tau$, if and only if it contains a match of~$\tau^-$ that cannot be extended to a match of~$\tau'$ by adding one entry of~$\pi$, not necessarily at the position specified by the barred entry.
Otherwise, we say that $\pi$ \emph{weakly avoids} $\tau$.
We let $S_n^w(\tau)$ denote the set of permutations that weakly avoid~$\tau$.
For instance, $\pi=1243$ contains the barred pattern $\tau=12\ol{3}$, as for the increasing pair~$24$ in~$\pi$, there is no entry in~$\pi$ to the right that extends this pair to an increasing triple.
However, $\pi$ weakly avoids~$\tau$, as each of the increasing pairs~$12$, $14$, $13$, $24$ and~$23$ can be extended to an increasing triple, by adding an entry from~$\pi$ to the right, middle, middle, left and left of this pair, respectively.

The relation between weak pattern avoidance and mesh pattern avoidance is captured by the following lemma.

\begin{lemma}
\label{lem:weak}
Let $\tau\in S_k$ be a single-barred pattern with barred entry~$b$.
Consider the longest increasing or decreasing substring of consecutive values in~$\tau'$ that contains~$b$, and let~$r$ and~$s$ be the start and end indices of this substring.
Let $\sigma$ be the mesh pattern defined by $\sigma:=(\tau^-,C)$ with
\begin{equation}
\label{eq:Cweak}
C := \big\{\big(r+i-1,\tau'(r+i)-1\big)\mid 0\leq i\leq s-r\big\}.
\end{equation}
Then we have $S_n^w(\tau)=S_n(\sigma)$.
\end{lemma}

Table~\ref{tab:weak} illustrates the mesh pattern~$\sigma=(\tau^-,C)$ defined in Lemma~\ref{lem:weak} for six different single-barred patterns.

\begin{table}
\tikzstyle{mesh node}=[circle,minimum size=2.4pt,inner sep=0pt,draw=black,fill=black]
\caption{Illustration of Lemma~\ref{lem:weak}.}
\label{tab:weak}
\begin{tabular}{cccllc}
\textbf{$\tau$} & \textbf{$\tau'$} & \textbf{$\tau^-$} & $r$ & $s$ & $\sigma$ \\ \hline
$\ol{1}45632$ & 145632 & 34521 & 1 & 1 &
{$\vcenter{\hbox{\begin{tikzpicture}[scale=0.2]
\path [fill=gray!60] (0,0) rectangle (1,1);
\mesh{5}{3,4,5,2,1}
\end{tikzpicture}
}}$}
\\ [-1em]
\\ \hline
\\ [-1em]
$1\ol{4}5632$ & 145632 & 14532 & 2 & 4 & \multirow{3}{*}{
{$\vcenter{\hbox{\begin{tikzpicture}[scale=0.2]
\path [fill=gray!60] (1,3) rectangle (2,4);
\path [fill=gray!60] (2,4) rectangle (3,5);
\path [fill=gray!60] (3,5) rectangle (4,6);
\mesh{5}{1,4,5,3,2}
\end{tikzpicture}
}}$}
}
\\
$14\ol{5}632$ & 145632 & 14532 & 2 & 4 & \\
$145\ol{6}32$ & 145632 & 14532 & 2 & 4 & \\ \hline
& & & & & \multirow{3}{*}{
{$\vcenter{\hbox{\begin{tikzpicture}[scale=0.2]
\path [fill=gray!60] (4,2) rectangle (5,3);
\path [fill=gray!60] (5,1) rectangle (6,2);
\mesh{5}{1,3,4,5,2}
\end{tikzpicture}
}}$}
}
\\ [-0.7em]
$1456\ol{3}2$ & 145632 & 13452 & 5 & 6 &
\\
$14563\ol{2}$ & 145632 & 13452 & 5 & 6 & \\
\end{tabular}
\end{table}

\begin{proof}
We only consider the case that the longest substring of consecutive values in~$\tau'$ that contains~$b$ is increasing, as the other case is symmetric.

We first show that if a permutation~$\pi$ weakly contains~$\tau$, then it also contains~$\sigma$.
For this consider a match of~$\tau^-$ that cannot be extended to a match of~$\tau'$ in the grid representation of~$\pi$; see the left hand side of Figure~\ref{fig:weak}.
Consider the $s-r$ points $P_r,\ldots,P_{s-1}$ of~$\pi$ to which the entries at positions~$r,\ldots,s-1$ of~$\tau^-$ are matched.
We know that they form an increasing sequence, and all other points in this match are below or above them.
Now consider the $s-r+1$ cells in~$C$ defined in~\eqref{eq:Cweak} between and around these points.
We need to show that none of them contains any points of~$\pi$, demonstrating that this is a match of the mesh pattern~$\sigma$.
Indeed, if one of these regions did contain a point~$Q$ from~$\pi$, then $Q$ together with $P_r,\ldots,P_{s-1}$ would form an increasing sequence of length $s-r+1$, i.e., $Q$ would extend the match of~$\tau^-$ to a match of~$\tau'$ in~$\pi$, a contradiction.

It remains to show that if a permutation~$\pi$ contains the mesh pattern~$\sigma$, then it weakly contains~$\tau$.
For this consider a match of $\sigma=(\tau^-,C)$ in the grid representation of~$\pi$; see the right hand side of Figure~\ref{fig:weak}.
We label the points to which the entries of~$\tau^-$ are matched by~$P_1,\ldots,P_{k-1}$.
By the definition of~$\sigma$, the points $P_r,\ldots,P_{s-1}$ form a longest increasing sequence of consecutive values in this match.
In particular, $P_{r-1}$ is not located at the bottom left corner of the leftmost cell of~$C$, and $P_s$ is not located at the top right corner of the rightmost cell of~$C$, so there is a point~$P_a$, $a\in[k-1]\setminus\{r-1,\ldots,s-1\}$, on the same height as the lower boundary of the leftmost cell of~$C$, and a point~$P_b$, $b\in[k-1]\setminus\{r,\ldots,s\}$, on the same height as the upper boundary of the rightmost cell of~$C$.
We know that no point of~$\pi$ lies within any of the cells in~$C$, but we also need to show that a point~$Q$ of~$\pi$ contained in any of the other cells cannot be used to extend this match of~$\tau^-$ to to a match of~$\tau'$.
For this we distinguish four cases:
Suppose that $Q$ is contained in a cell~$L$ to the left of the cells in~$C$.
Then $P_1,\ldots,P_{k-1}$ together with~$Q$ is not a match of~$\tau'$, as the longest increasing substring of consecutive values $P_r,\ldots,P_{s-1}$ contains only $s-r$ points.
A symmetric argument works if $Q$ is contained in a cell~$R$ to the right of the cells in~$C$.
Now suppose that~$Q$ is contained in a cell~$A$ above a cell from~$C$, but not above the rightmost one, or below a cell from~$C$, but not below the leftmost one.
In this case the points $P_r,\ldots,P_{s-1}$ together with~$Q$ do not form an increasing substring, so this is not a match of~$\tau'$.
It remains to consider the case that~$Q$ is contained in a cell~$B$ above the rightmost cell from~$C$, or below the leftmost cell from~$C$.
In this case the points $P_r,\ldots,P_{s-1}$ together with~$Q$ form an increasing substring, but the values are not consecutive, as $Q$ is separated from $P_r,\ldots,P_{s-1}$ by the point~$P_b$ or $P_a$, respectively, so this is not match of~$\tau'$ either.
This completes the proof.

\begin{figure}
\includegraphics{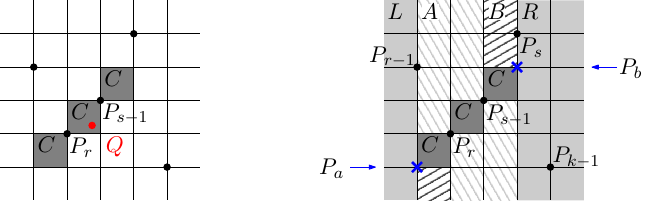}
\caption{Illustration of the proof of Lemma~\ref{lem:weak}.}
\label{fig:weak}
\end{figure}

\end{proof}

Baril~\cite{MR2836813} also introduced the notion of a \emph{dotted} pattern~$\tau\in S_k$, which is a pattern with some entries at positions $I\seq [k]$ that have a dot above them.
A permutation~$\pi$ \emph{avoids} the pattern~$\tau$ if and only if $\pi$ weakly avoids every single-barred pattern obtained by putting a bar above every entry not in~$I$, i.e., we have
\begin{equation}
\label{eq:dotted}
  S_n(\tau):=\bigcap_{j\in [k]\setminus I} S_n^w(\tau_j),
\end{equation}
where $\tau_j$ is the barred pattern obtained by putting the bar above entry~$j$.

\begin{lemma}
Let $\tau\in S_k$, $k\geq 3$, be a dotted pattern with dots over all positions~$I\seq [k]$.
If $\tau$ does not have the largest value~$k$ at the leftmost or rightmost position, and all entries in the longest increasing or decreasing substring of consecutive values including~$k$ are dotted, then $\tau$ is tame.
\end{lemma}

\begin{proof}
Let~$r$ and~$s$ be the start and end indices of the longest increasing or decreasing substring of consecutive values including~$k$.
The conditions of the lemma imply that $[r,s] \seq I$.

By~\eqref{eq:dotted}, Theorem~\ref{thm:tame}, and Lemma~\ref{lem:weak}, to prove that $\tau$ is tame it is sufficient to show this for the mesh pattern $\sigma=(\rho^-,C)$ with $C$ defined in~\eqref{eq:Cweak} for every single-barred pattern~$\rho$ obtained from $\tau'$ by placing a bar over an entry at a position $[k]\setminus I$.
As $[r,s]\seq I$, we obtain in particular that the largest value~$k$ in~$\rho$ is not barred, and if the second largest entry~$k-1$ is next to~$k$, then it is also not barred.
Combining this with the assumption that the largest value~$k$ is not at the leftmost or rightmost position in~$\rho$, we obtain that in~$\rho^-$, the largest entry $k-1$ is not at the leftmost or rightmost position.
Moreover, we obtain from the definition~\eqref{eq:Cweak} that $C$ does not have any cells in the topmost row, i.e., $(i,k-1)\notin C$ for $i=0,\ldots,k-1$.
Therefore, applying Theorem~\ref{thm:mesh} shows that $\sigma$ is tame, completing the proof.
\end{proof}

\subsection{Monotone and geometric grid classes}
\label{sec:grid}

Monotone grid classes of permutations were introduced by Huczynska and Vatter~\cite{MR2240760}.
To define them, we consider a matrix~$M$ with entries from~$\{0,+1,-1\}$, indexed first by columns from left to right, and then by rows from bottom to top.
A permutation~$\pi$ of~$[n]$ (for any~$n\geq 0$) is in the \emph{monotone grid class of~$M$}, denoted~$\Grid(M)$, if we can place the points labelled from~1 to~$n$ from bottom to top into a rectangular grid that has as many rows and columns as the matrix~$M$, and reading the labels from left to right will yield~$\pi$, subject to the following two conditions:
No two points are placed on the same horizontal or vertical line.
Moreover, for each cell~$(x,y)$ in the grid, if $M_{x,y}=0$, then the cell contains no points, if $M_{x,y}=+1$, then the points in this cell are increasing, and if $M_{x,y}=-1$, then the points in this cell are decreasing.
This definition is illustrated in the top part of Figure~\ref{fig:grid-geo}.
Based on this, we define $\Grid_n(M):=\Grid(M)\cap S_n$.
It is an open problem whether $\Grid(M)$ is characterized by finitely many forbidden patterns for any~$M$ (cf.~\cite{MR3548800}).

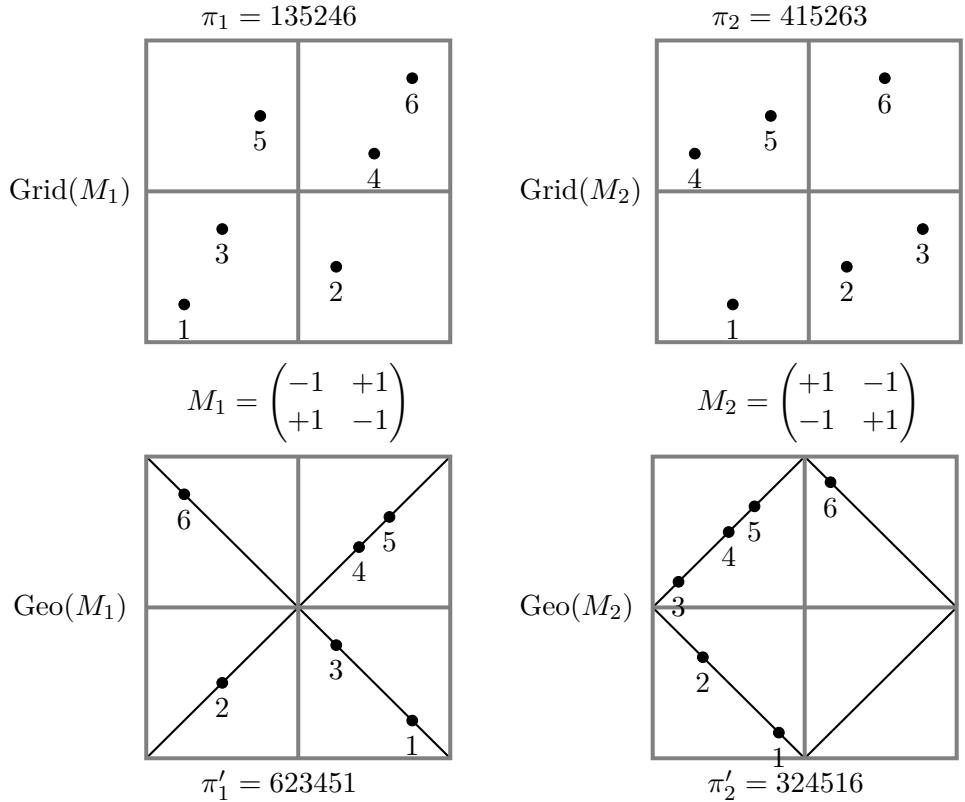
\begin{figure}
\begin{tikzpicture}
\path [grid edge] (0,0) rectangle (4,4);
\path [grid edge] (2,0) -- (2,4);
\path [grid edge] (0,2) -- (4,2);
\draw (0.5,0.5) node[mesh node,label=below:$1$]{};
\draw (2.5,1) node[mesh node,label=below:$2$]{};
\draw (1,1.5) node[mesh node,label=below:$3$]{};
\draw (3,2.5) node[mesh node,label=below:$4$]{};
\draw (1.5,3) node[mesh node,label=below:$5$]{};
\draw (3.5,3.5) node[mesh node,label=below:$6$]{};
\node at (2,-0.8) {$M_1=\begin{pmatrix} -1 & +1 \\ +1 & -1 \end{pmatrix}$};
\node at (1.75,4.3) {$\pi_1=135246$};
\node at (-1,2) {$\Grid(M_1)$};
\end{tikzpicture}
\hspace{5mm}
\begin{tikzpicture}
\path [grid edge] (0,0) rectangle (4,4);
\path [grid edge] (2,0) -- (2,4);
\path [grid edge] (0,2) -- (4,2);
\draw (1,0.5) node[mesh node,label=below:$1$]{};
\draw (2.5,1) node[mesh node,label=below:$2$]{};
\draw (3.5,1.5) node[mesh node,label=below:$3$]{};
\draw (0.5,2.5) node[mesh node,label=below:$4$]{};
\draw (1.5,3) node[mesh node,label=below:$5$]{};
\draw (3,3.5) node[mesh node,label=below:$6$]{};
\node at (2,-0.8) {$M_2=\begin{pmatrix} +1 & -1 \\ -1 & +1 \end{pmatrix}$};
\node at (1.75,4.3) {$\pi_2=415263$};
\node at (-1,2) {$\Grid(M_2)$};
\end{tikzpicture}
\\
\begin{tikzpicture}
\path [slope edge] (0,0) -- (4,4);
\path [slope edge] (0,4) -- (4,0);
\path [grid edge] (0,0) rectangle (4,4);
\path [grid edge] (2,0) -- (2,4);
\path [grid edge] (0,2) -- (4,2);
\draw (3.5,0.5) node[mesh node,label=below:$1$]{};
\draw (1,1) node[mesh node,label=below:$2$]{};
\draw (2.5,1.5) node[mesh node,label=below:$3$]{};
\draw (2.8,2.8) node[mesh node,label=below:$4$]{};
\draw (3.2,3.2) node[mesh node,label=below:$5$]{};
\draw (0.5,3.5) node[mesh node,label=below:$6$]{};
\node at (1.75,-0.35) {$\pi_1'=623451$};
\node at (-1,2) {$\Geo(M_1)$};
\end{tikzpicture}
\hspace{5mm}
\begin{tikzpicture}
\path [slope edge] (0,2) -- (2,4);
\path [slope edge] (0,2) -- (2,0);
\path [slope edge] (2,4) -- (4,2);
\path [slope edge] (2,0) -- (4,2);
\path [grid edge] (0,0) rectangle (4,4);
\path [grid edge] (2,0) -- (2,4);
\path [grid edge] (0,2) -- (4,2);
\draw (1.66,0.34) node[mesh node,label=below:$1$]{};
\draw (0.66,1.34) node[mesh node,label=below:$2$]{};
\draw (0.34,2.34) node[mesh node,label=below:$3$]{};
\draw (1,3) node[mesh node,label=below:$4$]{};
\draw (1.34,3.34) node[mesh node,label=below:$5$]{};
\draw (2.34,3.66) node[mesh node,label=below:$6$]{};
\node at (1.75,-0.35) {$\pi_2'=324516$};
\node at (-1,2) {$\Geo(M_2)$};
\end{tikzpicture} \\
\caption{Top: Illustration of monotone grid classes. Bottom: Illustration of geometric grid classes, with X-shaped permutations on the left, and O-shaped permutations on the right.
Observe that $\pi_1\in\Grid(M_1)$, but $\pi_1\notin\Geo(M_1)$, as $\pi_1$ contains the pattern~$2413$.
Similarly, we have $\pi_2\in\Grid(M_2)$, but $\pi_2\notin\Geo(M_2)$, as $\pi_2$ contains the pattern~$1423$.}
\label{fig:grid-geo}
\end{figure}

The second type of permutations we shall discuss in this section are geometric grid classes, introduced by Albert, Atkinson, Bouvel, Ru\v{s}kuc, and Vatter~\cite{MR3091268}.
They are defined using a matrix~$M$ with entries from~$\{0,+1,-1\}$ as before.
A permutation~$\pi$ of~$[n]$ is in the \emph{geometric grid class of~$M$}, denoted~$\Geo(M)$, if it can be drawn in a rectangular grid as described before, with the slightly strengthened conditions that if $M_{x,y}=+1$, then the points in the cell~$(x,y)$ lie on the increasing diagonal line through this cell, and if $M_{x,y}=-1$, then the points in the cell~$(x,y)$ lie on the decreasing diagonal line through this cell.
This definition is illustrated in the bottom part of Figure~\ref{fig:grid-geo}.
Similarly to before, we define~$\Geo_n(M):=\Geo(M)\cap S_n$.
We clearly have $\Geo(M)\seq \Grid(M)$ and $\Geo_n(M)\seq \Grid_n(M)$.

Unlike for monotone grid classes, it was shown in~\cite{MR3091268} that any geometric grid class~$\Geo(M)$ is characterized by finitely many forbidden patterns, i.e.,
\begin{equation*}
  \Geo_n(M)=S_n(\tau_1\wedge \cdots\wedge \tau_\ell)
\end{equation*}
for a suitable set of patterns $\tau_1,\ldots,\tau_\ell$ and for all $n\geq 0$.
For instance, X-shaped permutations studied in~\cite{waton_2007,MR2785755} are exactly the permutations in $S_n(2143\wedge 2413\wedge 3142\wedge 3412)$.
However, the argument given in~\cite{MR3091268} for the existence of~$\tau_1,\ldots,\tau_\ell$ is non-constructive, so there is no procedure known to compute these patterns from the matrix~$M$.

Nevertheless, our next theorem provides an easily verifiable sufficient condition for deciding whether~$\Grid_n(M)$ and~$\Geo_n(M)$ are zigzag languages, based only on two particular entries of~$M$.

\begin{theorem}
\label{thm:geo-grid}
Let $M$ be a matrix with $-1$ in the top-left corner and $+1$ in the top-right corner.
Then $\Grid_n(M)$, $n\geq 0$, and $\Geo_n(M)$, $n\geq 0$, are both hereditary sequences of zigzag languages.
Consequently, all of these languages can be generated by Algorithm~J.
\end{theorem}

From the two monotone and geometric grid classes shown in Figure~\ref{fig:grid-geo}, only the left two satisfy the conditions of the theorem.

\begin{proof}
We only prove the theorem for monotone grid classes~$\Grid_n(M)$.
The argument for geometric grid classes~$\Geo_n(M)$ is completely analogous.

We argue by induction on~$n$.
Note that $\Grid_0(M)=S_0=\{\varepsilon\}$ is a zigzag language by definition, so the induction basis is clear.
For the induction step let $n\geq 1$.
We first show that if $\pi\in\Grid_{n-1}(M)$, then $c_1(\pi),c_n(\pi)\in \Grid_n(M)$.
For this argument we use the assumption that the top-left entry of~$M$ is~$-1$, and the top-right entry of~$M$ is~$+1$, i.e., $\pi$ can be drawn into a grid so that the points in the top-left cell are decreasing, and the points in the top-right cell are increasing.
It follows that we can draw $c_1(\pi)$ on the same grid, by extending the drawing of~$\pi$ so that the new point~$n$ is placed to the top-left of all other points.
Similarly, we can draw $c_n(\pi)$ on the same grid, by extending the drawing of~$\pi$ so that the new point~$n$ is placed to the top-right of all other points.

To complete the induction step, we now show that if $\pi\in\Grid_n(M)$, then $p(\pi)\in\Grid_{n-1}(M)$.
As $\pi\in\Grid_n(M)$, we can draw $\pi$ into a grid respecting the monotonicity conditions described by~$M$.
Clearly, removing any entry from~$\pi$, in particular the largest one, maintains this property, i.e., we can draw $p(\pi)$ on the same grid, showing that $p(\pi)\in\Grid_{n-1}(M)$.
This completes the proof.
\end{proof}

\section{Limitations of our approach}
\label{sec:limits}

Theorem~\ref{thm:jump} asserts that if $L_n$ is a zigzag language, then Algorithm~J successfully visits every permutation from~$L_n$.
This condition is not necessary, however.
For instance, the set $L_4\seq S_4$ discussed in Section~\ref{sec:algo} is not a zigzag language, and still Algorithm~J is successful when initialized suitably.
From the proof of Theorem~\ref{thm:jump} we immediately see that condition~(z) in the definition of zigzag language could be weakened as follows, and the same proof would still go through:
\begin{enumerate}[leftmargin=8mm, noitemsep, topsep=3pt plus 3pt]
\item[(z')] Given the sequence~$J(L_{n-1})$, then for every permutation~$\pi$ in this sequence there is a direction~$\chi(\pi)\in\{\leftarrow,\rightarrow\}$ satisfying the following:
For every $\pi$ from~$J(L_{n-1})$ we define $I(\pi):=\{i\in[n]\mid c_i(\pi)\in L_n\}$, and we also define $\icheck(\pi):=\max I(\pi)$ and $\ihat(\pi):=\min I(\pi)$ if $\chi(\pi)={}\leftarrow$, and $\icheck(\pi):=\min I(\pi)$ and $\ihat(\pi):=\max I(\pi)$ if $\chi(\pi)={}\rightarrow$.
Then for any two permutations $\pi,\rho$ that appear consecutively in~$J(L_{n-1})$ and that differ in a jump from position~$i_1$ to position~$i_2$, we have $\ihat(\pi)=\icheck(\rho)$ and this number is not in the interval $]\min\{i_1,i_2\},\max\{i_1,i_2\}]$.
\end{enumerate}
Note that condition~(z) implies~(z'), as for a zigzag language we have $1,n\in I(\pi)$ for all $\pi\in L_{n-1}$, so $\min I(\pi)=1$ and $\max I(\pi)=n$, and the sequence~$\chi$ alternates between~$\leftarrow$ and~$\rightarrow$ in every step.
This shows that Algorithm~J succeeds to generate a much larger class of languages $L_n\seq S_n$.
However, condition~(z') is considerably more complicated, in particular as it depends on the ordering~$J(L_{n-1})$ generated for~$L_{n-1}$.
More importantly, in the context of permutation patterns we do not see any interesting languages~$L_n$ that would satisfy condition~(z') but not condition~(z), which is why we think condition~(z) is the `correct' one.

Let us also briefly motivate the definition of tame permutation patterns given in Section~\ref{sec:tame}.
On the one hand, for a fixed pattern~$\tau$, it is certainly reasonable to require that \emph{all} of the sets $S_n(\tau)$, $n\geq 0$, are zigzag languages.
This is because for \emph{any} classical pattern~$\tau\in S_k$, the set $S_n(\tau)=S_n$ for $n<k$ is trivially a zigzag language, but a very uninteresting one.
On the other hand, now that we are concerned with an infinite sequence of sets $S_n(\tau)$, $n\geq 0$, the requirement for the sequence to be hereditary is equally reasonable, as we shall see by considering the problems that arise if we drop this requirement.
For this consider the non-tame barred pattern~$\tau=132\ol{4}$.
We have $132\notin S_3(\tau)$, i.e., the permutation~$132$ contains the pattern, whereas $c_4(132)=1324\in S_4(\tau)$ avoids it.
Consequently, if we were interested in the language $L_4=S_4(\tau)$, then the corresponding set $L_3:=p(L_4)$ would be $L_3=S_3$, and not $S_3(\tau)=S_3\setminus\{132\}$.
In general, when considering~$S_n(\tau)$ for a particular value of~$n$, then the sets $L_{i-1}:=p(L_i)$ for $i=n,n-1,\ldots,1$ are not necessarily characterized by avoiding the pattern~$\tau$.
In fact, it is not clear whether they are characterized by any kind of pattern-avoidance.
In terms of the tree representation of zigzag languages, if the sequence $S_n(\tau)$, $n\geq 0$, is not hereditary, then different values of~$n$ correspond to different tree prunings.
In the example from before, the node~$132$ is pruned from the tree for~$S_3(\tau)$, but it is not pruned from the tree for~$S_4(\tau)$.
In contrast to that, in a hereditary sequence, all sets~$S_n(\tau)$, $n\geq 0$, arise from pruning the infinite rooted tree of permutations in a way that is consistent for all~$n$.
Summarizing, for patterns~$\tau$ for which $S_n(\tau)$ is not hereditary, our proof of Theorem~\ref{thm:jump} breaks seriously.
Not only that, Algorithm~J in general fails to generate~$S_n(\tau)$.
For instance, it fails to generate $S_4(\tau)$, $\tau=132\ol{4}$, when initialized with $\ide_4=1234$, while it succeeds to generate~$S_4(\tau')$ for the non-tame barred pattern~$\tau'=\ol{4}132$ .

For a classical pattern~$\tau\in S_k$ that has the largest value~$k$ at the leftmost or rightmost position, we have that $S_k(\tau)=S_k\setminus\{\tau\}$ is not a zigzag language (as $\tau$ equals either $c_1(p(\tau))$ or $c_k(p(\tau))$), i.e., $\tau$ is not tame.
Moreover, in general Algorithm~J fails to generate~$S_k(\tau)$.
For instance, running Algorithm~J on $S_3(321)$ gives only three permutations $123,132,312$, and then the algorithm stops.
This is admittedly a very strong limitation of our approach, as many interesting permutation patterns have the largest value at the boundary, such as $321$, which gives rise to an important Catalan family~$S_n(321)$.

By what we said before, the condition for tameness stated in Lemma~\ref{lem:class} is not only sufficient, but also necessary.
In a similar way, it can be shown that the conditions stated in Lemmas~\ref{lem:vinc}--\ref{lem:bivinc} are necessary for tameness.
The situation is slightly more complicated for Theorem~\ref{thm:mesh}:
Conditions~(i) and~(ii) of the theorem are indeed necessary.
Specifically, if condition~(i) is violated, then $S_k(\sigma)=S_k\setminus\{\tau\}$ is not a zigzag language, and if condition~(ii) is violated, then $S_k(\sigma)\neq p(S_{k+1}(\sigma))$, i.e., the hereditary property is violated.
However, conditions~(iii) and~(iv) are not necessary.
Consider for instance the patterns~$\sigma_1,\sigma_2,\sigma_3$ shown below:
\begin{center}
\begin{tikzpicture}[scale=0.4]
\path [fill=gray!60] (1,5) rectangle (2,6);
\path [fill=gray!60] (2,3) rectangle (3,4);
\draw [->,>=stealth] (2.5,3.5) -- (1.5,3.5);
\mesh{5}{1,5,4,3,2}
\node[anchor=west] at (2.0,-0.6) {$\sigma_1$};
\end{tikzpicture} \hspace{1mm}
\begin{tikzpicture}[scale=0.4]
\path [fill=gray!60] (1,5) rectangle (2,6);
\path [fill=gray!60] (2,3) rectangle (3,4);
\draw [->,>=stealth] (2.5,3.5) -- (1.5,3.5);
\path [fill=gray!60] (3,2) rectangle (4,3);
\mesh{5}{1,5,4,3,2}
\node[anchor=west] at (2.0,-0.6) {$\sigma_2$};
\end{tikzpicture} \hspace{1mm}
\begin{tikzpicture}[scale=0.4]
\path [fill=gray!60] (1,5) rectangle (2,6);
\path [fill=gray!60] (2,3) rectangle (3,4);
\draw [->,>=stealth] (2.5,3.5) -- (1.5,3.5);
\path [fill=gray!60] (3,2) rectangle (4,3);
\path [fill=gray!60] (4,1) rectangle (5,2);
\mesh{5}{1,5,4,3,2}
\node[anchor=west] at (2.0,-0.6) {$\sigma_3$};
\end{tikzpicture}
\end{center}
They all satisfy conditions~(i), (ii) and~(iv), but violate condition~(iii) due to the cells~$(2,3)$ and~$(1,3)$, connected by an arrow in the figures.
However, the proof of Theorem~\ref{thm:mesh} given in Section~\ref{sec:mesh} can be modified to show that $\sigma_1$ is tame.
The idea is to apply the exchange argument that involves a point in a match of the pattern and that is illustrated on the right hand side of Figure~\ref{fig:mesh} twice instead of only once.
This idea can be iterated, and by applying the exchange argument three or four times, respectively, one can show that~$\sigma_2$ and~$\sigma_3$ are tame as well.
This kind of reasoning apparently leads to combinatorial chaos, depending on the relative location of points and shaded cells in the pattern, and this prevents us from being able to formulate conditions for a mesh pattern that are necessary and sufficient for tameness.
This is not an issue from our point of view, because again, we do not see any interesting families of pattern-avoiding permutations that would satisfy such more complicated conditions but not the conditions stated in Theorem~\ref{thm:mesh}.

\section{Acknowledgments}

We thank Michael Albert, Mathilde Bouvel, Sergi Elizalde, V{\'{i}}t Jel{\'{i}}nek, Sergey Kitaev, Vincent Vajnovszki, and Vincent Vatter for very insightful feedback on this work, for pointing out relevant references, and for sharing their knowledge about pattern-avoiding permutations.
We also thank Jean Cardinal, Vincent Pilaud, and Nathan Reading for several stimulating discussions about lattice congruences of the weak order on the symmetric group.
Lastly, we thank the anonymous reviewer who provided very thoughtful feedback and many suggestions that helped improving the manuscript.

\bibliographystyle{alpha}
\bibliography{refs}

\end{document}